\documentclass[ejs,preprint,noshowframe]{imsart}

\RequirePackage{amsthm,amsmath,amsfonts,amssymb,bm,bbm}
\RequirePackage[numbers]{natbib}
\RequirePackage[colorlinks,citecolor=blue,urlcolor=blue]{hyperref}
\RequirePackage{graphicx}
\usepackage{cleveref}
\usepackage{mathtools}
\usepackage{lscape}

\startlocaldefs
\theoremstyle{plain}

\newtheorem{theorem}{Theorem}[section]
\newtheorem{lemma}[theorem]{Lemma}
\newtheorem{corollary}[theorem]{Corollary}
\newtheorem{prop}[theorem]{Proposition}
\theoremstyle{definition}

\newtheorem{remark}{Remark} 
\newtheorem{assumption}{Assumption}
\theoremstyle{remark}

\allowdisplaybreaks

\DeclareMathOperator*{\argmin}{arg\,min}

\DeclareMathOperator*{\sign}{sign}

\DeclareMathOperator*{\var}{Var}
\DeclareMathOperator*{\cov}{Cov}
\DeclareMathOperator*{\card}{card}
\endlocaldefs

\begin{document}
\begin{frontmatter}
\title{On tail-robust autocovariance matrix estimation for high-dimensional  and potentially nonstationary time series}
\runtitle{Tail-robust autocovariance matrix estimation}

\begin{aug}
\author[A]{\fnms{Haotian}~\snm{Xu}\ead[label=e1]{aotian.xu@auburn.edu}},
\author[B]{\fnms{St\'ephane}~\snm{Guerrier}\ead[label=e2]{stephane.guerrier@unige.ch}},
\author[C]{\fnms{Runze}~\snm{Li}\ead[label=e3]{rzli@psu.edu}}
\and
\author[D]{\fnms{Yuan}~\snm{Ke}\ead[label=e4]{yuan.ke@uga.edu}}
\address[A]{Department of Mathematics and Statistics, Auburn University\printead[presep={,\ }]{e1}}

\address[B]{Faculty of Science \& Geneva School of Economics and management, University of Geneva\printead[presep={,\ }]{e2}}

\address[C]{Department of Statistics, Pennsylvania State University\printead[presep={,\ }]{e3}}

\address[D]{Department of Statistics, University of Georgia\printead[presep={,\ }]{e4}}
\runauthor{H. Xu et al.}
\end{aug}

\begin{abstract}
In this paper, we study the autocovariance matrix estimation and inference problems under heavy-tailedness, high-dimensionality, general nonlinear temporal dependence, and potentially nonstationarity of time series. We consider two types of tail-robust autocovariance matrix estimation methods: the element-wise Huber's $M$-estimator and a computationally more efficient element-wise truncated estimator. Both estimators are designed to achieve sharp error bounds in matrix max-norm. The nonasymptotic properties of these estimators are proved based on new variants of Bernstein-type inequalities under functional dependence for the potentially nonstationary processes which may be of independent interest. Moreover, we prove a high-dimensional Gaussian approximation result, as a limiting distribution, for our element-wise truncated autocovariance estimator. A Gaussian multiplier bootstrap result is also given to facilitate the practicality. Our theoretical results are nonasymptotic, which gives explicit error bounds in terms of the sample size, dimensionality, moments, and the strength of temporal dependence. Numerical evidence is provided to support our theoretical results. Finally, we illustrate the benefits of the proposed methodology for detecting change points in monthly macroeconomic data.
\end{abstract}

\begin{keyword}[class=MSC]
\kwd[Primary ]{62M10}
\kwd{00X00}
\kwd[; secondary ]{00X00}
\end{keyword}

\begin{keyword}
\kwd{Gaussian approximation}
\kwd{heavy-tail}
\kwd{$M$-estimation}
\kwd{time series}
\end{keyword}

\end{frontmatter}

\section{Introduction}\label{sec:introduction}

High-dimensional time series are frequently encountered in modern data science applications. Indeed, it is common for multivariate data to be serially generated at many locations, whose number $d$ may be large compared to the number of time points $n$. Autocovariance matrices of time series play fundamental roles in many statistical and machine learning methods \cite[see e.g.,][]{hyvarinen2000independent,han2015direct,jentsch2015covariance,mcmurry2015high,chang2017testing,li2019testing}.
High-dimensionality poses significant challenges to autocovariance matrix estimation. For example, the sample autocovariance matrix is always singular when $d > n$ \cite[see e.g.,][]{wu2012covariance}. From the random matrix theory perspective, the eigenvalues of a sample covariance matrix do not converge to those of the population covariance matrix, when $d$ diverges no slower than $n$ \cite[see e.g.,][]{liu2015marvcenko}. To address these issues of autocovariance matrix estimation, several sparsity assumptions on population autocovariance matrices and their corresponding regularized estimators have been proposed. Taking advantage of their sharp error bounds, these estimators are suitable to use in high-dimensional settings \cite[see e.g.,][among others]{bickel2008covariance,bickel2008regularized,el2008operator,cai2010optimal,cai2011adaptive,ravikumar2011high}. 
To obtain optimal statistical guarantees for these regularized estimators, sub-Gaussianity on tail behavior is crucial.  However, it is a rather strong assumption and is often violated in practice \cite[see e.g.,][]{kuchibhotla2018moving}, which explains the poor performance of most classical estimators, both theoretically and numerically. Constructing a robust estimator against heavy-tailedness is a way to address this challenge. Many tail-robust mean estimators have been proposed and studied, including the Huber's $M$-estimator \cite[see e.g.,][]{catoni2012challenging} and the median-of-means estimator \cite[see e.g.,][]{nemirovsky1983problem, lerasle2011robust}. A comprehensive review of tail-robust mean estimators can be found in \cite{lugosi2019mean}. Avella-Medina et al. \cite{avella2018robust} proposed and studied tail-robust covariance matrix estimators. %
The issue of heavy-tailedness can also be solved from a spectrum domain perspective. Following the general robustification technique developed in Minsker \cite{minsker2018sub}, Ke et al. \cite{ke2019user} explored the spectrum-wise truncated estimator of covariance matrix which was constructed by truncating observed random matrices in their spectrum domain. Ke et al.\cite{ke2019user} also provided a selected survey of recent developments in covariance estimation for heavy-tailed distributions. However, the theoretical results of the above methods require the data are independent and identically distributed (iid).

In addition to sub-Gaussianity, temporal dependence significantly influences the estimation of autocovariance matrices. These effects have been explored by Chen et al. \cite{chen2013covariance} in high-dimensional settings for covariance matrix estimation, where they focus on the thresholding estimation for the structured covariance matrix of stationary time series. A similar study has been given by Shu and Nan \cite{shu2019estimation}, but under broad scenarios allowing for high-dimensional observations with sub-Gaussianity or heavy-tailedness and with short-range or long-range dependence. However, the data generating mechanism they considered is limited to ``linear spatiotemporal models'', and the dependence assumptions are able to be imposed directly on crosscorrelations. Under finite polynomial moments, their results only allow $d$ to grow with $n$ in a polynomial rate. Recently, Zhang \cite{zhang2018} showed that the polynomial rate can be improved to an exponential rate by using a tail-robust covariance matrix estimator based on the Huber's $M$-estimator. 
The majority of the existing methods for analyzing dependent data rely on the stationarity assumption of the data-generating process, which, to some extent, balances the rigors of statistical analysis and the broadness of their applicability. However, assuming stationarity over a long time range can be unrealistic as the probability properties of the observed time series may be subject to change due to trend, seasonality, time-varying dynamics, and change points. %
Researchers who need to estimate autocovariance matrices may have no alternative but to ignore the nonstationarity in their time series data as they do not have adequate tools to handle such problems. Unfortunately, their findings can be statistically misleading when the stationary assumption is violated. To overcome the aforementioned challenges, this paper studies, from a nonasymptotic perspective, the autocovariance matrix estimation and inference problems under heavy-tailedness, high-dimensionality, general nonlinear temporal dependence, and potentially nonstationarity. Technical tools developed in this paper, i.e.~Bernstein's inequalities, the Gaussian approximation, and the Gaussian multiplier bootstrap results, are compatible with nonstationary and piece-wise stationarity time series.

The rest of the paper is organized as follows. Section \ref{sec:preliminary} formalizes the considered problem and introduces definitions regarding temporal dependence and tail-robustness. Several Bernstein's inequalities under temporal dependence and potentially nonstationarity are also presented in this section. Section \ref{sec:robust} presents our robust autocovariance matrix estimators and their nonasymptotic properties. In Section \ref{sec:Gaussian_autocov}, we provide a Gaussian approximation result for our autocorrlation matrix estimation, and a Gaussian multiplier bootstrap method is also given to facilitate the practicality.
Section \ref{sec:simulation} gives extensive numerical experiments justifying the finite sample performance of our methods. Finally, Section \ref{sec:conclusion} concludes. 
The proofs of main theoretical results and additional simulation results are relegated to the supplementary material.

\subsection{Notation}\label{sec:notation}
Let $\mathbb{Z}$, $\mathbb{Z}^+$ and $\mathbb{R}$ denote the set of integers, positive integers and real numbers, respectively. For a set $A$, $\card(A)$ denotes the cardinality of $A$.  For $n_1, n_2 \in \mathbb{Z}^+$ and $N_1 < n_2$, we denote $[n_1, n_2] = \{n_1, n_1+1, \dots, n_2\}$ and $[n_1] = \{1, \dots, n_1\}$. The superscript $^{\intercal}$ denotes the transpose for a matrix or a vector. Given a vector $\bm{x} = (x_1, \dots, x_d)^{\intercal} \in \mathbb{R}^d$, we write the vector $\ell_q$-norm as $|\bm{x}|_q = \big(\sum_{j=1}^d |x_j|^q\big)^{1/q}$ for $1 \leq q < \infty$ and the vector $\ell_{\infty}$-norm as $|\bm{x}|_{\infty} = \max_{j \in [d]}|x_j|$.
Given a matrix $\bm{A} = \big( A_{(jk)} \big)_{j \in [d_1]; k \in [d_2]} \in \mathbb{R}^{d_1 \times d_2}$.
The max-norm of $\bm{A}$ is denoted as $\Vert \bm{A} \Vert_{\max} = \max_{j,k}|A_{(jk)}|$.
For a sequence of matrices $\{\bm{A}_i\}_{i \in S}$ with $S \subseteq \mathbb{Z}$, we write $A_{i, (jk)}$ as the $(j,k)$-th entry of $\bm{A}_i$. We write $\bm{A}_1  \prec \bm{A}_2$ if $(\bm{A}_2-\bm{A}_1)$ is a positive definite matrix. $\bm{I}_d$ denotes the $d$-dimensional identity matrix. For an $\mathbb{R}$-valued random variable $X$, we write the $L_q$-norm of $X$ as $\Vert X \Vert_{q} = \big(\mathbb{E}|X|^q\big)^{1/q}$, for $q > 0$. For $a, b \in \mathbb{R}$, we denote by $\sign(a)$ the sign of $a$, and denote $a \wedge b \coloneqq \min(a, b)$ and $a \vee b \coloneqq \max(a, b)$. For two positive values $a$ and $b$, we write $a \asymp b$ (resp. $a \lesssim b$ or $a = \mathcal{O}(b)$) if there exist absolute constants $C_2 \geq C_1 > 0$ such that $C_1 \leq a/b \leq C_2$ (resp. $a/b \leq C_1$).
Absolute constants are denoted as $C, C_1, C_2, \dots > 0$, which may be different in each place.

\section{Preliminary}\label{sec:preliminary}

Let $\{\bm{X}_i\}_{i \in \mathbb{Z}} \subset \mathbb{R}^d$ be a sequence of random vectors such that
    \begin{equation}\label{eq:functional_wold}
        \bm{X}_i = (X_{i,1}, X_{i,2}, \dots, X_{i,d})^{\intercal} = G_i(\mathcal{F}_i),
    \end{equation}
where $G_i(\cdot) = \big(g_{i,1}(\cdot), g_{i,2}(\cdot), \dots, g_{i,d}(\cdot)\big)^{\intercal}$ is an $\mathbb{R}^d$-valued measurable function, and $\mathcal{F}_i = \sigma(\dots, \, \epsilon_{i-1}, \, \epsilon_i)$ is a natural filtration of $\{\epsilon_i\}_{i \in \mathbb{Z}}$, and $\{\epsilon_i\}_{i \in \mathbb{Z}}$ is a sequence of iid random variables. We allow $G_i(\cdot)$ to be time-dependent, and thus $\{\bm{X}_i\}_{i \in \mathbb{Z}}$ is not required to be stationary. %
In spite of the causal nature of \eqref{eq:functional_wold}, this representation includes a large class of linear and nonlinear time series models, such as linear processes, autoregressive moving average (ARMA) models, generalized autoregressive conditional heteroskedasticity (GARCH) models \cite[see e.g.,][]{wu2005nonlinear,wu2011asymptotic} and their nonstationary variants \cite[see e.g.,][]{zhou2009local}. %
Wu \cite{wu2005nonlinear} considered the representation \eqref{eq:functional_wold} and developed the functional dependence measure (detailed in \Cref{sec:dependence}). Throughout this paper, we use functional dependence measure to quantify temporal dependence of $\{\bm{X}_i\}_{i \in \mathbb{Z}}$ and its measurable transformations. In \Cref{sec:dist_rob}, we introduce the truncation argument to deal with the heavy-tailedness. In \Cref{sec:bern}, we review Bernstein's inequalities for dependent processes, which will be used to develop our nonasymptotic results.

\subsection{Funcational dependence measure}\label{sec:dependence}
The functional dependence measure \citep{wu2005nonlinear} is one of the commonly used tools for quantifying general nonlinear temporal dependence. For any time indices $i, s_1, s_2 \in \mathbb{Z}$ such that $s_2 \leq s_1$, define random vectors
\begin{align}\label{eq:couple.ver}
\bm{X}_{i, \{s_1, s_2\}} = G_i(\mathcal{F}_{i,\{s_1, s_2\}}),
\end{align}
where the filtration
\begin{align*}
    \mathcal{F}_{i,\{s_1, s_2\}} = 
    \begin{cases}
        \sigma(\dots, \epsilon_{s_2-1}, \epsilon_{s_2}^*, \dots, \epsilon_{s_1}^*, \epsilon_{s_1+1}, \dots, \epsilon_i),  \quad &s_1 < i,\\
        \sigma(\dots, \epsilon_{s_2-1}, \epsilon_{s_2}^*, \dots, \epsilon_i^*),  \quad &s_2 \leq i \leq s_1,\\
        \sigma(\dots, \epsilon_{i}),  \quad &i < s_2,\\
    \end{cases}
\end{align*}
and $\epsilon_i^*$ is an independent copy of $\epsilon_i$. For notational convenience, we write $ \bm{X}_{i,\{s\}} = \bm{X}_{i, \{s,s\}} $. Recalling the representation \eqref{eq:functional_wold}, $\bm{X}_{i, \{s_1, s_2\}}$ is a coupled version of $\bm{X}_i$, whose innovations at time interval $[s_1, s_2]$ are replaced by their independent copies.

The functional dependence measure quantifies dependence in terms of moments. If there exists some integer $q > 0$ such that 
\[
\max_{1 \leq j \leq d}\sup_{i \in \mathbb{Z}} \Vert X_{i,j} \Vert_q < \infty,
\] 
we define the functional dependence measure of order $q$ and its tail cumulative version respectively as
\begin{equation}\label{FDM}
        \delta_{s,q} = \max_{1 \leq j \leq d}\delta_{s,q,j} = \max_{1 \leq j \leq d}\sup_{i \in \mathbb{Z}}\Vert X_{i,j} - X_{i,j,\{i-s\}} \Vert_{q} \quad \text{and} \quad \Delta_{m, q} = \sum_{s = m}^{\infty} \delta_{s,q}.
    \end{equation}
We consider, throughout this paper, the time series with a short-range dependence, i.e.~$\Delta_{0,q} < \infty$. In addition to the short-range dependence, we require that $\delta_{s,q}$ or equivalently $\Delta_{s,q}$ decays exponentially to $0$ as the time lag $s$ increases. To be more precise, we assume that there exists some constant $c > 0$, such that
\begin{equation}\label{eq:GMC}
        \Vert \bm{X}_{\cdot} \Vert_{q} = \sup_{m \geq 0} \exp(cm) \Delta_{m,q} < \infty.
    \end{equation}
Condition \eqref{eq:GMC} is equivalent to saying that the functional dependence measure decays exponentially. 
The following lemma formalizes this claim.
\begin{lemma}\label{lemma:dep.condition.relation}
    Let $\{\bm{X}_i\}_{i \in \mathbb{Z}} \subset \mathbb{R}^d$ be a process satisfying\eqref{eq:functional_wold}.
    For some $c > 0$, we have $\Vert \bm{X}_{\cdot} \Vert_{q} < \infty$ is equivalent to $\delta_{s,q} = \mathcal{O}(\exp(-cs))$ for all $s \geq 0$.
\end{lemma}
The decay rate of the functional dependence measures reflects the strength of temporal dependence. Indeed, similar measurements can be realized by various types of mixing coefficients \cite[see e.g.,][]{bradley2005basic, dedecker2007weak}. We highlight that all of our results based on the functional dependence measure can also be obtained without technical difficulties under the corresponding conditions of $\alpha$-mixing coefficients. Below, we briefly state some advantages of using the functional dependence measure. First, verifying the conditions based on the functional dependence measure is relatively easy, especially for some complex nonlinear processes. Second, given a process of the form \eqref{eq:functional_wold}, one can construct its martingale approximation \cite[see e.g.,][]{wu2007strong} or $m$-dependence approximation \cite[see e.g.,][]{berkes2009asymptotic, liu2013probability}, by the coupling technique \cite[see e.g.,][]{berkes2009asymptotic}. Consequently, approximation errors can be quantified by the functional dependence measure, and well-established theories for martingale difference sequences or iid sequences can be sought. Third, mixing coefficients of multivariate data depend on their dimensionality. Conditions associated with mixing coefficients must be verified case by case when the dimension $d$ tends to infinity. Although the functional dependence measure is also dimension-dependent, as a coupling-based quantity, it can be explicitly computed even in high-dimensional settings. See Han and Wu \cite{han2019probability} for detailed discussion on this issue.
Although both the functional dependence measure and the mixing coefficients can quantify the strength of temporal dependence, the conditions based on the functional dependence measure and the strong mixing coefficients do not imply each other. The former relies on the causal representation \eqref{eq:functional_wold}, while the latter is more in a nonparametric spirit.

\subsection{Tail-robustness and the effects of truncation}\label{sec:dist_rob}
We also allow for heavy-tailedness, meaning that the marginal distribution of time series has finite polynomial moments up to an order $q > 0$.
Our goal is to obtain a tail-robust estimator. Namely, the estimator possesses exponential-type error bounds even when the observations are heavy-tailed \cite[see e.g.,][]{catoni2012challenging,devroye2016sub}.
It is essential to mention that the tail-robustness is different from the classical infinitesimal robustness \cite[see e.g.,][]{huber2009robust, hampel2011robust}. The latter focuses on consistent and efficient estimation of parametric models when the data-generating process lies in a neighborhood of the assumed parametric model, defined as small and arbitrary model deviations, such as gross error contamination, but not structural ones, such as heavy-tails.

Since Huber's $M$-estimator has been used to achieve both these robustness goals, we illustrate the differences by the following setting of univariate mean estimation. Let $\{X_{i}\}_{i \in \mathbb{Z}} \subset \mathbb{R}$ be a stationary process with an unknown mean $\mu = \mathbb{E}[X_1]$. For any $u \in \mathbb{R}$, we define the truncation operator as
    \begin{equation}\label{eq:trun.operator}
        \psi_{\tau}(u) = \sign(u)(|u| \wedge \tau),
    \end{equation}
where $\tau > 0$ is the robustification parameter. Notice that $\psi_{\tau}(u)$ is also the first order derivative of the Huber loss function $L_{\tau}(u)$ \cite[See Huber][]{huber1984finite}, which is defined for any $u \in \mathbb{R}$ as
    \begin{equation}
    \label{eq:Huber.loss}
        L_{\tau}(u) = \begin{cases}
    u^2/2, &  \;\; \text{if} \;\; |u| \leq \tau,  \\
   \tau |u| - \tau^2 /2     &  \;\; \mbox{if} \;\; |u| > \tau.
    \end{cases}
    \end{equation}
The Huber's $M$-estimator $\widehat{\mu}$ is constructed by solving
\begin{equation*}
        \frac{1}{n}\sum_{i=1}^n\psi_{\tau}(X_{i} - \widehat{\mu}) = 0.
    \end{equation*}
From the infinitesimal robustness point of view, $\tau$ should be fixed so that the influence function \cite[see e.g.,][]{hampel2011robust}, which is proportional to $\psi_{\tau}(\cdot)$, is bounded. The classical Huber's $M$-estimation also assumes the marginal distribution of $X_{i}$ is symmetric. If this is not the case, the bias due to truncation $\mathbb{E}[\psi_{\tau}(X_{i}-\mu)] \neq 0$ and a bias correction term corresponding to $\tau$ is needed. 
In contrast, to achieve the tail-robustness, we require $\tau$ to diverge with the sample size $n$, and the symmetry of the distribution of $X_{i}$ is unnecessary. For any $\tau > 0$ and any fixed $u \in \mathbb{R}$, the process $\{\psi_{\tau}(X_{i} - u)\}_{i \in \mathbb{Z}}$ is a uniformly bounded approximation of the original process $\{X_i - u\}_{i \in \mathbb{Z}}$. The bias due to truncation vanishes as $\tau$ diverges, which is detailed in the next lemma.
\begin{lemma}\label{lemma:bias}
    Let $\{X_{i}\}_{i \in \mathbb{Z}} \subset \mathbb{R}$ be a process in the form of \eqref{eq:functional_wold} with $d=1$. Assume that $\sup_{i \in \mathbb{Z}}\mathbb{E}[X_i^2] < \infty$, then for any $\tau > 0$ and any fixed $u \in \mathbb{R}$, we have that
    \[
    \sup_{i \in \mathbb{Z}}\left|\mathbb{E}[\psi_{\tau}(X_{i}-u)] - \mathbb{E}[X_i - u]\right| \leq \frac{\sup_{i \in \mathbb{Z}}\mathbb{E}[(X_{i}-u)^2]}{\tau}.
    \] 
\end{lemma}

Apart from bounding the bias, it is needed to uniformly bound the deviation of $n^{-1}\sum_{i=1}^n\psi_{\tau}(X_{i} - u)$ for any $u \in \mathbb{R}$. The tool we will be using is Bernstein's inequality, which is detailed in the next subsection. We show that Huber’s $M$-estimator $\widehat{\mu}$ achieves an exponential-type error bound, i.e.~tail-robustness, by using a properly chosen $\tau$ balancing the bias and the deviation. Besides the tail-robustness, Lemma~\ref{lemma:higher.moment} below shows that the truncation operator reduces the temporal dependence in terms of the functional dependence measure.
\begin{lemma}\label{lemma:truc_FDM}
    \label{lemma:higher.moment}
    Let $\{X_{i}\}_{i \in \mathbb{Z}} \subset \mathbb{R}$ be a process in the form of \eqref{eq:functional_wold} with $d=1$. For some $q > 0$, denote $\delta_{s,q}^X$ and $\delta_{s,q}^{\mathrm{tru}}$ respectively the $q$-th order functional dependence measure of $\{X_{i}\}_{i \in \mathbb{Z}}$ and of the truncated process $\{\psi_{\tau}(X_{i} - u)\}_{i \in \mathbb{Z}} \subset \mathbb{R}$ for any fixed $u \in \mathbb{R}$. Then for any $\tau > 0$, we have that
    $$\delta_{s,q}^{\mathrm{tru}} = \sup_{u \in \mathbb{R}}\sup_{i \in \mathbb{Z}}\big\Vert \psi_{\tau}(X_{i}-u) - \psi_{\tau}(X_{i,\{i-s\}}-u) \big\Vert_{q} \leq \delta_{s,q}^X.$$   
\end{lemma}

\subsection{Bernstein's inequality under the functional dependence}\label{sec:bern}
Bernstein's inequality provides an exponential-type tail probability bound for partial sums of random variables, and the tightness of the bound depends on both the boundness or the sub-exponential parameter \cite[see e.g.,][]{vershynin2018high} and the variance of summand. The most well-known version of Bernstein's inequality concerns partial sums of independent (not necessarily identically distributed) random variables, which are either bounded \cite[see e.g.,][]{bernstein1946prob} or sub-exponential \cite[see e.g., Section 2.1.3 in][]{wainwright2019high}. In this paper, we focus on the former type of Bernstein's inequality, since the truncation operator leads to a sequence of random variables bounded by $\tau \in (0, \infty)$. To be more specific, let $\{X_i\}_{i = 1}^n \subset \mathbb{R}$ be a sequence of random variables. By applying the truncation operator $\psi_{\tau}(\cdot)$, we have that $\{\psi_{\tau}(X_i)\}_{i = 1}^n$ is uniformly bounded, i.e.~$\sup_{1 \leq i \leq n}|\psi_{\tau}(X_i)| \leq \tau$. Note that without loss  of generality and for brevity, we consider the location shift parameter $u = 0$. Below, we review the existing Bernstein's inequalities under different dependence assumptions.
\begin{itemize}
    \item When $\{X_i\}_{i = 1}^n$ are independent, Bernstein's inequality for independent random variables states that for any $x > 0$
\begin{align}\label{eq:bernstein_iid}
&\mathbb{P}\bigg(\bigg|\sum_{i = 1}^n\psi_{\tau}(X_i) - \sum_{i = 1}^n\mathbb{E}[\psi_{\tau}(X_i)]\bigg| \geq x\bigg)\\
\leq& 2\exp\Big(-\frac{x^2}{2\sum_{i = 1}^n\mathbb{E}\big[\big(\psi_{\tau}(X_i)\big)^2\big] + 2\tau x/3}\Big) \nonumber\\
\leq& 2\exp\Big(-\frac{x^2}{2\sum_{i = 1}^n\mathbb{E}[X_i^2] + 2\tau x/3}\Big).
\end{align}
\item When $\{X_i\}_{i = 1}^n$ are $\alpha$-mixing with mixing coefficients decay exponentially, i.e.~$\alpha(\ell) \leq \exp(-c\ell)$ for some $c > 0$, Merlev{\`e}de et al. \cite{merlevede2009bernstein} proves that for $n \geq 2$ and any $x > 0$
\begin{align}\label{eq:bernstein_alpha}
\mathbb{P}\bigg(\bigg|\sum_{i = 1}^n\psi_{\tau}(X_i) - \sum_{i = 1}^n\mathbb{E}[\psi_{\tau}(X_i)]\bigg| \geq x\bigg)
\leq& 2\exp\Big(-\frac{Cx^2}{nC_{\mathrm{LRV}} + \tau^2 + \tau x(\log n)^2}\Big),
\end{align}
where $C > 0$ is an absolute constant. Besides, $C_{\mathrm{LRV}}$ represents the long-run variance defined in \eqref{eq:def_LRV_UB}, and $C_{\mathrm{LRV}} \in (0, \infty)$ is guaranteed under the exponential decay of $\alpha$-mixing coefficients. This result matches \eqref{eq:bernstein_iid} up to a $\log n$ factor in the sub-exponential tail part. Under temporal dependence, the main challenge of proving a Bernstein-type inequality is on how to bound the exponential moment of partial sums, which can no longer be factored out into products of marginal exponential moments. Merlev{\`e}de et al. \cite{merlevede2009bernstein} addressed this issue by proposing a recursive block technique, which divides time series into Cantor-like blocks. Inside each block, the exponential moment of partial sums is bounded using the boundedness of random variables. To combine the exponential moments of all blocks, a sequence of mutually independent blocks is created to approximate the original blocks by the decoupling lemma.
\item When $\{X_i\}_{i = 1}^n$ are stationary and satisfy the exponential decay of the functional dependence measure. Using the same block technique, Zhang \cite{zhang2018} in Theorem 2.1 therein proved the same Bernstein's inequality as \eqref{eq:bernstein_alpha} for stationary processes under functional dependence.
\end{itemize}

To make \eqref{eq:bernstein_alpha} useful, under the temporal dependent settings, we need to ensure the long-run variance is finite. The next lemma gives an upper bound on the long-run variance of a potentially nonstationary process with the functional dependence measure that decays exponentially.
\begin{lemma}\label{lemma:lrv_univar}
    Let $\{X_i\}_{i \in \mathbb{Z}} \subset \mathbb{R}$ be a centered process in the form of \eqref{eq:functional_wold} with $d = 1$. Assume $\sup_{i \in \mathbb{Z}}\Vert X_i \Vert_2 < \infty$ and there exist some absolute constants $c, \gamma_1 > 0$ such that
    \begin{equation*}
        \|X_{\cdot}\|_2 = \sup_{m \geq 0}\exp(cm^{\gamma_1})\Delta_{m,2} < \infty.
    \end{equation*}
    Then, the long-run variance of $\{X_i\}_{i \in \mathbb{Z}}$ satisfies that
    \begin{equation}\label{eq:def_LRV_UB}
        \lim_{n \to \infty}\var\Big(\frac{1}{\sqrt{n}}\sum_{i = 1}^n X_i\Big) \leq \sum_{\ell = -\infty}^{\infty}\sup_{i \in \mathbb{Z}}\big|\cov(X_{i}, X_{i+\ell})\big| = C_{\mathrm{LRV}} < \infty.
    \end{equation}
\end{lemma}

The next theorem further generalizes Theorem 2.1 of Zhang \cite{zhang2018} by allowing nonstationarity. Its proof is given in Section~I in the supplementary material, which follows and extends the proofs of Merlev{\`e}de et al. \cite{merlevede2009bernstein} and Zhang \cite{zhang2018}, with some necessary modifications.
\begin{theorem}\label{thm:bernstein_bounded}
    Let $\{X_i\}_{i \in \mathbb{Z}} \subset \mathbb{R}$ be a centered process in the form of \eqref{eq:functional_wold} with $d = 1$. Assume there exist absolute constants $c > 0$ such that
    \begin{equation}\label{eq:fdm_exp_decay}
        \|X_{\cdot}\|_2 = \sup_{m \geq 0}\exp(cm)\Delta_{m,2} < \infty.
    \end{equation}
    For any $\tau > 0$, $n \geq 2$ and $x > 0$, we have that
    \begin{equation*}
    \begin{aligned}
        \mathbb{P}\bigg(\bigg|\sum_{i = 1}^n \psi_{\tau}(X_i) - \sum_{i=1}^n\mathbb{E}[\psi_{\tau}(X_i)]\bigg| \geq x\bigg) \leq 2\exp\Big(-\frac{Cx^2}{nC_{\mathrm{LRV}} + \tau^2 + \tau x(\log n)^2}\Big),
    \end{aligned}
    \end{equation*}
    where $C > 0$ is an absolute constants depending only on $c$ and  $\|X_{\cdot}\|_2$ defined in  Lemma \ref{lemma:lrv_univar}.
\end{theorem}

 \Cref{thm:bernstein_bounded} provides the same concentration bound as in \eqref{eq:bernstein_alpha} up to an absolute constant, but based on the functional dependence measure. Similarly, \Cref{thm:bernstein_bounded} contains an extra $\log n$ factor in the sub-exponential tail part compared to \eqref{eq:bernstein_iid}. We note that the $\log n$ factor arises due to the use of the block technique proposed by Merlev{\`e}de et al. \cite{merlevede2009bernstein} for handling the Laplace transform of partial sums of bounded dependent random variables. In a special temporal dependence setting, when $\{X_i\}_{i \in \mathbb{Z}}$ is a linear process, the $\log n$ factor can be removed by using a different prove technique based on martingale difference sequences. Unfortunately, this proof heavily relies on linearity, and extending it to more general nonlinear processes appears difficult. A Bernstein-type inequality for potentially nonstationary linear processes is given in the next theorem. Its proof is also provided in Section~I in the supplementary material. By the definition of the functional dependence measure \eqref{FDM}, the exponential decay of coefficients of a linear process, i.e.~\eqref{eq:linear_coef_exp_decay}, implies the exponential decay of the functional dependence measure, i.e.~\eqref{eq:fdm_exp_decay}.

\begin{theorem}\label{thm:bernstein_bound_expdec_new}
    Let $\{X_i\}_{i\in\mathbb{Z}} \subset \mathbb{R}$ be a linear process given by
    \begin{equation}\label{eq:linear_process}
        X_i = \sum_{\ell = 0}^{\infty}a_{i}(\ell)\epsilon_{i-j},
    \end{equation}
    where $\{\epsilon_i\}_{i \in \mathbb{Z}}$ is a sequence of iid random variables with mean zero and $\mathbb{E}[\epsilon_i^2] = \sigma_{\epsilon}^2 < \infty$, and $\{a_i(\ell)\}_{i \in \mathbb{Z}, \ell \geq 0}$ is a time-dependent deterministic sequence satisfying that
    \begin{equation}\label{eq:linear_coef_exp_decay}
        \sup_{i \in \mathbb{Z}}a_i(\ell) \leq C_{Lin}\exp(-c\ell), \;\; \text{for any} \;\; \ell \geq 0,
    \end{equation}
    where $c, C_{Lin} > 0$ are some absolute constants. For any $\tau > 0$ and $x > 0$, we have that
\begin{equation*}
\begin{aligned}
    &\mathbb{P}\bigg( \sum_{i = 1}^n\big\{\psi_{\tau}(X_i) - \mathbb{E}[\psi_{\tau}(X_i)]\big\} \geq x \Big) \leq \exp\bigg\{-\frac{Cx^2}{nC_{\mathrm{LRV}} + \tau x}\bigg\},
\end{aligned}
\end{equation*}
where $C > 0$ are some absolute constants depending only on $c$ and  $C_{\mathrm{Lin}}$, and $C_{\mathrm{LRV}}$ is defined in \eqref{eq:def_LRV_UB}.
\end{theorem}

\section{Robust autocovariance matrix estimation}\label{sec:robust}

Let $\{\bm{X}_i \}_{i \in \mathbb{Z}} \subset \mathbb{R}^d$ be a constant mean potentially nonstationary process in the form of \eqref{eq:functional_wold}. For any integer $\ell \in [-n+1, n-1]$, denote the lag-$\ell$ autocovariance matrix of $\{\bm{X}_i\}_{i \in \mathbb{Z}}$ as
\begin{align*}
\bm{\Sigma}_{\ell} =& \frac{1}{n-\ell}\sum_{i=1+\ell}^{n}\mathbb{E}\big[(\bm{X}_{i-\ell} - \bm{\mu})(\bm{X}_{i} - \bm{\mu})^{\intercal}\big]\\
=& \frac{1}{n-\ell}\sum_{i=1+\ell}^{n}\mathbb{E}[\bm{X}_{i-\ell}\bm{X}_{i}^{\intercal}] - \bm{\mu}\bm{\mu}^{\intercal} = \big(\gamma_{\ell,(jk)}\big)_{1 \leq j,k \leq d},
\end{align*}
where $\bm{\mu} = \mathbb{E}[\bm{X}_i] = (\mu_1, \mu_2, \dots, \mu_d)^{\intercal}$. For any integers $j, k \in [d]$, we can write
\begin{equation}\label{eq:autocov}
    \gamma_{\ell,(jk)} = \frac{1}{n-\ell}\sum_{i=1+\ell}^{n}\mathbb{E}[X_{i-\ell,j}X_{i,k}] - \mu_j\mu_k.
\end{equation}
Due to the fact that $\bm{\Sigma}_{\ell} = \bm{\Sigma}_{-\ell}^{\intercal}$, we only consider $\bm{\Sigma}_{\ell}$ with $\ell \in [0, n-1]$ throughout Section~\ref{sec:robust}.
According to \eqref{eq:autocov}, estimating $\frac{1}{n-\ell}\sum_{i=1+\ell}^{n}\mathbb{E}[X_{i-\ell,j}X_{i,k}]$ and $\mu_j$ can be treated separately with the same type of tail-robust estimators. A discussion on tail-robust mean estimation methods and their nonasymptotic properties is presented in Section~C in the supplementary material. For notational convenience, we denote the lag-$\ell$ outer products as
\begin{equation}\label{eq:outer_prod}
    \bm{H}_{i,\ell} = \bm{X}_{i-\ell}\bm{X}_{i}^{\intercal}, \; \text{ for } i \in \mathbb{Z},
\end{equation}
and the lag-$\ell$ cross product for $(j,k)$-th coordinate is denoted as $H_{i,\ell,(jk)} = X_{i-\ell,j}X_{i,k}$.

As described in Section \ref{sec:introduction}, truncation can be used to address heavy-tailedness, especially in estimating high-dimensional mean vectors and covariance matrices. However, their nonasymptotic properties under temporal dependence and potentially nonstationarity are still lacking, so our results attempt to fill this gap. To be more specific, we consider a constant mean and  potentially nonstationary process $\{\bm{X}_i\}_{i \in \mathbb{Z}}$ satisfying the following assumptions.
\begin{assumption}\label{assumption:coor_4-moment}
    Let $\omega_4 = \sup_{i \in \mathbb Z}\max_{j \in [d]} \Vert X_{i,j} \Vert_4 < \infty$.
\end{assumption}
\begin{assumption}\label{assumption:coor_dep}
    There exists some constant $c > 0$ such that $$\Vert \bm{X}_{\cdot} \Vert_4  = \sup_{m \geq 0}\exp(cm)\sum_{s = m}^{\infty}\delta_{s,4} < \infty.$$
\end{assumption}
\noindent Assumption \ref{assumption:coor_4-moment} requires finite coordinate-wise moments up to the fourth order for all coordinates. This is a necessary condition for obtaining a Sub-Gaussian type estimator for covariances. See e.g.~Theorem 3.1 of Devroye et al. \cite{devroye2016sub} for the minimax lower bound of the mean estimator under finite ($1+\epsilon$) moment constraint with $\epsilon \in (0,1)$. Assumption \ref{assumption:coor_dep} relies on \Cref{assumption:coor_4-moment} and further imposes exponential decay of dependence measure for all coordinates, which is required by Bernstein's inequalities: \Cref{thm:bernstein_bounded} for general nonlinear processes and \Cref{thm:bernstein_bound_expdec_new} for linear processes. The exponential decay of the functional dependence measure also implies the exponential decay of autocovariances, thus it suggests us to only estimate $\bm{\Sigma}_{\ell}$ up to certain lags such that $|\ell| \leq \lfloor (2c)^{-1}\log n \rfloor$, and set the rest to be $\bm{0}_{d \times d}$. This claim is formalized in the next lemma.
\begin{lemma}\label{lemma:autocov_decay}
    Under Assumptions \ref{assumption:coor_4-moment} and \ref{assumption:coor_dep}, we have that
    \begin{align*}
        \|\bm{\Sigma}_{\ell}\|_{\max} \leq \|\bm{X}_{\cdot}\|_2^2\exp(-c\ell),
    \end{align*}
    where $\|\bm{X}_{\cdot}\|_2 = \sup_{m \geq 0}\exp(cm)\sum_{s = m}^{\infty}\delta_{s,2} \leq \|\bm{X}_{\cdot}\|_4$.\\ Moreover, for $|\ell| > \lfloor (2c)^{-1}\log n \rfloor$, we have that
    \begin{align*}
        \|\bm{\Sigma}_{\ell}\|_{\max} \leq \frac{\|\bm{X}_{\cdot}\|_2^2}{\sqrt{n}}.
    \end{align*}
\end{lemma}
The constant $c$ in \Cref{assumption:coor_dep} is generally unknown. In practice, we consider the lags such that $|\ell| \leq \lfloor C_{\mathrm{lag}}\log n \rfloor$ for some sufficiently large absolute constant $C_{\mathrm{lag}} > 0$.
Next, we study two types of element-wise tail-robust estimation methods for high-dimensional autocovariance matrices, which are (a) the element-wise Huber's $M$-estimator $\widetilde{\bm{\Sigma}}_{\ell}$ and (b) the element-wise truncated estimator $\widehat{\bm{\Sigma}}_{\ell}$. Under Assumptions \ref{assumption:coor_4-moment} and \ref{assumption:coor_dep}, we show that the error bounds of these two estimators are optimal (up to an $\log n$ factor) in matrix max-norm. Due to the consideration of computation efficiency and weaker assumptions, we recommend using the element-wise truncated estimator $\widehat{\bm{\Sigma}}_{\ell}$, and in \Cref{sec:Gaussian_autocov} we study the statistical inference of autocovariance matrices based on $\widehat{\bm{\Sigma}}_{\ell}$. In the remainder of this section, we describe each estimation method and provide the associated nonasymptotic results.

\subsection{Element-wise Huber's $M$-estimator}\label{sec:Huber}
We first introduce $\widetilde{\bm{\Sigma}}_{\ell}$ based on Huber's $M$-estimation, which we call the element-wise Huber's $M$-estimator. For each element, the $M$-estimator $\widetilde{\gamma}_{\ell,(jk)}$ of $\gamma_{\ell,(jk)}$ is defined as
\begin{equation*}
    \widetilde{\gamma}_{\ell,(jk)} = \widetilde{H}_{\ell,(jk)} - \widetilde{\mu}_j\widetilde{\mu}_k, \; \text{ for } j,k \in [d] \text{ and } |\ell| \leq \lfloor C_{\mathrm{lag}}\log n \rfloor,
\end{equation*}
with
\begin{align*}
    \widetilde{H}_{\ell,(jk)} = \argmin_{u \in \mathbb{R}} (n-\ell)^{-1}\sum_{i=\ell+1}^{n}L_{\tau_{\ell}}(H_{i,\ell,(jk)} - u)
\end{align*}
and
\begin{align*}
    \widetilde{\mu}_j = \argmin_{u \in \mathbb{R}} n^{-1}\sum_{i=1}^n\psi_{\tau}(X_{i,j} - u).
\end{align*}
estimating respectively $\frac{1}{n-\ell}\sum_{i=1}^{n-\ell}\mathbb{E}[H_{i,\ell,(jk)}]$ and $\mu_j$, and the Huber loss $L_{\tau}(u)$ is given in \eqref{eq:Huber.loss}. Note that such an estimator involves robustification parameters $\tau_{\ell}$ and $\tau$. As discussed in Section \ref{sec:dist_rob}, both $\tau_{\ell}$ and $\tau$ need to be chosen properly to balance the tail-robustness, i.e.~exponential-type tail probability deviation bound, and the bias due to truncation. Theoretical guidance for choosing these tuning parameters is given in \Cref{Thm:m-est.autocov}, which shows that we can set $\tau$ and $\tau_{\ell}$ for any $|\ell| \leq \lfloor C_{\mathrm{lag}}\log n \rfloor$ to be the same value. Moreover, in \Cref{sec:gap-block_CV}, we describe a block-wise cross-validation method that practically selects them. The corresponding lag-$\ell$ autocovariance matrix is denoted as
    \begin{equation}\label{eq:ele.M-est.autocov}
        \widetilde{\bm{\Sigma}}_{\ell} = \big( \widetilde{\gamma}_{\ell,(jk)} \big)_{j,k \in [d]} \; \text{ and } \; \widetilde{\bm{\mu}} = (\widetilde{\mu}_1, \dots, \widetilde{\mu}_{d})^{\intercal}.
    \end{equation}
The nonasymptotic result for $\widetilde{\bm{\Sigma}}_{\ell}$ is provided in Theorem \ref{Thm:m-est.autocov} below. 
Before stating our theorem, we introduce the following smoothness assumption on the distribution of $X_{i,j}$.
\begin{assumption}\label{assumption:bounded_pdf}
    For all $j \in [d]$, the marginal distribution of $X_{i,j}$ is absolutely continuous and has a bounded density function, i.e.~$\sup_{x \in \mathbb{R}}f_{X}(x) = \sup_{i \in \mathbb Z}\max_{j \in [d]}\sup_{x \in \mathbb{R}}f_{X_{i,j}}(x) < \infty$.
\end{assumption}
\begin{theorem}\label{Thm:m-est.autocov}
Consider only $|\ell| \leq \lfloor C_{\mathrm{lag}}\log n \rfloor$ with $C_{\mathrm{lag}} > 0$ being an absolute constant.
    For any $t > 0$, provided $n$ is large enough such that
\begin{equation*}
    n \geq C(\log n)^2(t + 2\log d),
\end{equation*}
with $C > 0$ being a sufficiently large absolute constant. Choose the robustification parameters 
    \begin{equation*}
        \tau_{\ell} = \tau \asymp (\log n)^{-1} \sqrt{\frac{n}{t+\log d}}, \;\;\text{for}\;\; |\ell| \leq \lfloor C_{\mathrm{lag}}\log n \rfloor.
    \end{equation*}
    Then, under Assumptions \ref{assumption:coor_4-moment}, \ref{assumption:coor_dep} and \ref{assumption:bounded_pdf}, we have with probability at least $1 - 6e^{-t}$
    \begin{equation}\label{eq:cov.m-est}
        \begin{aligned}
        \big\Vert\widetilde{\bm{\Sigma}}_{\ell} - \bm{\Sigma}_{\ell}\big\Vert_{\max} \lesssim \Vert X_{\cdot}\Vert_{4}\omega_4(\log n)\sqrt{\frac{t+2\log d}{n}}.
        \end{aligned}
    \end{equation}
\end{theorem}
\begin{remark}\label{remark:plugin_mean_est}
    Recall that \eqref{eq:autocov} suggests estimating $\mathbb{E}[X_{i-\ell,j}X_{i,k}]$ and $\mu_j$ can be treated separately with the same type of element-wise tail-robust estimators. We show in the proof of \Cref{Thm:m-est.autocov} that the estimation error of the former dominates the latter.
    Therefore, the deviation error given in \eqref{eq:cov.m-est} is essentially the deviation error of $\|\widetilde{\bm{H}}_{\ell} - \frac{1}{n-\ell}\sum_{i=1}^{n-\ell}\mathbb{E}[\bm{H}_{i,\ell}]\|_{\max}$. A heuristic explanation for this domination is that $\frac{1}{n-\ell}\sum_{i=1}^{n-\ell}\mathbb{E}[\bm{H}_{i,\ell}]$ contains $d^2$ number of unknown parameters and the error accumulates across entries, while $\bm{\mu}$ only contains $d$ number of unknown parameters. This observation is also applicable to the element-wise truncated estimator to be introduced in Section \ref{sec:elementwise}. Hence, we will focus on analyzing the estimation error of $\frac{1}{n-\ell}\sum_{i=1}^{n-\ell}\mathbb{E}[\bm{H}_{i,\ell}]$ therein.
\end{remark}
\begin{remark}\label{remark:m-est2}
The error $\Vert \widetilde{\bm{\Sigma}}_{\ell} - \bm{\Sigma}_{\ell}\Vert_{\max}$ in Theorem \ref{Thm:m-est.autocov} is of order $(\log n)\sqrt{\log d/n}$, which is optimal up to an $\log n$ factor in the minimax sense. In terms of consistency, the dimension $d$ is allowed to grow exponentially with $n$ as long as $(\log d) (\log n)^2/n \to 0$. The $\log n$ term is a result of applying Bernstein's inequality (\Cref{thm:bernstein_bounded}) for general nonlinear processes. Thus, for linear processes, using instead \Cref{thm:bernstein_bound_expdec_new}, we can show that $\widetilde{\bm{\Sigma}}_{\ell}$ achieves exactly the minimax optimal rate.
\end{remark}

\subsection{Element-wise truncated estimator}\label{sec:elementwise}
Recall the truncation operator $\psi_{\tau}(\cdot)$ defined in \eqref{eq:trun.operator}.
Following \eqref{eq:autocov}, we define the truncated estimator of $\gamma_{\ell,(jk)}$ as
\begin{equation*}
    \widehat{\gamma}_{\ell,(jk)} = \widehat{H}_{\ell,(jk)} - \widehat{\mu}_j\widehat{\mu}_k, \; \text{ for } j,k \in [d] \text{ and } |\ell| \leq \lfloor C_{\mathrm{lag}}\log n \rfloor,
\end{equation*}
with 
\begin{align*}
    \widehat{H}_{\ell,(jk)} = \frac{1}{n-\ell}\sum_{i=l+1}^{n}\psi_{\tau}(H_{i,\ell,(jk)}) \;\; \text{and} \;\; \widehat{\mu}_j = \frac{1}{n}\sum_{i=1}^n\psi_{\tau}(X_{i,j}),
\end{align*}
estimating respectively $\mathbb{E}[H_{i,\ell,(jk)}]$ and $\mu_j$. Here, we set the same robustification parameter $\tau > 0$ for both $\widehat{H}_{\ell,(jk)}$ and $\widehat{\mu}_j$, which can be selected in practice by the block-wise cross-validation method given in \Cref{sec:gap-block_CV}.  The corresponding lag-$\ell$ autocovariance matrix estimator is denoted as
    \begin{equation}\label{eq:ele.trunc.est.autocov}
        \widehat{\bm{\Sigma}}_{\ell} = \big( \widehat{\gamma}_{\ell,(jk)} \big)_{j,k \in [d]} \; \text{ and } \; \widehat{\bm{\mu}} = (\widehat{\mu}_1, \dots, \widehat{\mu}_{d})^{\intercal}.
    \end{equation}
Compared to \eqref{eq:ele.M-est.autocov}, the truncated autocovariance matrix estimator \eqref{eq:ele.trunc.est.autocov} has a closed form and hence it can be computed easily. Moreover, the nonasymptotic property of $\widehat{\bm{\Sigma}}_{\ell}$ provided in the following theorem shows that the same optimal (up to an $\log n$ factor) error rate is attainable without assuming the bounded marginal density, i.e.~\Cref{assumption:bounded_pdf}.
\begin{theorem}
\label{Thm:cov.element.huber.exp}
For any $t > 0$, choose the robustification parameter
\begin{equation}\label{eq:tau_tail}
    \tau \asymp (\log n)^{-1}\sqrt{\frac{n}{t + 2\log d}}.
\end{equation}
Under Assumptions \ref{assumption:coor_4-moment} and \ref{assumption:coor_dep}, for a sufficiently large $n$ such that $n-\ell \geq 4 \vee c/2$, we have with probability at least $1- 4e^{-t}$,
\begin{equation}\label{eq:cov.element.huber.exp}
    \big\Vert \widehat{\bm{\Sigma}}_{\ell} - \bm{\Sigma}_{\ell} \big\Vert_{\max} \lesssim \Vert X_{.}\Vert_{4}\omega_4(\log n)\sqrt{\frac{t+2\log d}{n}}.
\end{equation}
\end{theorem}
\begin{remark}
Similar remark as \Cref{remark:m-est2} can also be drawn for $\widehat{\bm{\Sigma}}_{\ell}$. We highlight that \Cref{thm:bernstein_bounded} is the key for proving \Cref{Thm:cov.element.huber.exp}, thus our results can be extended to piece-wise stationary processes and potentially nonstationary processes with the population quantity $\bm{\Sigma}_{\ell}$ being necessarily modified.
\end{remark}

\subsection{Gap-block cross-validation}\label{sec:gap-block_CV}
In this subsection, we introduce a gap-block cross-validation method adapted from Shu and Nan \cite{shu2019estimation} to select the robustification parameter $\tau$ for our tail-robust estimators. The steps of this method is detailed below. Its good performance is justified by our numerical studies, as presented in \Cref{sec:simulation}.
\begin{enumerate}
    \item Given observations $\{\bm{X}_i\}_{i \in [n]}$, partition the index set $[n]$ into $H_1 \geq 4$ consecutive blocks $\{B_i\}_{i \in [H_1]}$ with approximately equal-sizes $\lfloor n/H_1 \rfloor$, such that $[n] = \cup_{i \in [H_1]} B_i$. For each $i \in [H_1]$, consider $B_i$ as the set of indices of the validation data, and use the remaining data after removing the neighboring blocks at both sides as the training data.
    \item Randomly sample $H_2$ starting indices without replacement form $[n-\lfloor n/H_1 \rfloor+1]$. Based on the starting indices, extract $H_2$ blocks of size $\lfloor n/H_1 \rfloor$ from $[n]$, denoted as $\{B_{H_1+j}\}_{j \in [H_2]}$. For each $j$, consider $B_{H_1+j}$ as the set of indices of the validation data, and use the remaining data after removing $\lfloor n/H_1 \rfloor$ elements at both sides as the training data.
    \item For each $i \in [H_1+H_2]$, we compute a reference lag-$l$ autocovariance matrix based on the corresponding validation data. This reference matrix is computed by averaging the smallest $95\%$ lag-$l$ outer products with respect to max-norm. Then, we compute the tail-robust autocovariance matrix estimator based on corresponding training data with each candidate $\tau$, and compute the difference in max-norm between the tail-robust autocovariance estimator and the reference matrix for each $\tau$.
    \item Select the robustification parameter $\tau$ by minimizing the averaged error obtained in the previous step. 
\end{enumerate}
The first step is similar to the classical cross-validation for temporal independent data. However, with dependent data, the neighboring blocks of the validation set are removed in order to reduce the dependence between the validation data and the training data. Once $H_1$ is given, the block splits in the first step are determined. In the second step, additional splits are provided. After these two steps, $H_1 + H_2$ sets of validation and training data are obtained. In the third step, we compute the empirical errors in max-norm for each candidate $\tau$. Since the lag-$\ell$ population autocovariance matrix is unknown, we compute a reference matrix based on the smallest $95\%$ lag-$\ell$ outer products to reduce the impact of heavy-tailedness. The ratio $95\%$ is an arbitrary choice. Then, the fourth step produces the selected $\tau$ by minimizing the averaged empirical error.

\section{Gaussian approximation}\label{sec:Gaussian_autocov}
In this section, we study the Gaussian approximation for our element-wise truncated autocovariance estimator. To be specific, we aim to show that, for any fixed $\ell \in \mathbb{Z}$, the limiting distribution of $(n-\ell)^{1/2}\big\Vert\widehat{\bm{\Sigma}}_{\ell}\big\Vert_{\max}$ can be approximated well, in terms of the Kolmogorov-Smirnov distance, by the $\ell_{\infty}$ norm of an $\mathbb{R}^{d^2}$-valued Gaussian vector $\bm{Z} \sim N(\bm{0}, \bm{\Gamma})$, where $\bm{\Gamma}$ is defined as a long-run covariance matrix, such that for any $\bm{s} \in \mathbb{R}^{d^2}$
\begin{align}\label{eq:long-run}
\bm{s}^{\top}\bm{\Gamma}\bm{s} = \lim_{n \to \infty} \var\big(\bm{s}^{\top}\bm{U}\big),
\end{align}
and $\bm{U}$ is the vectorization (i.e.~staking the columns into a vector) of $(n - \ell)^{1/2}\widehat{\bm{\Sigma}}_{\ell}$. We give the Gaussian approximation result in \Cref{thm:GA_autocov}. The proof of \Cref{thm:GA_autocov} is a direct application of Theorem~D.1 in the supplementary material, i.e.~a Gaussian approximation result for element-wise truncated mean estimator under heavy-tailedness and temporal dependence, which may of independent interest. Next, we introduce and discuss several assumptions before stating the theorem. Note that, instead of pursuing the minimum moment condition ($\theta \in (0,1]$) required by Theorem~D.1, the following assumptions are based on $\theta = 1$ for the simplicity of presentation. Also, since we are focusing on the inference of high-dimensional autocovariances, we assume that $\{\bm{X}_i\}_{i \in \mathbb{Z}}$ is a zero mean potentially nonstationary process. In practice, we can always centralized our time series before proceeding the proposed inference methodology.
\begin{assumption}\label{assumption:coor_6-moment}
    Let $\omega_6 = \sup_{i \in \mathbb Z}\max_{j \in [d]}\Vert X_{i,j} \Vert_6 < \infty$.
\end{assumption}
\begin{assumption}\label{assumption:coor_dep_6}
    There exists some $c > 0$, such that
    \begin{align*}
        \Vert X_{\cdot} \Vert_6 = \max_{j \in [d]}\sup_{m \geq 0}\exp(-cm)\sum_{k = m}^{\infty}\delta_{k,6,j} < \infty.
    \end{align*}
\end{assumption}
\begin{assumption}\label{assumption:var_nondegen}
    There exists an absolute constant $b > 0$, such that 
    \begin{align*}
        \min_{j,k \in [d]}\inf_{\mathcal{S} \subseteq [1,n]}\frac{1}{|\mathcal{S}|}\var\Big(\sum_{i \in \mathcal{S}}X_{i-\ell,j}X_{i,k} \Big) > b.
    \end{align*}
\end{assumption}

\begin{remark}
Assumption \ref{assumption:coor_6-moment}  requires finite coordinate-wise moments up to sixth order for all dimensions of $\{\bm{X}_i\}_{i = 1}^n$. Assumption \ref{assumption:coor_dep_6} requires an exponential decay of dependence measure for all dimensions of $\{\bm{X}_i\}_{i = 1}^n$. Assumptions \ref{assumption:coor_6-moment} and \ref{assumption:coor_dep_6} are imposed on higher order moments of $\{\bm{X}_i\}_{i = 1}^n$ and hence can imply Assumptions \ref{assumption:coor_4-moment} and \ref{assumption:coor_dep}. Assumption \ref{assumption:var_nondegen} ensures the nondegeneracy of the partial sums of lag-$\ell$ cross products, which is a very mild condition.
\end{remark}

\begin{corollary}[Gaussian approximation of truncated autocovariance estimator]\label{thm:GA_autocov}
Suppose Assumptions \ref{assumption:coor_6-moment}, \ref{assumption:coor_dep_6} and \ref{assumption:var_nondegen} hold. Let $C, C_{\tau} > 0$ be some absolute constants. Assume that $\log d = Cn^{\beta}$ for some $\beta < 1/19$. Choose the robustafication parameter 
\begin{equation}\label{eq:tau_GA}
\tau = C_{\tau}\left(\frac{n^{16/19}}{\log d}\right)^{1/3}.
\end{equation}
    Then, as $n \to \infty$, we have that for any fixed $\ell \in \mathbb{Z}$
    \begin{align*}
    \sup_{t \in \mathbb{R}}\left|\mathbb{P}\Big((n-\ell)^{1/2}\big\Vert\widehat{\bm{\Sigma}}_{\ell} - \bm{\Sigma}_{\ell}\big\Vert_{\max} \leq t\Big) - \mathbb{P}\Big(\big|\bm{Z}\big|_{\infty} \leq t\Big)\right|
    \lesssim n^{-2(1/19-\beta)/3} \to 0,
\end{align*}
where $\bm{Z} \sim N(\bm{0}, \bm{\Gamma})$ and $\bm{\Gamma}$ is defined as  in \eqref{eq:long-run}.
\end{corollary}
\begin{remark}\label{remark:different_tau}
    \Cref{thm:GA_autocov} follows directly from Theorem~D.1 in the supplementary material. Note that the choice of the robust parameter $\tau$ in \eqref{eq:tau_GA} may be different from the choice in \eqref{eq:tau_tail}. 
    Both choices try to balance the bias and robustness trade-off but under different measurements. The bias and robustness are measured by the Kolmogorov-Smirnov distance in \Cref{thm:GA_autocov}, while they are measured by the matrix max norm in \Cref{Thm:cov.element.huber.exp}. These two choices of $\tau$ can be of the same order if we set $t$ in \eqref{eq:tau_tail} as
    \begin{equation*}
        t = C_1\frac{n^{25/57}(\log d)^{2/3}}{(\log n)^2}.
    \end{equation*}
    The above choice of $t$ leads to  $t \to \infty$ as $n \to \infty$. Therefore, \Cref{Thm:cov.element.huber.exp} together with $\beta < 1/19$ guarantees that with probability at least $1 - 4e^{-t}$
    \begin{equation*}
        \big\Vert \widehat{\bm{\Sigma}}_{\ell} - \bm{\Sigma}_{\ell} \big\Vert_{\max} \lesssim \left(\frac{\log d}{n^{16/19}}\right)^{2/3} \bigvee (\log n)\sqrt{\frac{\log d}{n}} \lesssim (\log n)\sqrt{\frac{\log d}{n}}.
    \end{equation*}
\end{remark}

The Gaussian approximation result in \Cref{thm:GA_autocov} addresses various inference problems, such as the test of serial correlations, i.e.~$\bm{\Sigma}_{\ell} = 0$ for some or all $\ell \in \mathbb{Z}$, and the change point detection in autocovariance structures,  i.e.~$\bm{\Sigma}^{(1)}_{\ell} = \bm{\Sigma}^{(2)}_{\ell}$ and $\bm{\Sigma}^{(1)}_{\ell}$ and $\bm{\Sigma}^{(2)}_{\ell}$ are true lag-$\ell$ autocovariances before and after the change point. However, the asymptotic covariance matrix $\bm{\Gamma}$, having dimensions $d^2 \times d^2$, is typically unknown and challenging to estimate directly due to its size and complexity. In response to this challenge, we propose a block-wise Gaussian multiplier bootstrap method.

Given $\{\bm{X}_i\}_{i = 1}^n$, we construct $\{\bm{H}_{i,\ell}\}_{i = \ell+1}^n$, the sequence of lag-$\ell$ outer products defined as in \eqref{eq:outer_prod}. For $R \in \mathbb{N}_+$, we divide the whole time interval $[\ell+1, n]$ into $2R$ number of sub-intervals (blocks). For simplicity, we assume it is divisible and let the block size $S = (n-\ell)/(2R) \in \mathbb{N}_+$. The $2R$ number of sub-intervals, denoted by $\mathcal{S}_1, \dots, \mathcal{S}_{2R}$, can be expressed as
            \begin{align*}
                \mathcal{S}_r=
                [\ell + 1 + (r-1)S, \ell + 1 + rS], \;\; \text{for}\;\; r = 1, \dots, 2R.
            \end{align*}
For $r^* \in \{1, \dots, R\}$, the local tail-robust autocovariance estimators in the $r^*$-th paired odd and even blocks can be written as follows
\begin{align*}
    \widehat{\bm{\Sigma}}_{\mathcal{S}_{2r^*-1}} = \frac{1}{S}\sum_{t \in \mathcal{S}_{2r^*-1}} \psi_{\tau}(\bm{H}_{t, \ell}) 
    \quad \text{and} \quad 
    \widehat{\bm{\Sigma}}_{\mathcal{S}_{2r^*}} = \frac{1}{S}\sum_{t \in \mathcal{S}_{2r^*}} \psi_{\tau}(\bm{H}_{t, \ell}).
\end{align*}

We are ready to describe the block-wise Gaussian multiplier bootstrap.  Let $M$ be the number of bootstrap samples.
Let $\{e_{r^*}\}_{r^* = 1}^{R}$ be a sequence of iid standard normal random variables, and  $\{e_{r^*}^{(m)}\}_{r^* = 1}^{R}$ be an iid copy of $\{e_{r^*}\}_{r^* = 1}^{R}$ used in the $m$-th bootstrap sample, for $m=1, \ldots, M$. The $m$-th Gaussian multiplier lag-$\ell$ moving sum difference matrix can be constructed by
\begin{align*}
    \bm{S}_{\ell}^{(m)} = (n-\ell)^{-1/2}S\sum_{r^* = 1}^{R}e_{r^*}^{(m)} (\widehat{\bm{\Sigma}}_{\mathcal{S}_{2r^*-1}} - \widehat{\bm{\Sigma}}_{\mathcal{S}_{2r^*}}), \;\; \text{for} \;\; m=1, \ldots, M.
\end{align*}
The $m$-th Gaussian multiplier bootstrapped sample of $(n-\ell)^{1/2}\big\Vert\widehat{\bm{\Sigma}}_{\ell} - \bm{\Sigma}_{\ell}\big\Vert_{\max}$ is defined as
\begin{align}\label{eq:T_stat_multiplier_boot}
    T_{\ell}^{(m)} = \big\Vert \bm{S}_{\ell}^{(m)}\big\Vert_{\max}, \;\; \text{for} \;\; m=1, \ldots, M.
\end{align}
The block-wise differences involved in $\bm{S}_{\ell}^{(m)}$ intends to remove $\bm{\Sigma}_{\ell}$ in the Gaussian multiplier bootstrap, which is different from the classical version proposed in Chernozhukov et al. \cite{chernozhukov2013gaussian} and Zhang and Cheng \cite{zhang2018gaussian}, where the global estimator $\widehat{\bm{\Sigma}}_{\ell}$ is deducted for this purpose. We use blocks of size $S$ to preserve the underlying temporal dependence. Thus, the bootstrapped sample $\{\big\Vert \bm{S}_{\ell}^{(m)}\big\Vert_{\max}\}_{m=1}^M$ well approximates the empirical distribution of $T_{\ell} = (n-\ell)^{1/2}\big\Vert\widehat{\bm{\Sigma}}_{\ell} - \bm{\Sigma}_{\ell}\big\Vert_{\max}$ with suitable $S$ and with large $M$.

Denote a generic Gaussian multiplier bootstrapped statistic defined in \eqref{eq:T_stat_multiplier_boot} by $T_{\ell}^{(boot)}$.
For a significance level $\alpha \in (0,1)$, define the conditional $(1-\alpha)$-th quantile of $T_{\ell}^{(boot)}$ given $\mathcal{X} = \{\bm{X}_i\}_{i = 1}^n$ as
\begin{equation}\label{eq:empirical_quantile}
    q^{(boot)}(1-\alpha) = \inf\left\{u \in \mathbb{R}: \; \mathbb{P}\left(T_{\ell}^{(boot)} \leq u \ \big| \  \mathcal{X}\right) \geq 1-\alpha\right\}.
\end{equation} 
The next theorem shows the consistency of the proposed Gaussian multiplier bootstrap.

\begin{theorem}\label{thm:test_size}
    Suppose the conditions in Corollary \ref{thm:GA_autocov} hold. We choose 
    \[\tau = C_{\tau}\bigg(\frac{n^{16/19}}{\log d}\bigg)^{1/3}\]
    and  $S = C_1n^{(1-\beta)/5}$ (or equivalently $R = C_2n^{(4+\beta)/5}$), with absolute constants $C_{\tau}, C_1, C_2 > 0$ and $\beta < 1/19$ is given in Corollary \ref{thm:GA_autocov}. Then it holds, as $n \to \infty$, that for any fixed $\ell \in \mathbb{Z}$
    \begin{align*}
        \sup_{\alpha \in (0,1)}\Big|\mathbb{P}\Big(T_{\ell} \leq q^{(boot)}(\alpha)\Big) - \alpha\Big| \leq C_3n^{-2(1/19 - \beta)/3} \to 0, 
    \end{align*}
where $T_{\ell} = (n-\ell)^{1/2}\big\Vert\widehat{\bm{\Sigma}}_{\ell} - \bm{\Sigma}_{\ell}\big\Vert_{\max}$.
\end{theorem}

\section{Numeric results}\label{sec:simulation}
In this section, we conduct simulation studies in various scenarios as well as a real data example. In particular, we assess the finite sample performance of our proposed autocovariance estimators in terms of estimation and inference, as detailed in Sections~\ref{sec:simu-est} and \ref{sec:simu-inference}, respectively. In Section~\ref{sec:real:data}, we apply our procedure to detect change points in real-world economic data.
\subsection{Estimation}\label{sec:simu-est}
We conduct simulated experiments to validate the nonasymptotic results for the two tail-robust autocovariance matrix estimators studied in Section \ref{sec:robust}. It is known that the median-of-means is another tail-robust estimator \cite[see e.g.,][]{lerasle2011robust}. We adapt the median-of-means to estimate autocovariance matrices and use it as a competitor. We compare the performance of these estimators to that of the sample autocovariance matrix.
For the element-wise truncated estimator and the element-wise Huber's $M$-estimator, we select the robustification parameter $\tau$ by the gap-block cross-validation, detailed in Subsection \ref{sec:gap-block_CV}, with $H_1 = H_2 = 10$. We consider the following two scenarios.
\\
{\bf Scenario 1: Stationary process.} Data are simulated from the $d$-dimensional VAR($1$) model
\begin{equation*}
    \bm{X}_i = \rho\bm{X}_{i-1} + \bm{Z}_i,
\end{equation*}
where the parameter $\rho = 0.5$ is a scalar, $\{\bm{Z}_i\}_{i \in \mathbb{Z}}$ are iid error process, with $\mathbb{E}[\bm{Z}_i] = \bm{0}$ and $\var(\bm{Z}_i) = \bm{\Gamma} \in \mathbb{R}^{d \times d}$ is a deterministic matrix. Equivalently, we write $\bm{Z}_i = \bm{\Gamma}^{1/2}\bm{\epsilon}_i$ and $\bm{\epsilon}_i \in \mathbb{R}^d$ are iid with $\mathbb{E}[\bm{\epsilon}_i] = \bm{0}$ and $\var(\bm{\epsilon}_i) = \bm{I}_d$.
We consider the following four distributions of $\epsilon_{i,j}$, $i \in [n], j \in [d]$.
\begin{enumerate}
    \item[(1)] (Normal). $\epsilon_{i,j}$ follows a standard Normal distribution.
    \item[(2)] (Pareto). $\epsilon_{i,j}$ follows a standardized Pareto distribution, i.e.~$\epsilon_{i,j} = (3/4)^{-1/2}(Y_{i,j} - 3/2)$ where $Y_{i,j}$'s are iid from a Pareto distribution with a shape parameter $3$ and a scale parameter $1$.
    \item[(3)] (Log-Normal). $\epsilon_{i,j}$ follows a standardized Log-normal distribution, i.e.~$\epsilon_{i,j} = (e^2 - e)^{-1/2}[\exp(Y_{i,j}) - \exp(1/2)]$ where $Y_{i,j}$'s are iid from a standard Normal distribution.
    \item[(4)] (Student's $t$). $\epsilon_{i,j}$ follows a standardized Student's $t_4$ distribution, i.e.~$\epsilon_{i,j} = 2^{-1/2}Y_{i,j}$ where $Y_{i,j}$'s are iid from a $t_4$ distribution.
\end{enumerate}
Moreover, we consider the following three different structures for $\bm{\Gamma}$.
\begin{enumerate}
    \item[(a)] (Diagonal structure). $\bm{\Gamma} = \bm{I}_d$.
    \item[(b)] (Equal correlation structure). $\Gamma_{(ij)} = 1$ if $i = j$ and $\Gamma_{(ij)} = 0.5$ if $i \neq j$.
    \item[(c)] (Power decay structure). $\Gamma_{(ij)} = 0.5^{|i-j|}$.
\end{enumerate}
Since $\bm{\Gamma}$ is symmetric, the population lag-$\ell$ autocovariance matrix is
\begin{equation*}
    \bm{\Sigma}_\ell = (1-\rho^2)^{-1}\rho^{|\ell|}\bm{\Gamma}.
\end{equation*}
For each of the above scenarios, we vary $n \in \{50, 100\}$ and $d \in \{50,100,150\}$, and simulate $200$ replicates. Following Ke et al. \cite{ke2019user}, we assess the comparison by the Relative Mean Error (RME) under spectral, max and Frobenius norms:
\begin{equation*}
    \text{RME}_{\ell} = \frac{\sum_{r = 1}^{200}\Vert \widehat{\bm{\Sigma}}_{\ell,(r)} - \bm{\Sigma}_{\ell} \Vert_{,\max,\text{F}}}{\sum_{r = 1}^{200}\Vert \widetilde{\bm{\Sigma}}_{\ell,(r)} - \bm{\Sigma}_{\ell} \Vert_{,\max,\text{F}}},
\end{equation*}
where $\widehat{\bm{\Sigma}}_{\ell,(r)}$ is one of the tail-robust estimators of $\bm{\Sigma}_{\ell}$ in the $r$-th simulation, and $\widetilde{\bm{\Sigma}}_{\ell,(r)}$ is the sample estimator. %

\begin{figure}
\begin{center}
\includegraphics[width=5in]{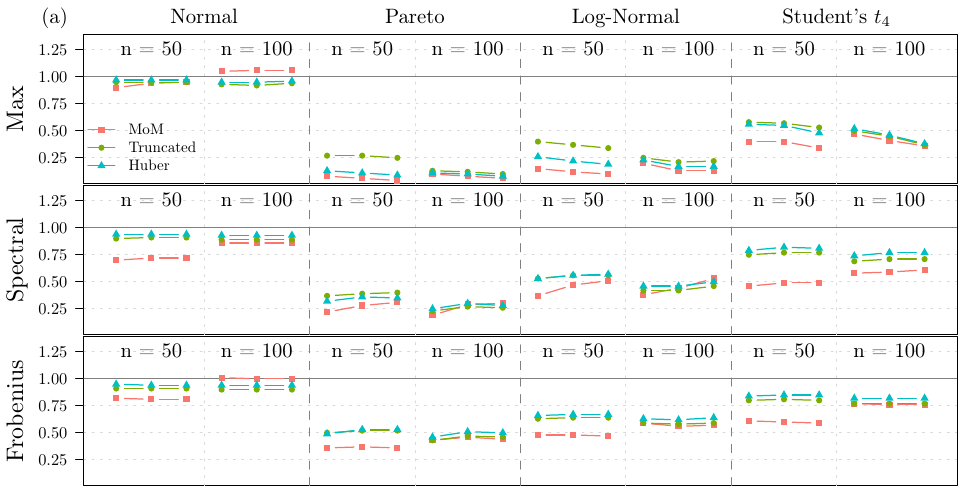}
\includegraphics[width=5in]{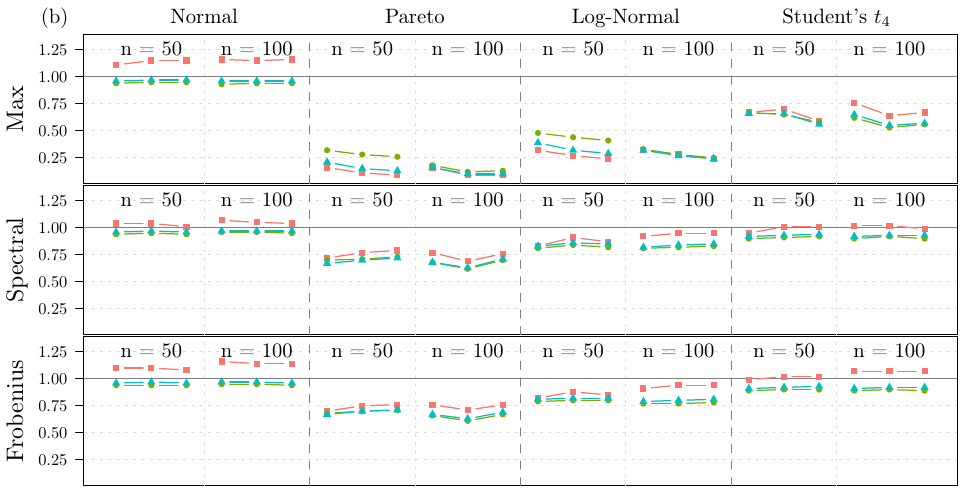}
\includegraphics[width=5in]{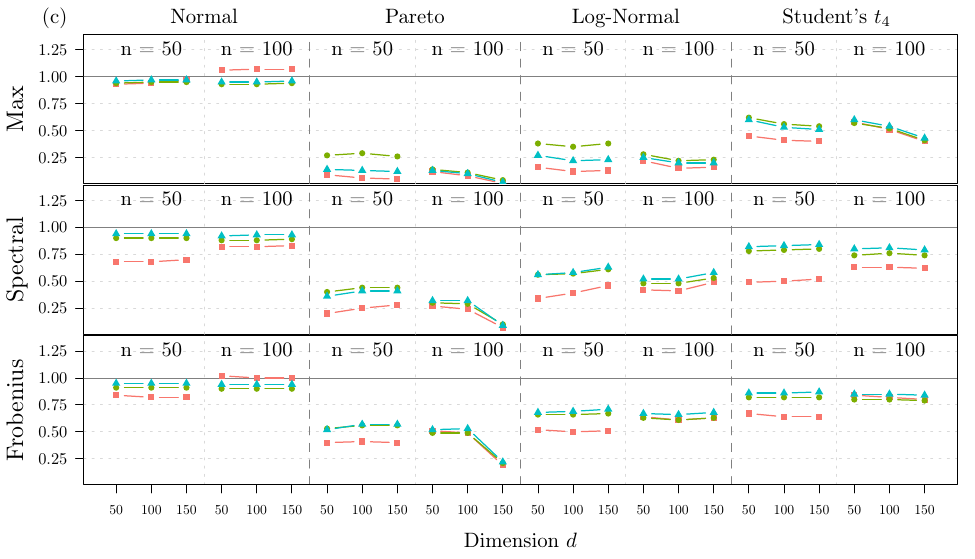}
\caption{RMEs of tail-robust lag-$1$ autocovariance matrix estimators with respective to max, spectral and Frobenius norms. Data are simulated, with $200$ replicates, from VAR($1$) model with $\rho = 0.5$, and innovations follow Normal, Pareto, Log-Normal and Students's $t_4$ distributions.%
Three structures of $\bm{\Gamma}$ (diagonal, equal correlation and power decay) are used respectively in panel (a), (b) and (c).}
\label{fig:var_lag1}
\end{center}
\end{figure}

We summarize the RMEs of these four tail-robust autocovariance estimators in Figure \ref{fig:var_lag1}, for $\ell = 1$ and under the three covariance structures. In the figure legend, we abbreviate the names of three tail-robust estimators as ``Truncated'', ``Huber'' and ``MoM''. For all three matrix norms, when the value of RME is less than $1$, the corresponding tail-robust estimator outperforms the sample autocovariance, and vice versa. It is shown that the performance of the median-of-means estimator has large dispersion compared to our estimators. This suggests that its performance is highly sensitive to the selection of the number of blocks, which becomes more evident when $n=50$. However, in general, our estimators perform much better than the sample autocovariance matrix in all heavy-tailed settings, and only slightly outperform the sample autocovariance matrix in the Normal setting. These results correspond to our nonasymptotic results given in Section \ref{sec:robust}. Moreover, with the same cross-validation criteria for selecting the robustification parameters, the element-wise Huber's $M$-estimator always performs similarly to the element-wise truncated estimator, which suggests the use of element-wise truncated estimator in practical, due to its computational efficiency. Similar conclusions can also be drawn when estimating $\bm{\Sigma}_{\ell}$ with $\ell = 0$ and $\ell = 2$. These results are presented in Section A of the supplementary material.
\\
\\
{\bf Scenario 2: Nonstationary process.}  Data are simulated from the following $\mathbb{R}^d$-valued process
\begin{equation*}
    \bm{X}_i = 
    \begin{cases}
    \sum_{k=0}^{1000}\bm{A}^{(1)}_k\bm{\epsilon}_{i-k}, \;\; \text{for} \;\; i \in \{1, \dots, \lfloor n/2\rfloor\},\\
    \sum_{k=0}^{1000}\bm{A}^{(2)}_k\bm{\epsilon}_{i-k},  \;\; \text{for} \;\; i \in \{\lfloor n/2 \rfloor+1, \dots, n\},
    \end{cases}
\end{equation*}
where $\{\bm{A}^{(1)}_k\}_{k=0}^{1000}$ and $\{\bm{A}^{(2)}_k\}_{k=0}^{1000}$ are sequences of $d \times d$ deterministic matrices, whose entries are independently generated from standard normal distribution. Once generated, these matrices are fixed throughout the simulation. The innovations $\{\bm{\epsilon}_i \in \mathbb{R}^d\}_{i \in \mathbb Z}$ are iid with $\mathbb{E}[\bm{\epsilon}_i] = \bm{0}$ and $\var(\bm{\epsilon}_i) = \bm{I}_d$.
We consider the same distributions of $\epsilon_{i,j}$ as in Section \ref{sec:simulation}.
The population lag-$\ell$ autocovariance matrix is
\begin{equation*}
    \bm{\Sigma}_\ell(i) = \begin{cases}
    \sum_{k=0}^{1000-\ell}\bm{A}^{(1)}_{k}(\bm{A}^{(1)})_{k+\ell}^{\intercal}, &\;\; \text{for} \;\; i \in \{1, \dots, \lfloor n/2\rfloor - \ell\},\\
    \sum_{k=0}^{1000-\ell}\bm{A}^{(1)}_{k}(\bm{A}^{(2)})_{k+\ell}^{\intercal}, &\;\; \text{for} \;\; i \in \{\lfloor n/2\rfloor - \ell, \dots, \lfloor n/2\rfloor\},\\
    \sum_{k=0}^{1000-\ell}\bm{A}^{(2)}_{k}(\bm{A}^{(2)})_{k+\ell}^{\intercal}, &\;\; \text{for} \;\; i \in \{\lfloor n/2\rfloor+1, \dots, n\}.
    \end{cases}
\end{equation*}

We present the simulation results for lag $\ell = 0, 1, 2$ in panels (a)-(c) of Figure \ref{fig:LP} respectively. These graphs are of the same format as the ones under VAR($1$) model, and they show the similar robust performance of these tail-robust autocovariance estimators.

\begin{figure}
\begin{center}
\includegraphics[width=5in]{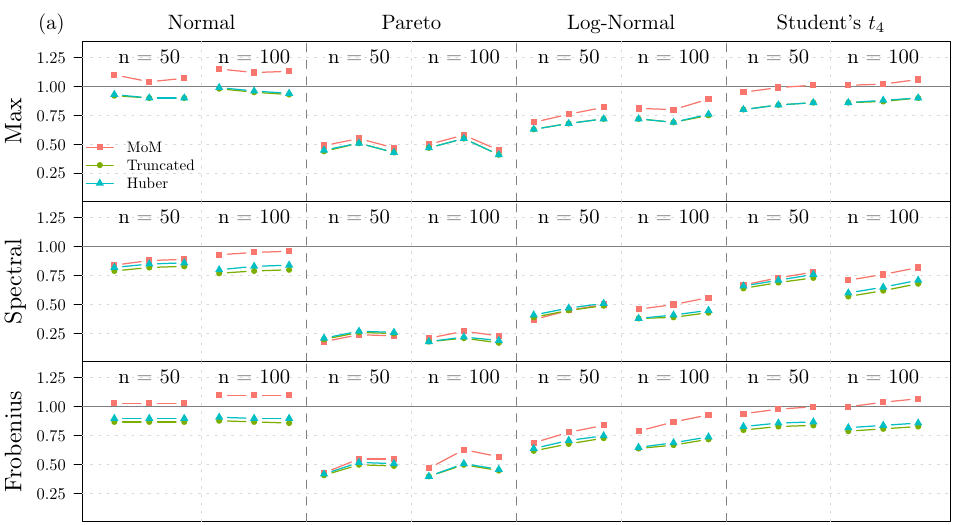}
\vspace{-0.2cm}
\includegraphics[width=5in]{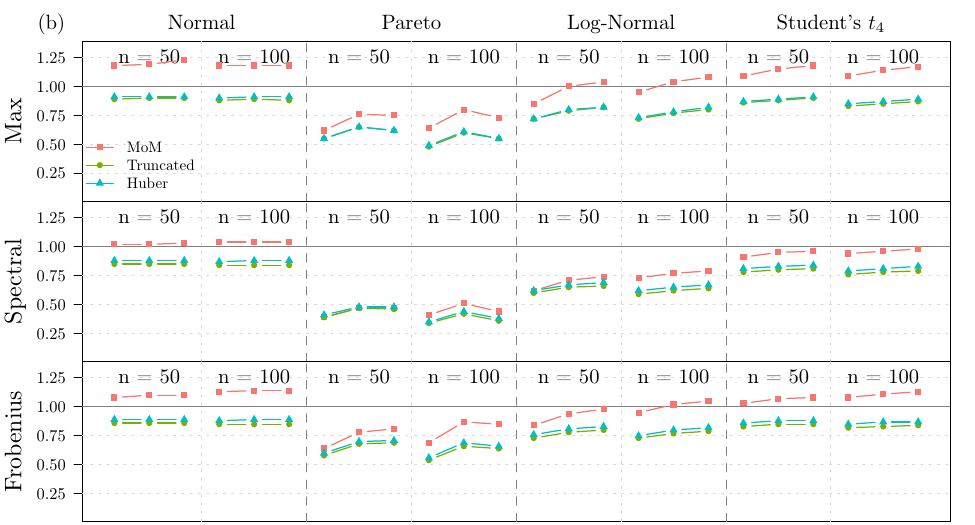}
\vspace{-0.2cm}
\includegraphics[width=5in]{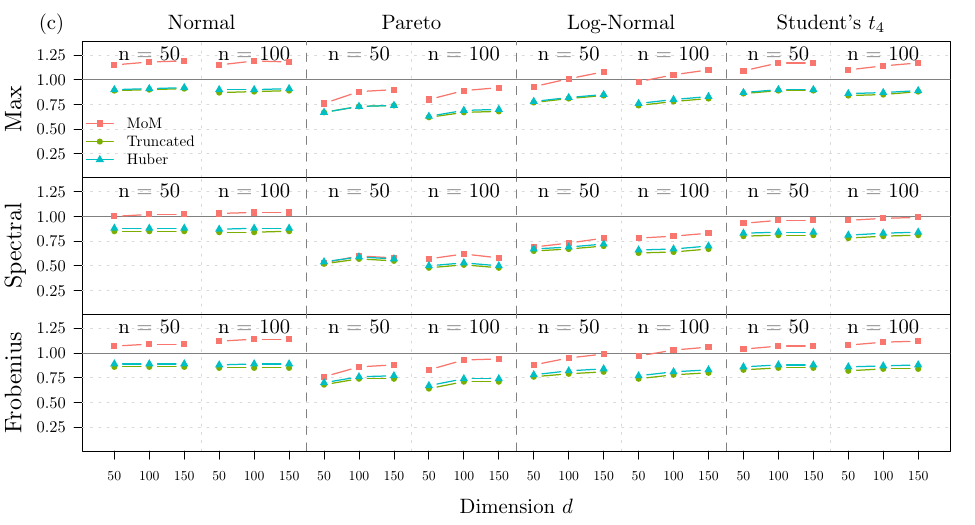}
\caption{RMEs of tail-robust autocovariance matrix estimators with respective to max, spectral and Frobenius norms. Data are simulated, with $200$ replicates, from a nonstationary linear process with innovations respectively following Normal, Pareto, Log-Normal and Students's $t_4$ distributions. 
Results for $\bm{\Sigma}_\ell$ with $\ell = 0, 1, 2$ are given respectively in panels (a), (b) and (c).}
\label{fig:LP}
\end{center}
\end{figure}

\subsection{Inference}\label{sec:simu-inference}
In this subsection, we focus on the element-wise truncated covariance matrix estimator and show the numeric performance of our proposed inference procedures in \Cref{sec:Gaussian_autocov} on testing covariance matrix of high-dimensional time series $\{\bm{X}_i\}_{i = 1}^n$. The null and alternative hypotheses are formalized as
\[
H_0: \cov(\bm{X}_i) = \bm{\Sigma}_0, \quad \text{and} \quad H_A: \cov(\bm{X}_i) \prec \bm{\Sigma}_0.
\]
We consider the same VAR(1) model as in \Cref{sec:simu-est} with $\rho = 0.3$, $\epsilon_{i,j}$ following standardized $t_6$ distribution. Under the null hypothesis, we let the covariance matrix of the innovation $\bm{\Gamma}_0 = \bm{I}_d$, which leads to the population covariance matrix be $\bm{\Sigma}_0 = (1-\rho^2)^{-1}\bm{I}_d$. Under the alternatives, we let $\bm{\Gamma}_A = \omega \cdot \bm{I}_d$. Thus, $\bm{\Sigma}_A = (1-\rho^2)^{-1}\omega \bm{I}_d$. 

Let $T = n^{1/2}\|\widehat{\bm{\Sigma}} - \bm{\Sigma}_0 \|_{\max}$ be the test statistic and $\widehat{\bm{\Sigma}}$ be the element-wise truncated covariance matrix estimator. For a significant level $\alpha \in (0,1)$, by \Cref{thm:test_size}, the (Gaussian multiplier) bootstrapped critical value is denoted as $q^{(boot)}(1-\alpha)$ defined in \eqref{eq:empirical_quantile}. In the following, we vary $n \in \{300, 500\}$ and $d \in \{50, 75, 100\}$. Note that $q^{(boot)}(1-\alpha)$ does not involve any unknown model parameters, and it is adaptive to each setting. The inference problem considered is more challenging than the estimation problem in \Cref{sec:simu-est}, since the former inherently involves estimating the $d^2\times d^2$ asymptotic covariance matrix. However, this issue is avoided by using the Gaussian multiplier bootstrap.

There are two tuning parameters involved in the procedure: (a) the robustification parameter $\tau$ and (b) the block size $S$ for the Gaussian multiplier bootstrap. In this subsection, we set $\tau = 3.2 \cdot (n^{16/19}/\log d)^{1/3}$ guided by \Cref{thm:GA_autocov}, and fix $S = 19$.
\\
\\
\textbf{Under $H_0$.} We simulate $\{\bm{X}_i\}_{i = 1}^n$ under the null with different $n$ and $d$. Based on the simulated data, $500$ Gaussian multiplier bootstrap are performed. Define
\begin{align*}
    \text{accept}_{n,d}(1-\alpha)  = \mathbbm{1}\left\{ T_{n,d} \leq q_{n,d}^{(boot)}(1-\alpha)\right\}.
\end{align*}
For each case, $500$ repetitions are conducted, and we report the proposition of acceptance in \Cref{tab:size}. The table shows that our procedure produces critical values match well with the nominal confidence level $1-\alpha$. As expected, the performance of the procedure increases as $n$ increases and/or $d$ decreases. 
\begin{table}[htbp]
\centering
\caption{Proposition of acceptance under $H_0$ based on the Gaussian multiplier bootstrapped critical value.}
\begin{tabular}{ccccccc}
\hline
 & \multicolumn{3}{c}{$n = 300$} & \multicolumn{3}{c}{$n = 500$} \\
$1-\alpha$ & $d = 50$  & $d = 75$ & $d = 100$ & $d = 50$ & $d = 75$ & $d = 100$ \\ \hline
90\% & 0.894            & 0.898  & 0.918            & 0.900 & 0.894            & 0.894 \\
95\% & 0.954            & 0.966  & 0.974   & 0.948 & 0.954            & 0.960\\
99\% & 0.998            & 1.000  & 0.998    & 0.994 & 0.996            & 0.996\\\hline
\end{tabular}
\label{tab:size}
\end{table}
\\
\\
\textbf{Under $H_A$.}
We simulate $\{\bm{X}_i\}_{i = 1}^n$ under the alternative models with $\bm{\Gamma} = \omega \cdot \bm{I}_d$, where $\omega \in \{0.95, 0.90, 0.85, 0.80\}$. In theses settings, $\bm{\Sigma} \prec \bm{\Sigma}_0$. As $\omega$ drifting away from $1$, for fixed $n$ and $d$, $\bm{\Sigma}$ deviates from $\bm{\Sigma}_0$. We set $\alpha = 0.05$ and the number of repetitions for each case be $500$. For each repetition, we record if it is rejected, i.e.~$1-\text{accept}_{n,d}(0.95)$. \Cref{fig:power} summarize the proposition of rejection under $H_A$, i.e.~power, based on the Gaussian multiplier bootstrapped critical value $q_{n,d}^{(boot)}(0.95)$. We can see that all the power curves increase as the alternatives move away from the null. The power increases, as the sample size increases or as the dimensionality decreases. But the trend in dimensionality is less noticeable when the sample size is $500$. 

\begin{figure}[htbp]
\begin{center}
\includegraphics[width=4in]{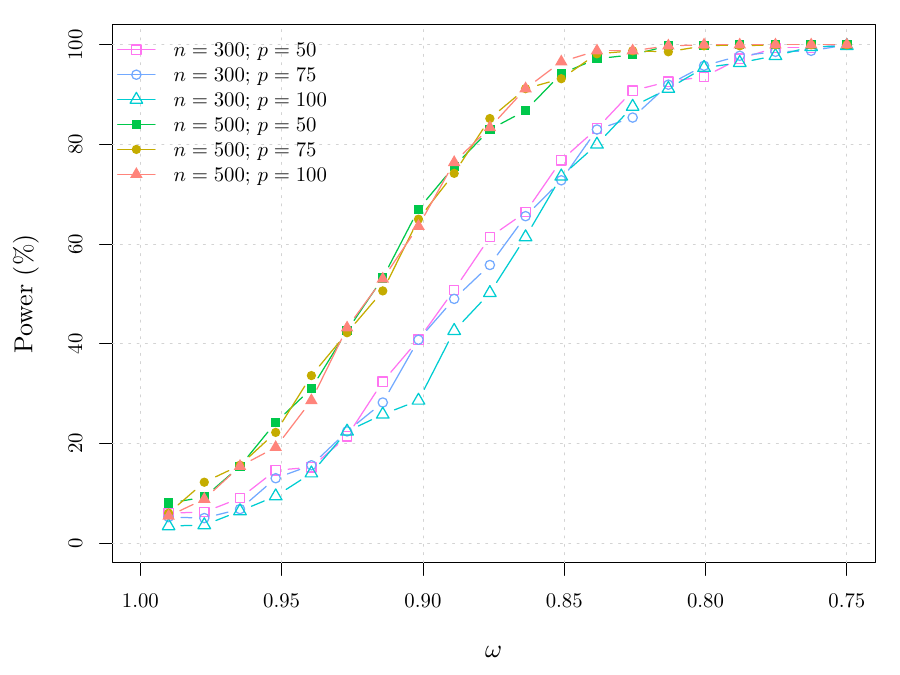}
\caption{Proposition of rejection under $H_A$ based on the Gaussian multiplier bootstrapped critical value.}
\label{fig:power}
\end{center}
\end{figure}

\subsection{Real data application}\label{sec:real:data}
We consider the Federal Reserve Economic Database\footnote{The dataset is publicly available at \url{https://research.stlouisfed.org/econ/mccracken/fred-databases}.} - a monthly data containing over $100$ macroeconomic variables, and aim to apply the proposed tail-robust autocovariance matrix estimator to detect underlying second-order change points. The change point analysis has been performed on the same data by Wang and Zhao \cite{wang2022optimal} and Xu et al. \cite{xu2022change}. More specifically, the hypothesis testing conducted by Wang and Zhao \cite{wang2022optimal} suggests that change points exist in the relationship between the monthly growth rate of the US industrial production index – an important indicator of macroeconomic activity, and other macroeconomic variables. Xu et al. \cite{xu2022change} further estimated and performed statistical inference on those change point locations. Although we consider a different type of change, i.e.~changes in autocovariance among macroeconomic variables, their results provide preliminary evidence for the existence of nonstationarity and structure changes. 

In this study, we consider the time period from January 1998 to December 2022. The data has been pre-processed by the R package \texttt{fbi} \citep{Rfbi} following the suggestions in the FRED-MD website.
The processed data is with dimension $d = 68$ and sample size $n = 300$. 
To detect a lag-$\ell$ autocovariance change point with $\ell\geq 0$, we consider a cumulative sum type statistic that is defined as
\[
T_{\ell}(t) = \sqrt{\frac{(n-t)t}{n}}\left\|\widehat{\bm{\Sigma}}^{[1,t]}_{\ell} - \widehat{\bm{\Sigma}}^{[t+1,n]}_{\ell}\right\|_{\max},
\]
where $\widehat{\bm{\Sigma}}^{[1,t]}_{\ell}$ and $\widehat{\bm{\Sigma}}^{[t+1,n]}_{\ell}$ are lag-$\ell$ autocovariance estimators before and after a time point $t \in [1+\delta, n-\delta]$ and $\delta \in \mathbb{Z}^+$ is a boundary removal parameter to avoid $t$ and $n-t$ being too small. We search for the location $\widehat{t}$ that minimizes $T_{\ell}(t)$ and consider $\widehat{t}$ as the change point estimator. When there is no change point, $T_{\ell}(t)$ can be seen as a variant of the test statistic $T_{\ell}$  in \Cref{thm:test_size} with unequal weights on each data point. Since \Cref{thm:GA_autocov} allows nonstationarity, we still use our Gaussian multiplier bootstrap to obtain the critical value of $T_{\ell}(t)$ for each $t$.

We compute the cumulative sum statistic $T_{\ell}(t)$ based on the element-wise truncated autocovariance and the sample autocovariance estimators. For our element-wise truncated autocovariance estimator and the associated Gaussian multiplier bootstrap, we choose the robustification parameter $\tau = 3.2\, (n^{16/19}/\log d)^{1/3}$ and the block size $S = 10$. The boundary removal parameter $\delta$ is set to be $20$. The left panel of \Cref{fig:realdata} shows the cumulative sum statistics based on our element-wise truncated autocovariance estimator with the $95\%$ critical value obtained by the Gaussian multiplier bootstrap. The detected change point ``June 2020'' is shortly after the outbreak of Covid-19 in the United States, which provides additional evidence that Covid-19 may profoundly affect the U.S. economy in a negative way. The right panel of \Cref{fig:realdata} plots the cumulative sum statistics based on the sample covariance estimator. The cumulative sums based on the robust and sample covariance estimators show similar trends. However, the latter seems to be heavily influenced by the heavy-tailedness, and thus not reliable for change point detection. We also detected the change point using lag $\ell=3$ to reflect the 3-month autocovariance structure change. The results presented in \Cref{fig:realdata_lag3} show that ``June 2020'' is also a significant change point for a 3-month autocovariance matrix change. This reveals the fact that the second-order structure change in the U.S. economy caused by the spread of Covid-19 is not only monthly but also quarterly. To visually illustrate our findings, we provide heat maps depicting the element-wise truncated estimators ($\widehat{\bm{\Sigma}}^{\mathrm{Before}}_{\ell}$ and $\widehat{\bm{\Sigma}}^{\mathrm{After}}_{\ell}$) computed using data collected before and after the estimated change point, as shown in \Cref{fig:realdata_heat0,fig:realdata_heat3}. These figures correspond to the heat maps with $\ell = 0$ and $\ell = 3$, respectively. They visually illustrate the changes in patterns that occurred before and after June 2020.

\begin{figure}[htbp]
\centering
\includegraphics[width=2.3in]{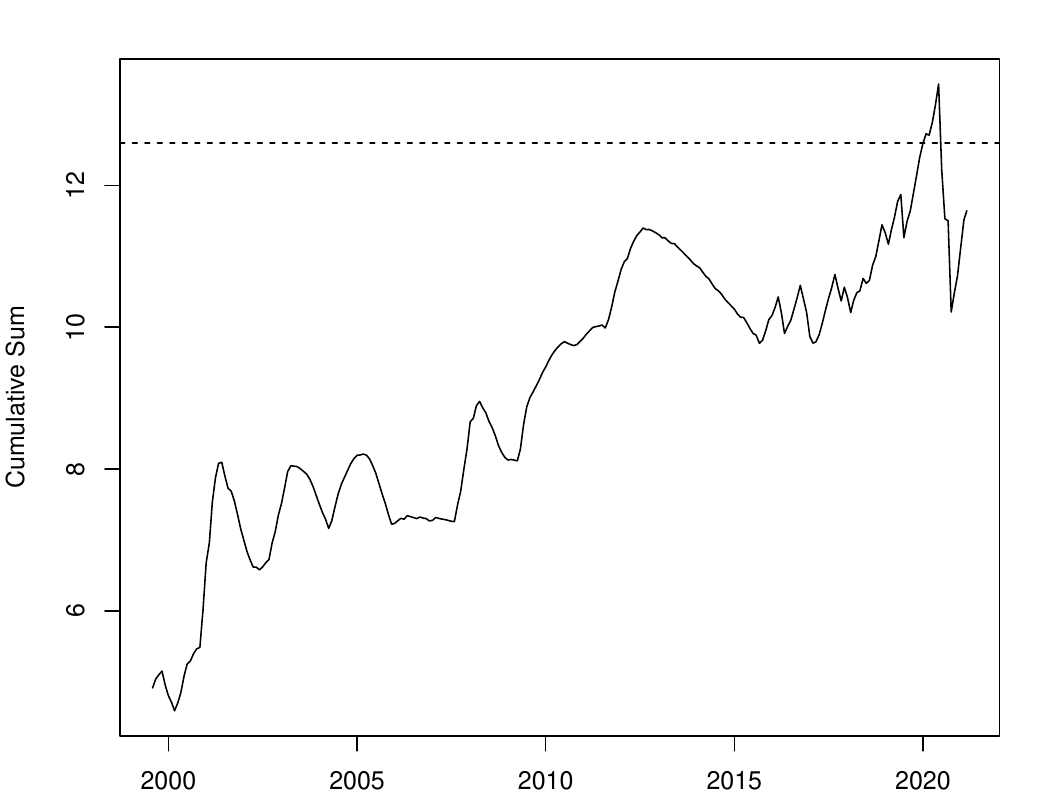}
\includegraphics[width=2.3in]{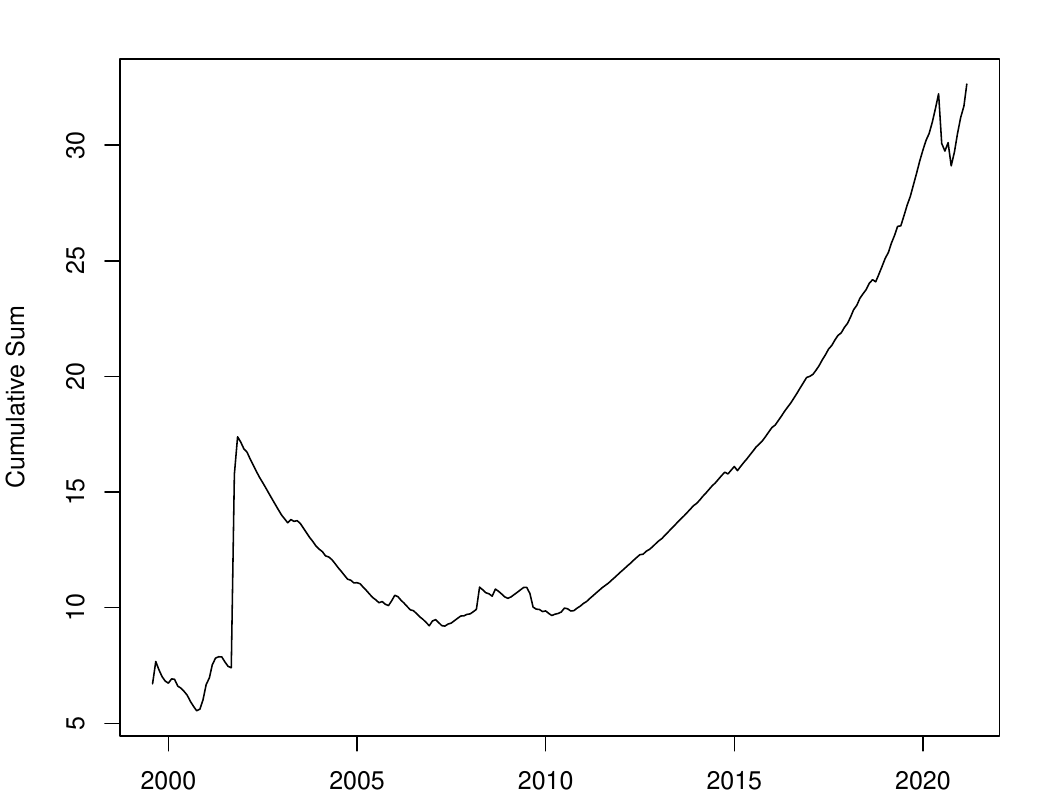}
\caption{Left: the cumulative sum statistics based on the element-wise truncated covariance estimator. The horizontal dashed line represents the $95\%$ critical value obtained from the Gaussian multiplier bootstrap. Right: the cumulative sum statistics based on the sample covariance estimator.}
\label{fig:realdata}
\end{figure}

\begin{landscape}
\begin{figure}[htbp]
\centering
\includegraphics[width=7in]{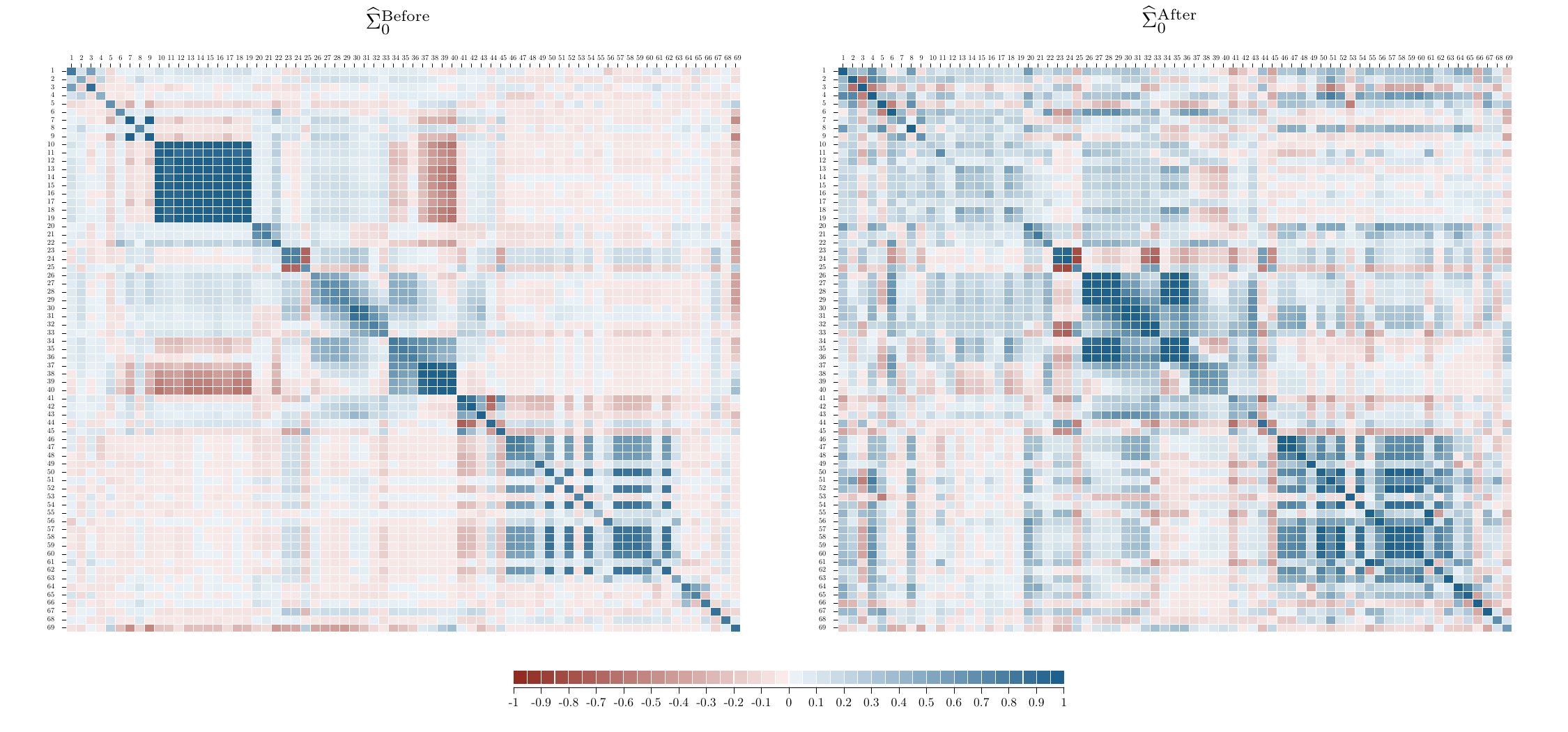}
\caption{Left: the element-wise truncated covariance estimator based on the data before June 2020. Right: the element-wise truncated covariance estimator based on the data after June 2020.}
\label{fig:realdata_heat0}
\end{figure}
\end{landscape}

\begin{figure}[htbp]
\centering
\includegraphics[width=2.3in]{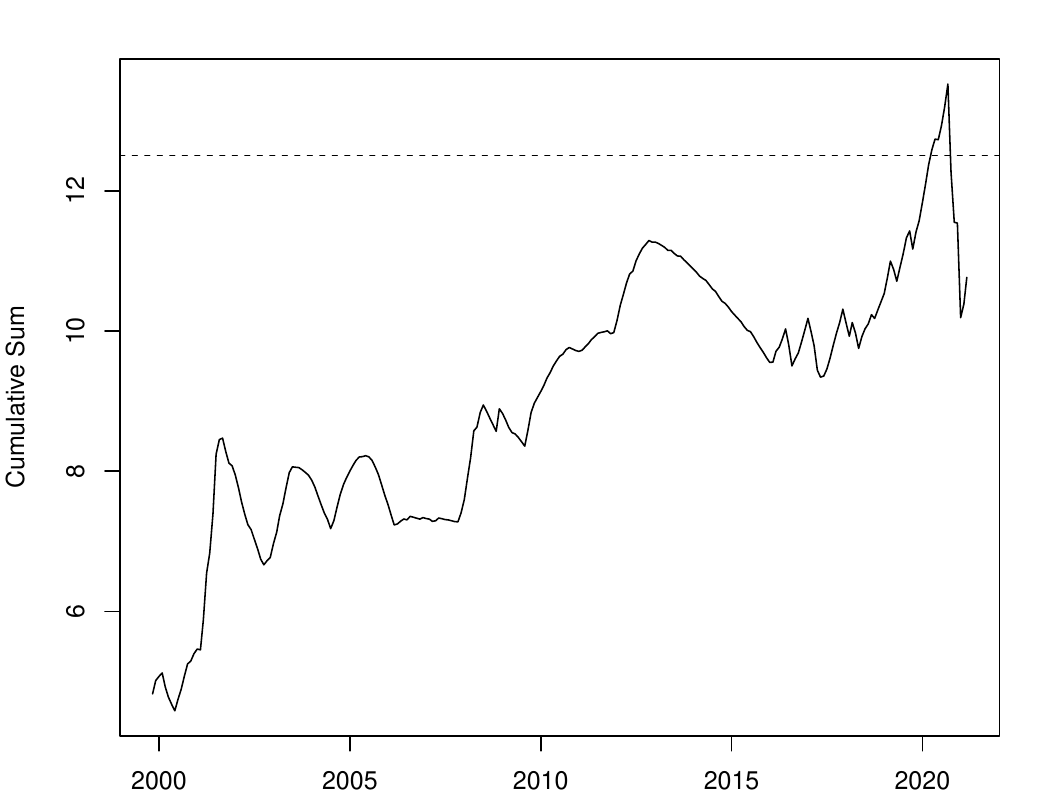}
\includegraphics[width=2.3in]{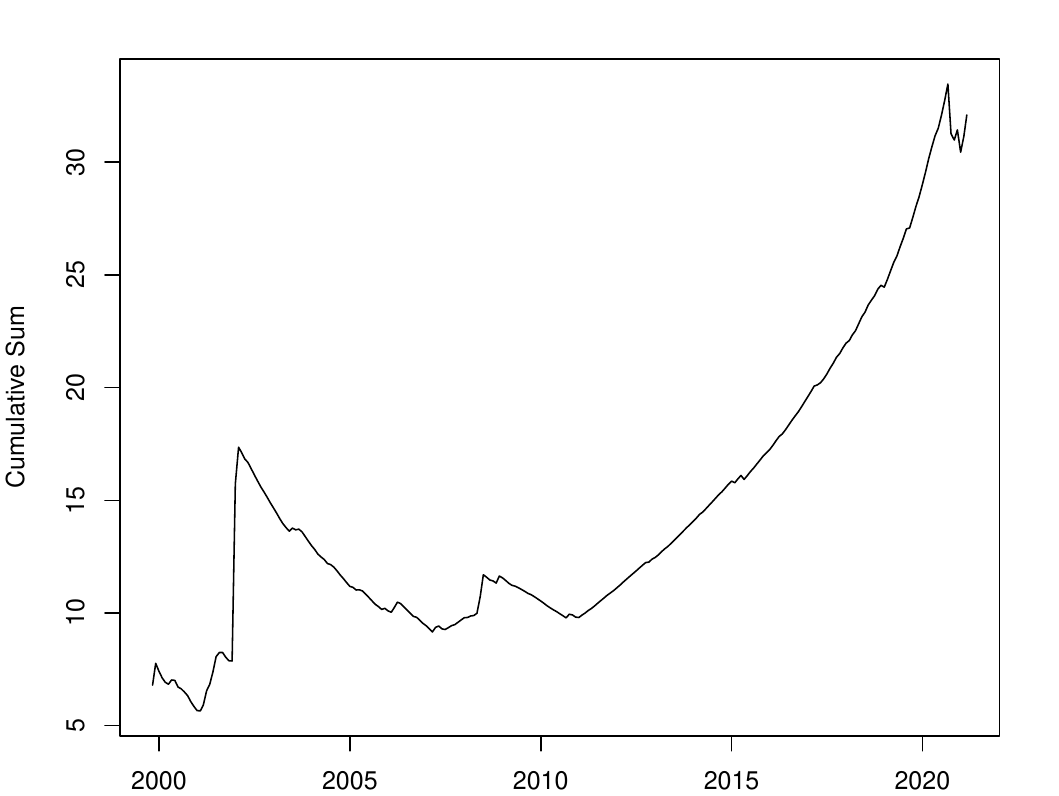}
\caption{Left: the cumulative sum statistics based on the element-wise truncated autocovariance estimator with lag-$3$. The horizontal dashed line represents the $95\%$ critical value obtained from the Gaussian multiplier bootstrap. Right: the cumulative sum statistics based on the sample autocovariance estimator with lag-$3$.}
\label{fig:realdata_lag3}
\end{figure}

\begin{landscape}
\begin{figure}[htbp]
\centering
\includegraphics[width=7in]{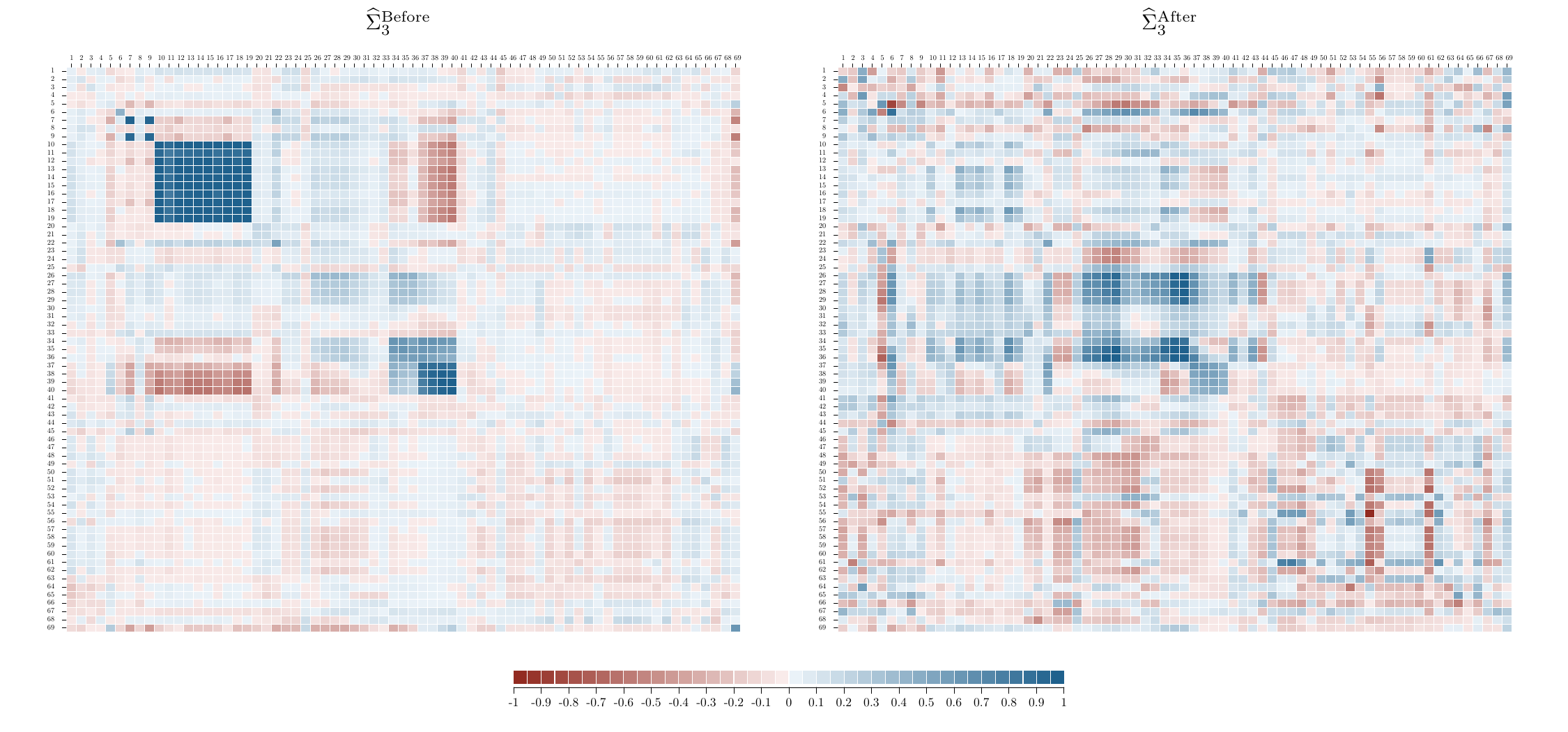}
\caption{Left: the element-wise truncated autocovariance estimator with lag-$3$ based on data before June 2020. Right: the element-wise truncated autocovariance estimator with lag-$3$ based on data after June 2020.}
\label{fig:realdata_heat3}
\end{figure}
\end{landscape}

\section{Conclusion}\label{sec:conclusion}
In this paper, we tackle problems of estimation and inference on autocovariance matrices under heavy-tailedness, high-dimensionality, general nonlinear temporal dependence, and potential nonstationarity. For estimation, we consider two types of tail-robust autocovariance matrix estimation methods: the element-wise Huber's $M$-estimator and a computationally more efficient element-wise truncated estimator. Both estimators are designed to achieve sharp error bounds with respect to the matrix max-norm. The nonasymptotic properties of these estimators are proved based on Bernstein-type inequalities under functional dependence for the potentially nonstationary processes which may be of independent interest. For inference, we focus on the element-wise truncated autocovariance estimator, which is simpler and computational more efficient. We prove a Gaussian approximation result, as a limiting distribution, for our element-wise truncated autocovariance estimator. A Gaussian multiplier bootstrap result is also given to facilitate the practicality. Our theoretical results are nonasymptotic, which give explicit error bounds in terms of sample size, dimensionality, moment, and the strength of temporal dependence. Numerical evidence is provided to support our theoretical results.

\clearpage
\setcounter{section}{0}
\setcounter{equation}{0}
\setcounter{figure}{0}
\setcounter{table}{0}
\setcounter{theorem}{0}
\setcounter{remark}{0}
\setcounter{assumption}{0}
\renewcommand{\thesection}{\Alph{section}}
\renewcommand{\theremark}{\thesection.\arabic{remark}}
\renewcommand{\theassumption}{\thesection.\arabic{assumption}}
\makeatletter
\@addtoreset{remark}{section}
\@addtoreset{assumption}{section}
\makeatother
% Unique hyperlink destinations for the supplement's restarted counters.
\renewcommand{\theHsection}{supp.\Alph{section}}
\renewcommand{\theHequation}{supp.\arabic{equation}}
\renewcommand{\theHfigure}{supp.\arabic{figure}}
\renewcommand{\theHtable}{supp.\arabic{table}}
\renewcommand{\theHtheorem}{supp.\thesection.\arabic{theorem}}
\renewcommand{\theHremark}{supp.\thesection.\arabic{remark}}
\renewcommand{\theHassumption}{supp.\thesection.\arabic{assumption}}
\section*{Supplementary material}

This supplementary material provides additional simulation results and collects all the technical proofs.

\section{Additional simulations}\label{app:additional_simu}
This section gives some additional simulation results of {\bf Scenario 1}. In Figures \ref{fig:var_lag0}-\ref{fig:var_lag2}, we summarize the RMEs of these four tail-robust autocovariance estimators for $\ell = 0$ and $\ell = 2$ and under the three covariance structures respectively. These figures show similar results as in Section 5, which provide additional evidence for the tail-robust properties of these tail-robust estimators. 

\begin{figure}
\begin{center}
\includegraphics[width=5in]{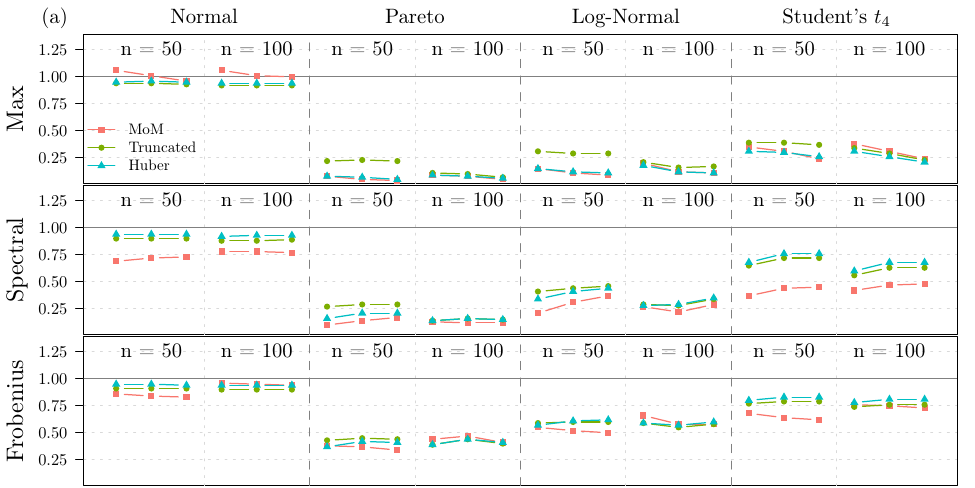}
\includegraphics[width=5in]{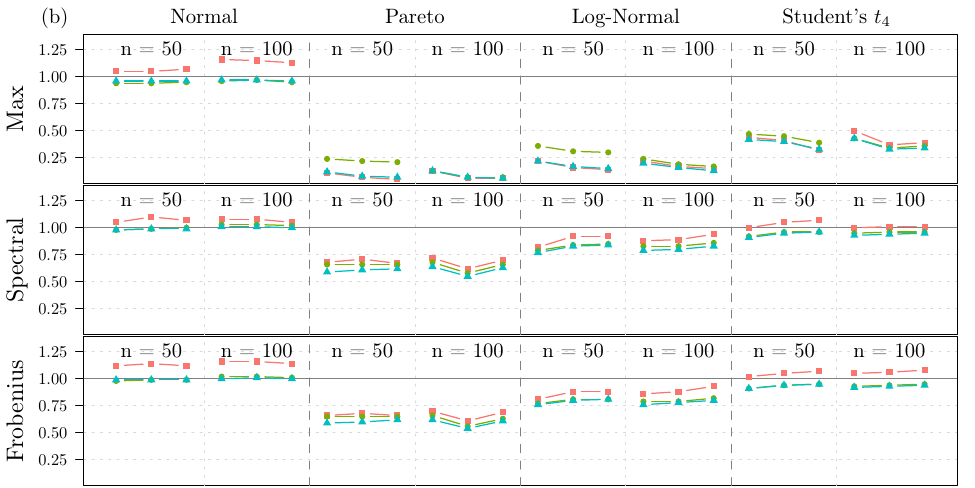}
\includegraphics[width=5in]{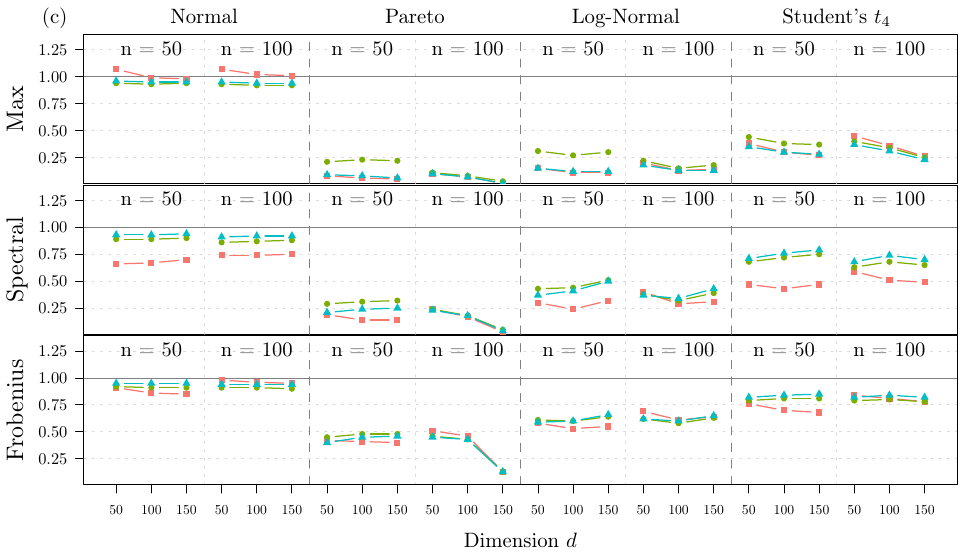}
\caption{RMEs of tail-robust lag-$0$ autocovariance matrix estimators with respective to max, spectral and Frobenius norms. Data are simulated, with $200$ replicates, from VAR($1$) model with $\rho = 0.5$, and innovations follow Normal, Pareto, Log-Normal and Students's $t_4$ distributions.%
Three structures of $\bm{\Gamma}$ (diagonal, equal correlation and power decay) are used respectively in panel (a), (b) and (c).}
\label{fig:var_lag0}
\end{center}
\end{figure}

\begin{figure}
\begin{center}
\includegraphics[width=5in]{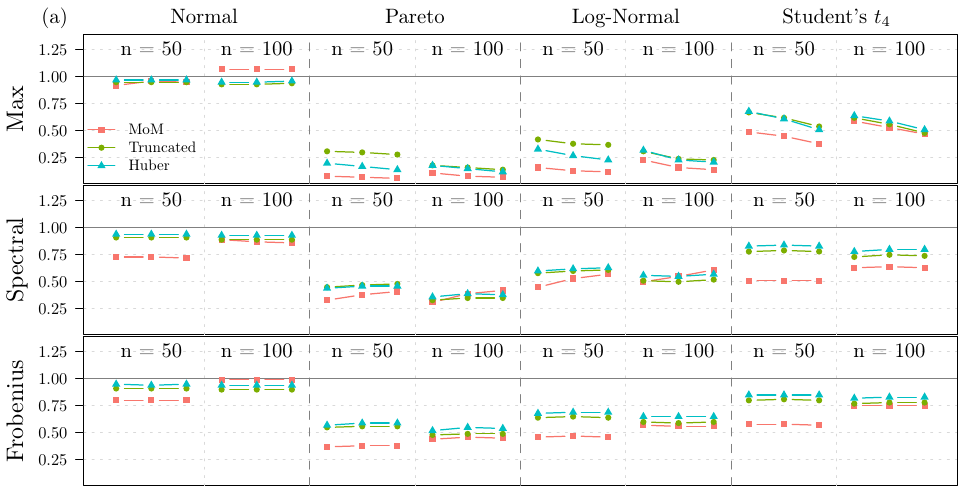}
\includegraphics[width=5in]{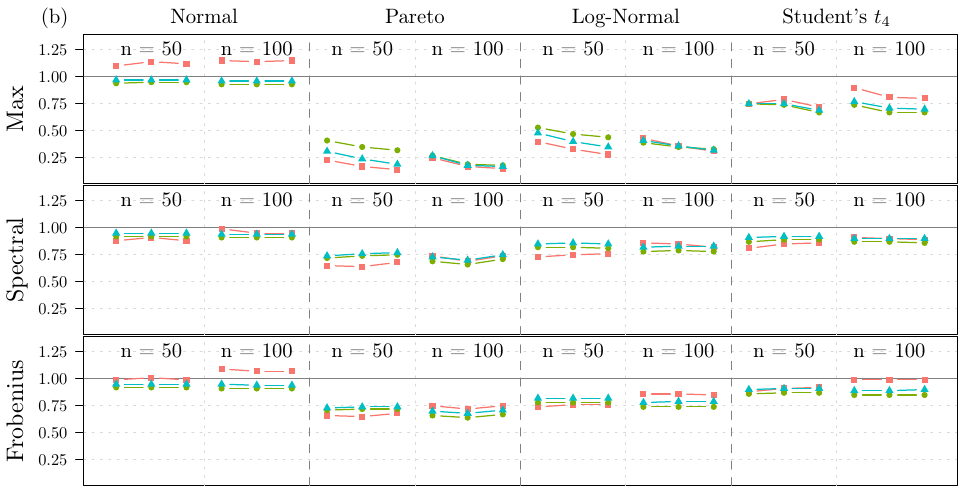}
\includegraphics[width=5in]{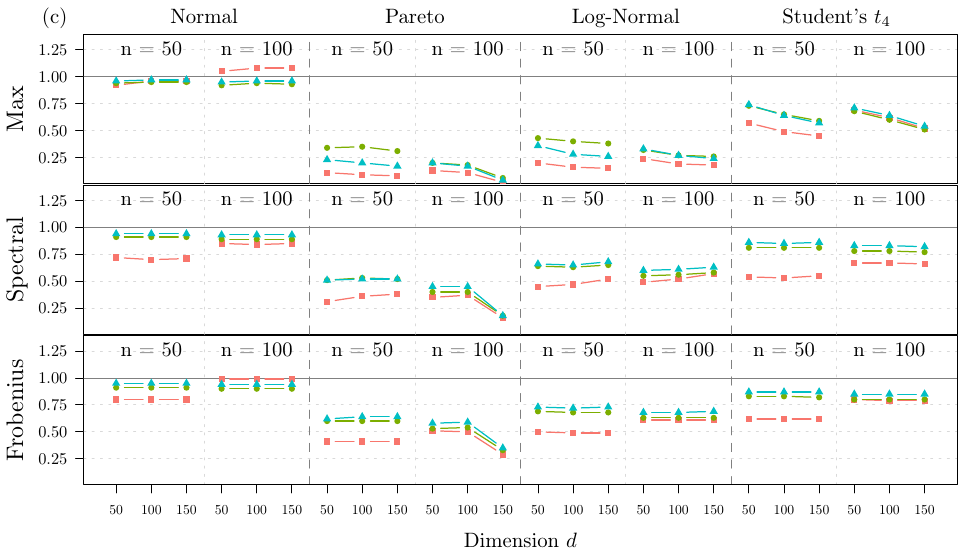}
\caption{RMEs of tail-robust lag-$2$ autocovariance matrix estimators with respective to max, spectral and Frobenius norms. Data are simulated, with $200$ replicates, from VAR($1$) model with $\rho = 0.5$, and innovations follow Normal, Pareto, Log-Normal and Students's $t_4$ distributions.%
Three structures of $\bm{\Gamma}$ (diagonal, equal correlation and power decay) are used respectively in panel (a), (b) and (c).}
\label{fig:var_lag2}
\end{center}
\end{figure}

\section{Auxiliary lemmas}\label{sec:auxiliary_lemma}
In this section, we present several auxiliary lemmas, which are useful for the proofs of our results. Throughout this section, we consider an $\mathbb{R}$-valued process $\{X_i\}_{i \in \mathbb{Z}}$ of the form (1).

The following lemma in \cite{sun2020adaptive} provides a deterministic inequality regarding a convex loss function $\mathcal{L}: \mathbb{R} \mapsto \mathbb{R}$. This inequality allows us to simultaneously control the symmetrized error of a local linear approximation for $\mathcal{L}$ in a neighbour of $\theta_0$. It will be used later to study the Huber's $M$-estimator.
\begin{lemma}[Lemma C.1 in \cite{sun2020adaptive}]\label{lemma:local}
    Let $D_{\mathcal{L}}(\theta_1, \theta_2) = \mathcal{L}(\theta_1) - \mathcal{L}(\theta_2) - \mathcal{L}^{\prime}(\theta_2)(\theta_1 - \theta_2)$ and its symmetrized version $D_{\mathcal{L}}^s(\theta_1, \theta_2) = D_{\mathcal{L}}(\theta_1, \theta_2) + D_{\mathcal{L}}(\theta_2, \theta_1)$. For $\theta_{\eta} = \theta_0 + \eta(\theta - \theta_0)$ with $\eta \in (0,1]$ and any convex function $\mathcal{L}$, we have
    \begin{align*}
        D_{\mathcal{L}}^s(\theta_{\eta}, \theta_0) \leq \eta D_{\mathcal{L}}^s(\theta, \theta_0).
    \end{align*}
\end{lemma}

For $x \in \mathbb{R}$, define the binary random variable $Y_{i}(x) = \mathbbm{1}\{X_{i} \leq x\}$, where $\mathbbm{1}\{\cdot\}$ is an indicator function. For some integer $q \geq 2$, we assume $\Vert X_{i} \Vert_q < \infty$ and the $q$-th order functional dependence measure of $X_{i}$ satisfies certain decay rates. The following lemma shows that the decay rates can be preserved by $\{Y_{i}(x)\}_{i\in \mathbb{Z}}$ uniformly for all $x \in \mathbb{R}$, if the density function of the marginal distribution of $X_i$ is bounded.

\begin{lemma}\label{lemma:indicator_fct}
    Assume that for some integer $q \geq 2$, $\Vert X_i\Vert_q < \infty$, and the marginal distribution of $X_{i}$ is absolutely continuous and has a bounded density function, i.e.~$\sup_{x \in \mathbb{R}}f_{X}(x) < \infty$. Suppose there exists some $c > 0$ such that $\|X_{\cdot}\|_{q} < \infty$. Then, we have for any $x \in \mathbb{R}$ and any $q^{\prime} \geq 2$ that
    \[
    \|Y_{\cdot}(x)\|_{q^{\prime}} = \sup_{m \geq 0} \exp(c^{\prime}m)\sum_{s = m}^{\infty}\delta_{s, q^{\prime}}^{Y(x)} < \infty,
    \]
    where $c^{\prime} = cq\{q^{\prime}(1+q)\}^{-1} > 0$ and $\delta_{s, q^{\prime}}^{Y(x)} = \sup_{s \in \mathbb{Z}}\Vert Y_{s}(x) - Y_{s, \{0\}}(x)\Vert_{q^{\prime}}$.
\end{lemma}
\begin{proof}
    By Lemma 2.1, we have that $\|X_{\cdot}\|_{q} < \infty$ is equivalent to $\delta_{s,q}^X = \sup_{s \in \mathbb{Z}}\Vert X_{s} - X_{s,\{0\}} \Vert_q \leq C\exp(-cs)$. Then, for any $x \in \mathbb{R}$ and $\epsilon > 0$, we have
\begin{align*}
    &\Vert Y_{s}(x) - Y_{s,\{0\}}^{\prime}(x)\Vert_{q^{\prime}}^{q^{\prime}} = \Vert \mathbbm{1}_{\{X_{s} \leq x\}} - \mathbbm{1}_{\{X_{s,\{0\}} \leq x\}}\Vert_{q^{\prime}}^{q^{\prime}} = \Vert \mathbbm{1}_{\{X_{s} \leq x, X_{s, \{0\}} > x\}} - \mathbbm{1}_{\{X_{s} > x, X_{s,\{0\}} \leq x\}}\Vert_{q^{\prime}}^{q^{\prime}}\\
    =& \mathbb{P}(X_{s} \leq x, X_{s, \{0\}} > x, |X_{s} - X_{s,\{0\}}| \leq \epsilon^s) + \mathbb{P}(X_{s} > x, X_{s, \{0\}} \leq x, |X_{s} - X_{s,\{0\}}| \leq \epsilon^s)\\
    &+ \mathbb{P}(X_{s} \leq x, X_{s,\{0\}} > x, |X_{s} - X_{s, \{0\}}| > \epsilon^s) + \mathbb{P}(X_{s} > x, X_{s, \{0\}} \leq x, |X_{s} - X_{s, \{0\}}| > \epsilon^s)\\
    \leq& 2\mathbb{P}(x \leq X_{s} \leq x+\epsilon^s) + \mathbb{P}(|X_{s} - X_{s,\{0\}}| > \epsilon^s) \leq 2\epsilon^s\sup_{x}f_X(x) + C\exp(-cqs)/\epsilon^{sq}\\
    \leq& C_1\exp\left(-\frac{qcs}{1+q}\right),
\end{align*}
where $C_1 > 0$ is an absolute constant.
Therefore, the $q^{\prime}$-th order functional dependence measure of $\{Y_{i}(x)\}_{i \in \mathbb{Z}}$ is 
\begin{equation*}
    \begin{aligned}
    \delta_{s,q^{\prime}}^{Y(x)} = \sup_{s \in \mathbb{Z}}\Vert Y_{s}(x) - Y_{s, \{0\}}(x)\Vert_{q^{\prime}} \leq C_2\exp\left(-\frac{qcs}{q^{\prime}(1+q)}\right).
    \end{aligned}
\end{equation*}
Let $c^{\prime} = cq\{q^{\prime}(1+q)\}^{-1}$, using Theorem~2.1 again concludes the proof.
\end{proof}

In the paper, we frequently use the $L_q$ norm of partial sum $S_n = \sum_{i = 1}^{n}X_i$. We introduce Burkholder's inequality, which can be combined with the martingale decomposition technique to deliver an upper bound of $\Vert S_n \Vert_q$.
\begin{lemma}[Burkholder's inequality \cite{rio2009moment}]\label{lemma:Burkholder}
    Let $q > 1$, $q^{\prime} = \min(q, 2)$. Let $M_n = \sum_{i = 1}^n D_i$, where $\{D_i\}_{i = 1}^n$ are martingale differences, such that $\Vert D_i \Vert_q < \infty$. Then, we have
    \begin{equation*}
        \Vert M_n \Vert_q^{q^{\prime}} \leq K_q^{q^{\prime}}\sum_{i = 1}^n\Vert D_i \Vert_q^{q^{\prime}}, \; \text{ where } K_q = \max( (q-1)^{-1}, \sqrt{q-1}).
    \end{equation*}
\end{lemma}
\noindent Burkholder's inequality considers the sum of martingale difference sequence which is not the case for $S_n$. However, we can construct a martingale difference sequence by rewriting each summand as:
\begin{equation}\label{eq:projection}
    X_i = \sum_{k = 0}^{\infty}\mathcal{P}_{i-k}X_i,
\end{equation}
where $\mathcal{P}_{k}\cdot = \mathbb{E}(\cdot|\mathcal{F}_k) - \mathbb{E}(\cdot|\mathcal{F}_{k-1})$ and $\mathcal{F}_k = (\dots, \epsilon_{k-1}, \epsilon_k)$. By construction, $\{\mathcal{P}_{i-k}X_i\}_{i \in \mathbb{Z}}$ is a martingale difference sequence. For $q \geq 2$, by Burkholder's inequality, we have
\begin{equation*}
    \bigg\Vert \sum_{i = 1}^n\mathcal{P}_{i-k}X_i \bigg\Vert_q \leq (q-1)^{1/2}n^{1/2}\Vert \mathcal{P}_{0}X_k \Vert_q \leq (q-1)^{1/2}n^{1/2}\delta_{k,q},
\end{equation*}
and
\begin{equation}\label{eq:moment_ineq1}
    \Vert S_n \Vert_q = \bigg\Vert \sum_{i = 1}^n\sum_{k = 0}^{\infty}\mathcal{P}_{i-k}X_i \bigg\Vert_q \leq \sum_{k = 0}^{\infty}\bigg\Vert \sum_{i = 1}^n\mathcal{P}_{i-k}X_i \bigg\Vert_q \leq (q-1)^{1/2}n^{1/2}\sum_{k = 0}^{\infty}\delta_{k,q} .
\end{equation}
For example, by the moment inequality \eqref{eq:moment_ineq1} with $q=2$, we can bound the long-run covariance 
\begin{equation}\label{eq:LRV}
    \sigma_{\infty} = \lim_{n \to \infty}\var(n^{-1/2}S_n) \leq \sum_{k = 0}^{\infty} \delta_{k,2}.
\end{equation}

If the quantity of interest is $S_n^* = \max_{1 \leq i \leq n}| S_i |$, and with $\mathbb{E}X_i = 0$, Theorem 1 in \cite{wu2007strong} provides the following maximal inequality. This result is based on Doob's inequality in addition to the same martingale decomposition \eqref{eq:projection} and Burkholder's inequality.
\begin{lemma}[Theorem 1 in \citealp{wu2007strong}]\label{lemma:moment_maximal}
Assume $\mathbb{E}X_i = 0$, $\Vert X_i \Vert_q < \infty$ and $\sum_{k = 0}^{\infty} \delta_{k,q} < \infty$ for some $q \geq 2$, then, we have
    \begin{equation}\label{eq:max_moment_ineq1}
    \Vert S_n^* \Vert_q \leq \frac{qB_q}{q-1}n^{1/2}\sum_{k = 0}^{\infty} \delta_{k,q}, 
\end{equation}
where $B_q = 18q^{3/2}(q-1)^{-1/2}$ if $q > 2$ and $B_q = 1$ if $q = 2$.
\end{lemma}
\noindent Note that the moment inequality \eqref{eq:moment_ineq1} and the maximal inequality \eqref{eq:max_moment_ineq1} have the upper bounds being equivalent up to a constant.

The following lemma provides an exponential tail probability bound for a nonnegative random variable. This result is a special case of Lemma 7.3 in \cite{chen2016self} when $n = 1$.
\begin{lemma}\label{lemma:lb_nonnegative}
    Let $U$ be a nonnegative random variable with $\mathbb{E}U^2 < \infty$. Then, for any $0 < u < \mathbb{E}U$, we have
    \begin{equation*}
        \mathbb{P}(U - \mathbb{E}U \leq -u) \leq \exp\Big(-\frac{u^2}{4\mathbb{E}U^2} \Big).
    \end{equation*}
\end{lemma}
\begin{proof}
    For $t > 0$, we have by Markov's inequality
    \begin{equation*}
    \begin{aligned}
        \mathbb{P}(U - \mathbb{E}U \leq -u) &\leq e^{-tu+t\mathbb{E}U}\mathbb{E}e^{-tU} \leq e^{-tu+t\mathbb{E}U}\big(1 - t\mathbb{E}U + t^2\mathbb{E}U^2\big)\\
        &\leq e^{-tu+t^2\mathbb{E}U^2} \leq \exp\Big(-\frac{u^2}{4\mathbb{E}U^2} \Big),
    \end{aligned}
    \end{equation*}
    where the second inequality is due to inequality $e^{-x} \leq 1 - x + x^2$ for $x \geq 0$, the third inequality is due to the inequality $1 + x \leq e^x$, and the last inequality is obtained by letting $t = u/(2\mathbb{E}U^2)$.
\end{proof}

Given two dependence processes and their functional dependence measures, the following lemma provides the functional dependence measure of the product process. 
\begin{lemma}\label{lemma:fdm_XY}
    Let $\{X_i\}_{i \in \mathbb{Z}} \subset \mathbb{R}$ and $\{Y_i\}_{i \in \mathbb{Z}} \subset \mathbb{R}$ be two processes in the form of \textup{(1)}. Assume that $\sup_{i \in \mathbb{Z}}\|X_i\|_4 < \infty$ and $\sup_{i \in \mathbb{Z}}\|Y_i\|_4 < \infty$, denote $\{\delta_{s,2}^{XY}\}_{s \geq 0}$ the functional dependence measures of $\{X_iY_i\}_{i \in \mathbb{Z}}$. We have that for any $s \geq 0$
    \begin{equation*}
        \delta_{s,2}^{XY} \leq 2\max\{\sup_{i \in \mathbb{Z}}\|X_i\|_4, \sup_{i \in \mathbb{Z}}\|Y_i\|_4\}\max\{\delta_{s,4}^X, \delta_{s,4}^Y\},
    \end{equation*}
    where $\{\delta_{s,4}^X\}_{s \geq 0}$ and $\{\delta_{s,4}^Y\}_{s \geq 0}$ are respectively the functional dependence measures of  $\{X_i\}_{i \in \mathbb{Z}}$ and $\{Y_i\}_{i \in \mathbb{Z}}$.
\end{lemma}

\section{Nonasymptotic theory for tail-robust mean estimators}\label{sec:supp_robust_mean}
For simplicity, we consider only the stationary processes. We note that our results can be extended to potentially nonstationary processes with the target quantity being modified as an averaged mean among time points.
Let $\{\bm{X}_i\}_{i \in \mathbb{Z}} \subset \mathbb{R}^d$ be a stationary process with mean $\bm{\mu} = (\mu_1, \mu_2, \dots, \mu_d)^{\intercal}$ in the form of \textup{(1)}. In this section, we consider two element-wise tail-robust estimators for $\bm{\mu}$: (a) Huber's $M$-estimatior and (b) truncated estimator. Their nonasymptotic results are provided and will be used as building blocks to prove the nonasymptotic results for corresponding tail-robust estimators of autocovariance matrices. The nonasymptotic results only require finite second moments and allow for temporal dependence. Formally, we list the following assumptions.
\begin{assumption}\label{assumption:APP_finite.variance}
    $\omega_2^2 = \max_{j \in [d]}\mathbb{E}[X_{i,j}^2] < \infty$.
\end{assumption}
\begin{assumption}\label{assumption:APP_decay.2order}
    There exists some $c > 0$, such that $\Vert X_{\cdot} \Vert_2 = \sup_{m \geq 0}\exp(cm)\sum_{s = m}^{\infty}\delta_{s,2} < \infty$.
\end{assumption}
Note that Assumptions \ref{assumption:APP_finite.variance} and \ref{assumption:APP_decay.2order} are stated in the second moments and are implied by the corresponding assumptions in the fourth moments, i.e.~Assumptions 1 and 2. The following lemma shows this property of the functional depnendence measure.
\begin{lemma}
    \label{lemma:DAN_prop}
    Let $q_1, q_2 > 0$ such that $q_2 \in (0, q_1)$.
    Suppose there exists some $c > 0$ such that $\|\bm{X}_{\cdot}\|_{q_1} < \infty$, then we have $\|\bm{X}_{\cdot}\|_{q_2} < \infty$, for the same $c$.
\end{lemma}
\begin{proof}[Proof of Lemma \ref{lemma:DAN_prop}]
    Since $q_2 \in (0, q_1)$, by H\"older's inequality for any $j \in [d]$
    \begin{equation*}
    \delta_{i,q_2,j}^{q_2} = \mathbb{E}|X_{i,j} - X_{i,\{0\},j}|^{q_2} \leq \big(\mathbb{E}|X_{i,j} - X_{i,\{0\},j}|^{q_1}\big)^{q_2/q_1} = \delta_{i,q_1,j}^{q_2}, \; \text{ for all } \; i \geq 0.
    \end{equation*}
    Therefore, we have $\sup_{m \geq 0} \exp(-cm)\Delta_{m,q_2,j} \leq \sup_{m \geq 0} \exp(-cm)\Delta_{m,q_1,j} < \infty$.
\end{proof}

\subsection{Huber's $M$-estimator for high-dimensional mean}
Denote the element-wise Huber's $M$-estimator of $\bm{\mu}$ by
    \begin{equation*}
        \widetilde{\bm{\mu}} = (\widetilde{\mu}_1, \dots, \widetilde{\mu}_d)^{\intercal},
    \end{equation*}
    where $\widetilde{\mu}_j = \argmin_{u \in \mathbb{R}}n^{-1}\sum_{i=1}^n\ell_{\tau}(X_{i,j} - u)$ for $j \in [d]$.
In this subsection, we show that $\widetilde{\bm{\mu}}$ is a tail-robust estimator and is optimal (up to a $\log n$ factor) in the minimax sense, under some mild assumptions. These assumptions include Assumptions \ref{assumption:APP_finite.variance} and \ref{assumption:APP_decay.2order}, and, in addition, the bounded density assumption, i.e.~Assumption~3.

\begin{theorem}\label{Thm:m-est.mean}
    For $t > 0$, provided $n$ is large enough such that
    \begin{equation*}
    n \geq C(\log n)^2(t + \log d),
    \end{equation*}
    with $C > 0$ being a sufficiently large constant.
    Choose 
    \begin{equation*}
        \tau \asymp \omega_2(\log n)^{-1} \sqrt{\frac{n}{t+\log d}}.
    \end{equation*}
    Then, under Assumptions 3, \ref{assumption:APP_finite.variance} and \ref{assumption:APP_decay.2order}, we have with probability at least $1 - 3e^{-t}$,
    \begin{equation*}
        |\widetilde{\bm{\mu}} - \bm{\mu}|_{\infty} \lesssim \Vert X_{\cdot} \Vert_{2}(\log n)\sqrt{\frac{t+\log d}{n}}.
    \end{equation*}
\end{theorem}

\begin{proof}
It suffices to consider only $\{X_{i,j}\}_{i \in \mathbb{Z}}$ for any $j \in [d]$. Let $r > 0$ be a constant. Consider with $\eta \in (0,1]$ an intermediate estimator 
\[
\widetilde{\mu}_{j}^{\eta} = \mu_j + \eta(\widetilde{\mu}_j - \mu_j),
\]
such that $|\widetilde{\mu}^{\eta}_j - \mu_j| \leq r$. Note that if $|\widetilde{\mu}_j - \mu_j| \leq r$, we let $\eta = 1$ and $\widetilde{\mu}^{\eta}_j = \widetilde{\mu}_j$. The proof consists of the following four steps.
\\
\\
\textbf{Step 1: Bounding the intermediate estimator by local linear approximation}.\\
Let $\mathcal{L}(u) = n^{-1}\sum_{i =1}^n\ell_{\tau}(X_{i,j} - u)$. By definition, $\mathcal{L}^{\prime}(\widetilde{\mu}_j) = -n^{-1}\sum_{i =1}^n\psi_{\tau}(X_{i,j} - \widetilde{\mu}_j) = 0$. Applying Lemma \ref{lemma:local}, we have that
    \begin{equation*}
        \left\{\mathcal{L}^{\prime}(\widetilde{\mu}^{\eta}_j) - \mathcal{L}^{\prime}(\mu_j)\right\}(\widetilde{\mu}^{\eta}_j - \mu_j) \leq \eta\left\{\mathcal{L}^{\prime}(\widetilde{\mu}_j) - \mathcal{L}^{\prime}(\mu_j)\right\}(\widetilde{\mu}_j - \mu_j) = -\eta\mathcal{L}^{\prime}(\mu_j)(\widetilde{\mu}_j - \mu_j).
    \end{equation*}
Applying the mean-value theorem to the left-hand side of the above equation, we have that
    \begin{equation*}
        \mathcal{L}^{\prime\prime}(\mu^{\eta\,*}_j)(\widetilde{\mu}^{\eta}_j - \mu_j)^2 \leq -\eta\mathcal{L}^{\prime}(\mu_j)(\widetilde{\mu}_j - \mu_j) = -\mathcal{L}^{\prime}(\mu_j)(\widetilde{\mu}^{\eta}_j - \mu_j),
    \end{equation*}
where $\mu^{\eta \, *}_j$ is on the line segment between $\mu_j$ and $\widetilde{\mu}^{\eta}_j$, thus $|\mu^{\eta\,*}_j - \mu_j| \leq r$. If there exists a constant $D > 0$, such that $\min_{|u - \mu_j| \leq r}\mathcal{L}^{\prime\prime}(u) \geq D$, then, we have that
    \begin{equation}\label{eq:m-est_intermediate}
        |\widetilde{\mu}^{\eta}_j - \mu_j| \leq D^{-1}|\mathcal{L}^{\prime}(\mu_j)| = D^{-1}\left| \frac{1}{n}\sum_{i =1}^n\psi_{\tau}(X_{i,j} - \mu_j) \right|.
    \end{equation}
\\
\\
\textbf{Step 2: Deriving the lower bound of $\mathcal{L}^{\prime\prime}(u)$} for $u$ in the neighbour of $\mu_j$.\\
Let $\tau = 2r$, then we have $|u - \mu_j| \leq r = \tau/2$. For any $u$ such that $|u - \mu_j| \leq \tau/2$, we have that
\begin{equation*}
    \begin{aligned}
    \mathcal{L}^{\prime\prime}(u) =& \frac{1}{n}\sum_{i =1}^n\mathbbm{1}_{\{|X_{i,j} - u| \leq \tau/2\}} \geq 1 - \frac{1}{n}\sum_{i =1}^n\mathbbm{1}_{\{|X_{i,j} - \mu_j| > \tau/2\}} - \mathbbm{1}_{\{|\mu_j - u| > \tau/2\}}\\
    =& 1 - \left\{\frac{1}{n}\sum_{i =1}^n\mathbbm{1}_{\{|X_{i,j} - \mu_j| > \tau/2\}} - \mathbb{P}(|X_{i,j} - \mu_j| > \tau/2)\right\} - \mathbb{P}(|X_{i,j} - \mu_j| > \tau/2)\\
    \geq& 1 - 4\sigma_j^2/\tau^2- \left\{\frac{1}{n}\sum_{i =1}^n\mathbbm{1}_{\{|X_{i,j} - \mu_j| > \tau/2\}} - \mathbb{P}(|X_{i,j} - \mu_j| > \tau/2)\right\}\\
        =& 1 - 4\sigma_j^2/\tau^2 + \left\{\frac{1}{n}\sum_{i =1}^n\mathbbm{1}_{\{|X_{i,j} - \mu_j| \leq \tau/2\}} - \mathbb{P}(|X_{i,j} - \mu_j| \leq \tau/2)\right\}.
    \end{aligned}
\end{equation*}
To bound the third term on the right-hand side of the above equation, we apply Lemma \ref{lemma:lb_nonnegative} with $U = n^{-1}\sum_{i=1}^n\mathbbm{1}_{\{|X_{i,j} - \mu_j| \leq \tau/2\}}$. To upper $\mathbb{E}U^2$, we consider the dependence measure of the process $\{Y_{i,j}(x)\}_{i\in \mathbb{Z}}$ for $x \in \mathbb{R}$, where $Y_{i,j}(x) = \mathbbm{1}_{\{X_{i,j} - \mu_j \leq x\}}$. We have the functional dependence measure of $|X_{i,j} - \mu_j|$ as 
\begin{equation*}
    \big\Vert |X_{i,j} - \mu_j| - |X_{i,\{0\},j} - \mu_j| \big\Vert_2 \leq \Vert X_{i,j} - X_{i,\{0\},j} \Vert_2 = \delta_{i,2,j} \leq \delta_{i,2},
\end{equation*}
By applying Lemma \ref{lemma:indicator_fct}, we have that under Assumptions 3 and \ref{assumption:APP_decay.2order}, we have that
\begin{equation*}
    \Vert Y_{\cdot}(x) \Vert_2 = \sup_{m \geq 0}\exp(c^{\prime}m)\delta_{i,2,j}^{Y(x)} < \infty,
\end{equation*}
where $c^{\prime} = c/3 > 0$. Then, by Lemma \ref{lemma:moment_maximal}, we have that
\begin{equation*}
    \mathbb{E}U^2 = \frac{1}{n^2}\mathbb{E}\left[\left(\sum_{i=1}^n\mathbbm{1}_{\{|X_{i,j} - \mu_j| \leq \tau/2\}}\right)^2\right] \leq \frac{2}{n}\Vert Y_{\cdot}(x) \Vert_2^2.
\end{equation*}
By Lemma \ref{lemma:lb_nonnegative}, we have that with probability at least $1-e^{-t}$
\begin{equation*}
\begin{aligned}
    \frac{1}{n}\sum_{i=1}^n \left\{\mathbbm{1}_{\{|X_{i,j} - \mu_j| \leq \tau/2\}} - \mathbb{P}(|X_{i,j} - \mu_j| \leq \tau/2)\right\} \geq -2\sqrt{\mathbb{E}U^2t}\geq -2\sqrt{2}\Vert Y_{\cdot}(x) \Vert_2 n^{-1/2}t^{1/2},
\end{aligned}
\end{equation*}
where due to the condition of Lemma \ref{lemma:lb_nonnegative}, we require $t > 0$ and $n$ is large enough such that 
\begin{equation}\label{eq:m-est_nonneg_cond}
    2\sqrt{2}\Vert Y_{\cdot}(x) \Vert_2 n^{-1/2}t^{1/2} < \mathbb{P}(|X_{i,j} -\mu_j| \leq \tau/2).
\end{equation}
Therefore, we have, for $|u - \mu_j| \leq \tau/2$, with probability at least $1-e^{-t}$
\begin{equation}\label{eq:m-est_loss_2ed_lb}
\begin{aligned}
    \mathcal{L}^{\prime\prime}(u) \geq 1 - 4\sigma_j^2/\tau^2 -2\sqrt{2}\Vert Y_{\cdot}(x) \Vert_2 n^{-1/2}t^{1/2}.
\end{aligned}
\end{equation}
\\
\\
\textbf{Step 3: Bounding the deviation of $n^{-1}\sum_{i =1}^n\psi_{\tau}(X_{i,j} - \mu_j)$}.\\
Under Assumption \ref{assumption:APP_decay.2order}, we apply Theorem 2.5.
For any $t > 0$, we have that with probability at least $1 - 2e^{-t}$
    \begin{equation*}
        \left|\frac{1}{n}\sum_{i =1}^n\psi_{\tau}(X_{i,j} - \mu_j) - \mathbb{E}\psi_{\tau}(X_{i,j} - \mu_j)\right| \leq \sqrt{C_2}\|X_{\cdot}\|_2\sqrt{\frac{t}{n}} + \sqrt{C_2}\frac{\tau\sqrt{t}}{n} + C_1\tau \frac{(\log n)^2t}{n},
    \end{equation*}
where $C_1, C_2 > 0$ are absolute constants.
Due to Theorem~2.2, the bias term can be bounded by
    \begin{equation*}
        |\mathbb{E}\psi_{\tau}(X_{i,j} - \mu_j)| \leq \sigma_j^2/\tau,
    \end{equation*}
Then, we have that
\begin{equation}\label{eq:m-est_linear.app_ub}
    \begin{aligned}
        \left|\frac{1}{n}\sum_{i =1}^n\psi_{\tau}(X_{i,j} - \mu_j)\right| \leq& \left|\frac{1}{n}\sum_{i =1}^n\psi_{\tau}(X_{i,j} - \mu_j) - \mathbb{E}\psi_{\tau}(X_{i,j} - \mu_j)\right| + |\mathbb{E}\psi_{\tau}(X_{i,j} - \mu_j)|\\
        \leq& \sqrt{C_2}\|X_{\cdot}\|_2\sqrt{\frac{t}{n}} + 2\sigma_j\sqrt{\frac{\sqrt{C_2t} + C_1(\log n)^2t}{n}},
    \end{aligned}
\end{equation}
where the last inequality follows by setting
\begin{equation}\label{eq:tau_value}
    \tau = \sigma_{j}\sqrt{\frac{n}{C_1(\log n)^2t + C_2^{1/2}t^{1/2}}}.
\end{equation}
\\
\\
\textbf{Step 4: combine the previous steps}.\\
Plugging $\tau$ into (\ref{eq:m-est_loss_2ed_lb}), we have with probability at least $1-e^{-t}$
\begin{equation}\label{eq:m-est_loss_2ed_lb2}
    \begin{aligned}
    \mathcal{L}^{\prime\prime}(u) \geq 1 - 16n^{-1}\big[C_2(\log n)^2t + C_1^{1/2}t^{1/2}\big] -2\sqrt{2}\Vert Y_{\cdot}(x) \Vert_2 n^{-1/2}t^{1/2}.
    \end{aligned}
\end{equation}
Therefore, with the same $\tau$, the requirement (\ref{eq:m-est_nonneg_cond}) for $n$ reduces to
\begin{equation}\label{eq:m-est_nonneg_cond2}
    \mathbb{P}\big(|X_{i,j} - \mu_j|/\sigma_j \leq 4^{-1}n^{1/2}[C_2(\log n)^2t + C_1^{1/2}t^{1/2}]^{-1/2}\big) > 2\sqrt{2}\Vert Y_{\cdot}(x) \Vert_2 n^{-1/2}t^{1/2}.
\end{equation}
Let $C_3 > 0$ be a sufficient large constant. Provided $n$ is large enough such that
\begin{equation}\label{eq:AH_mean_n1}
    n \geq C_3\max\big\{C_2(\log n)^2t + C_1^{1/2}t^{1/2}, 2\Vert Y_{\cdot}(x) \Vert_2^2t\big\},
\end{equation}
then $\mathcal{L}^{\prime\prime}(u) \geq 1 - 16C_3^{-1} - 2C_3^{-1/2}$. Then, we have that (\ref{eq:m-est_nonneg_cond2}) can be simplified as
$$\mathbb{P}(|X_{i,j} - \mu_j|/\sigma_j \leq 4^{-1}C_3^{1/2}) > 2C_3^{-1/2},$$
which holds trivially since $C_3$ is sufficiently large.
Combining (\ref{eq:m-est_intermediate}) and (\ref{eq:m-est_linear.app_ub}), we have with probability at least $1 - 3e^{-t}$
    \begin{align}\label{eq:huber_upper_bd}
        |\widetilde{\mu}^{\eta}_j - \mu_j| \leq& D^{-1}\Big| \frac{1}{n}\sum_{i =1}^n\psi_{\tau}(X_{i,j} - \mu_j) \Big| \nonumber\\ 
        \leq& (1 - 16C_3^{-1} - 2C_3^{-1/2})^{-1}\left\{\sqrt{C_2}\|X_{\cdot}\|_2\sqrt{\frac{t}{n}} + 2\sigma_j\sqrt{\frac{\sqrt{C_2t} + C_1(\log n)^2t}{n}}\right\}\nonumber\\
        \leq& C_4\|X_{\cdot}\|_2(\log n)\sqrt{\frac{t}{n}}.
    \end{align}
In addition, provided $C_3$ is sufficiently large, \eqref{eq:AH_mean_n1} and \eqref{eq:huber_upper_bd} lead to $|\widetilde{\mu}^{\eta}_j - \mu_j|_2 \leq \tau/2 = r$ for all $\eta \in (0,1]$. By our construction in the beginning of the proof, this enforces $\widetilde{\mu}_j^{\eta} = \widetilde{\mu}_j$.
Finally, applying the union bound concludes the proof.
\end{proof}
\begin{remark}
    \Cref{Thm:m-est.mean} shows that, for a process whose second marginal moments are finite and whose functional dependence measure decays exponentially, the deviation error of Huber's $M$-estimator is of the rate $(\log n)\sqrt{(t+\log d)/n}$. We note that the extra $\log n$ factor is led by the $\log n$ factor appearing in Theorem~2.5. If we further restrict ourselves to linear processes, i.e.~\textup{(12)}, then using Theorem~2.6 instead results in the rate $\sqrt{(t+\log d)/n}$, which matches exactly the minimax lower bound of mean estimation.
\end{remark}

\subsection{Truncated estimator for high-dimensional mean}\label{sec:app:nonasy_trunc_mean}
In this subsection, we consider a simpler tail-robust mean estimator, the element-wise truncated mean estimator. This estimator achieves the same deviation error rate as Huber's $M$-estimator but without requiring bounded marginal density (Assumption~3). Moreover, this estimator has a closed form and can be computed directly. Recall the element-wise truncated mean estimator defined in \textup{(18)} as
    \begin{equation*}
        \widehat{\bm{\mu}} = (\widehat{\mu}_1, \dots, \widehat{\mu}_d)^{\intercal},
    \end{equation*}
    where $\widehat{\mu}_j = n^{-1}\sum_{i=1}^n\psi_{\tau}(X_{i,j})$. Recall that $\omega_{2}^2 = \max_{1 \leq j\leq d} \mathbb{E}[X_{1,j}^2]$ and $\Vert X_{\cdot} \Vert_{2} = \sup_{m \geq 0}\exp(cm)\sum_{s = m}^{\infty}\delta_{s,2}$.

    \begin{theorem}
    \label{thm:supp_element.mean2}
    For $t > 0$, choose the robustification parameter
    \begin{equation*}
        \tau \asymp \omega_{2}(\log n)^{-1}\sqrt{\frac{n}{t + \log d}},
    \end{equation*}
    Then, under Assumptions \ref{assumption:APP_finite.variance} and \ref{assumption:APP_decay.2order}, for a sufficient large $n$ such that $n \geq 4 \vee c/2$,
    we have with probability at least $1-2e^{-t}$
    \begin{equation*}
        |\widehat{\bm{\mu}} - \bm{\mu}|_{\infty} \lesssim \Vert X_{\cdot} \Vert_{2}(\log n)\sqrt{\frac{t+\log d}{n}}.
    \end{equation*}
    \end{theorem}
\begin{proof}
It suffices to consider only $\{X_{i,j}\}_{i \in \mathbb{Z}}$ for any $j \in [d]$. We have that
    \begin{equation*}
        \begin{aligned}
        n(\widehat{\mu}_j - \mu_j) =&  n\big[\mathbb{E}\psi_{\tau}(X_{1,j}) - \mu_j\big] + \sum_{i=1}^n\big[\psi_{\tau}(X_{i,j}) - \mathbb{E}\psi_{\tau}(X_{i,j})\big]\\ =& I + II.
        \end{aligned}
    \end{equation*}
Due to Theorem~2.2, the bias term can be bounded by
    \begin{equation*}
        |\mathbb{E}\psi_{\tau}(X_{i,j}) - \mu_j| \leq \omega^2/\tau.
    \end{equation*}
Thus, we have $-n\omega_{2}^2/\tau \leq I \leq n\omega_{2}^2/\tau$.
\\
    Next, we consider the term $II$. By Theorem~2.5, under Assumptions \ref{assumption:APP_finite.variance} and \ref{assumption:APP_decay.2order}, for $t > 0$, we have with probability at least $1 - 2e^{-t}$ that
    \begin{equation*}
        |II| \leq \sqrt{C_2}\|X_{\cdot}\|_2\sqrt{tn} + \sqrt{C_2}\tau\sqrt{t} + C_1\tau(\log n)^2t.
    \end{equation*}
    Combining with the upper bound of the term $I$, we have with probability at least $1-2e^{-t}$
    \begin{align*}
        n|\widehat{\mu}_j - \mu_j| \leq& \sqrt{C_2}\|X_{\cdot}\|_2\sqrt{tn} + \sqrt{C_2}\tau\sqrt{t} + C_1\tau(\log n)^2t + n\omega_{2}^2/\tau\\
        \lesssim& \sqrt{nt}(\log n)\Vert X_{\cdot} \Vert_{2},
    \end{align*}
    where the second inequality is obtained by letting
    \begin{equation*}
    \tau = \omega\sqrt{\frac{n}{C_1(\log n)^2t + C_2^{1/2}t^{1/2}}}.
\end{equation*}
    Dividing by $n$ on both sides and applying the union bound concludes the proof.
    \end{proof}
\begin{remark}
    Theorem \ref{thm:supp_element.mean2} indicates the truncated mean estimator achieves the same optimal (up to an $\log n$ factor) deviation error as Huber's $M$-estimator, which remove the bounded marginal density assumption (Assumption~3). Moreover, the truncated mean estimator can be computed without using optimization, thus it is more computational friendly than Huber's $M$-estimator.
\end{remark}

\section{Gaussian approximation based on the truncated estimator}\label{sec:Gaussian}
In this section, we study the Gaussian approximation for element-wise truncated mean estimator under temporal dependence. Let $\{\bm{X}_i\}_{i \in \mathbb{Z}} \subset \mathbb{R}^d$ be a potentially nonstationary time series in the form of \textup{(1)}. Since $\{\bm{X}_i\}_{i \in \mathbb{Z}}$ is allowed to be nonstationary, we denote 
\begin{equation*}
    \bm{\mu} = \frac{1}{n}\sum_{i = 1}^n\mathbb{E}[X_i].
\end{equation*}
Recall the element-wise truncated mean estimator $\widehat{\bm{\mu}}$ defined in \textup{(18)}, whose nonasymptotic properties is given in \Cref{sec:app:nonasy_trunc_mean}.
Let $\bm{Z} \in \mathbb{R}^d$ be a Gaussian vector such that $\bm{Z} \sim N(\bm{0}, \bm{\Sigma})$, where $\bm{\Sigma}$ is the long-run covariance matrix of $\{\bm{X}_i\}_{i \in \mathbb{Z}}$. Our goal is to obtain the error of the Gaussian approximation for $\widehat{\bm{\mu}}$ in Kolmogorov–Smirnov distance, i.e.
\begin{equation*}
    \rho_{n,\tau} = \sup_{t \in \mathbb{R}}\big|\mathbb{P}(\sqrt{n}|\widehat{\bm{\mu}} - \bm{\mu}|_{\infty} \leq t) - \mathbb{P}(|\bm{Z}|_{\infty} \leq t) \big|.
\end{equation*}
To achieve this goal, we decompose $\rho_{n,\tau}$ as the following three terms $\rho_{n,\tau}^*$, $\rho_{n,\tau}^{\diamond}$ and $\rho_{n,\tau}^{\circ}$. The first term is defined as
\begin{equation*}
    \rho_{n,\tau}^* = \sup_{t \in \mathbb{R}}\big|\mathbb{P}(\sqrt{n}|\widehat{\bm{\mu}} - \mathbb{E}[\widehat{\bm{\mu}}]|_{\infty} \leq t) - \mathbb{P}(|\bm{Z}^{\tau}|_{\infty} \leq t) \big|,
\end{equation*}
where $\bm{Z}^{\tau} \sim N(\bm{0}, \bm{\Sigma^{\tau}})$ with $\bm{\Sigma}^{\tau} = \lim_{n \to \infty}\cov(\sqrt{n}(\widehat{\bm{\mu}} - \mathbb{E}[\widehat{\bm{\mu}}]))$, i.e.~the long-run covariance matrix of $\{\psi_{\tau}(\bm{X}_i)\}_{i \in \mathbb{Z}}$.
The second and the third terms are defined as
\begin{equation*}
    \rho_{n,\tau}^{\diamond} = \sup_{t \in \mathbb{R}}\big|\mathbb{P}(\sqrt{n}|\widehat{\bm{\mu}} - \bm{\mu}|_{\infty} \leq t) - \mathbb{P}(\sqrt{n}|\widehat{\bm{\mu}} - \mathbb{E}[\widehat{\bm{\mu}}]|_{\infty} \leq t) \big|,
\end{equation*}
and
\begin{equation*}
    \rho_{n,\tau}^{\circ} = \sup_{t \in \mathbb{R}}\big|\mathbb{P}(|\bm{Z}^{\tau}|_{\infty} \leq t) - \mathbb{P}(|\bm{Z}|_{\infty} \leq t) \big|.
\end{equation*}
By the triangle inequality, we have the following decomposition
\begin{equation*}
    \rho_{n,\tau} \leq \rho_{n,\tau}^* + \rho_{n,\tau}^{\diamond} + \rho_{n,\tau}^{\circ}.
\end{equation*}
    The term $\rho_{n,\tau}^*$ represents the error of the Gaussian approximation for truncated process $\{\psi_{\tau}(\bm{X}_i)\}_{i \in \mathbb{Z}}$. The term $\rho_{n,\tau}^{\diamond}$ measures the Kolmogorov-Smirnov distance arised from the mean bias due to truncation, and the term $\rho_{n,\tau}^{\circ}$ represents the difference of two centered Gaussian random vectors with different covariance matrices. The Gaussian approximation result is provided in Theorem \ref{thm:GA_element-wise-trunc-mean}.
In additional to the exponential decay of second order functional dependence measure, i.e.~\Cref{assumption:APP_decay.2order}, our Gaussian approximation result requires the following assumptions.
\begin{assumption}\label{assumption:L2+theta}
    For some $0 < \theta \leq 1$, it satisfies that $M_{\theta} = \max_{1 \leq j \leq d}\Vert X_{i,j} \Vert_{2+\theta}^{2+\theta} < \infty$.
\end{assumption}
\begin{assumption}\label{assumption:blockvar_nondegen} 
    There exists a constant $b > 0$ such that 
    \begin{equation*}
        \min_{1 \leq j \leq d}\inf_{S \subseteq [1,n]}\frac{1}{|S|}\mathbb{E}\Big[\sum_{i \in S}(X_{i,j} - \mu_j) \Big]^2 > b.
    \end{equation*}
\end{assumption}
\noindent Assumption \ref{assumption:L2+theta} assumes finite $(2+\theta)$-th moment of the marginal distribution, which is slightly stronger than \Cref{assumption:APP_finite.variance}. Assumption \ref{assumption:blockvar_nondegen} ensures the nondegeneracy of the partial sums, which is needed to verify the condition of Theorem 2.1 in \cite{chernozhukov2017central}.

\begin{theorem}\label{thm:GA_element-wise-trunc-mean}
    Let 
    \[
    \tau \asymp \Big(\frac{n}{M^{3/2}\log d}\Big)^{\frac{1}{2+\theta}}, \;\;M \asymp n^{\alpha} \;\; \text{and} \;\; \log d \asymp n^{\beta}
    \]
     such that (a) $\alpha > 2\beta$ and (b) $16\alpha + 6\theta > (59+36\theta)\beta$. Let $\bm{Z} \sim N(\bm{0}, \bm{\Sigma})$ with $\bm{\Sigma}$ being the long-run covariance matrix of $\{\bm{X}_i\}_{i \in \mathbb{Z}}$. Then, under Assumptions \ref{assumption:APP_decay.2order}, \ref{assumption:L2+theta} and \ref{assumption:blockvar_nondegen}, we have as $n \to \infty$
\begin{equation*}
    \begin{aligned}
    \sup_{t \in \mathbb{R}}\big|\mathbb{P}(\sqrt{n}|\widehat{\bm{\mu}} - \bm{\mu}|_{\infty} \leq t) - \mathbb{P}(|\bm{Z}|_{\infty} \leq t) \big| \lesssim n^{(2\beta - \alpha)/3} \vee n^{ [(3+3\theta)\alpha + (4+3\theta)\beta - \theta]/(4+2\theta)} \to 0.
    \end{aligned}
\end{equation*}
\end{theorem}
\begin{remark}\label{remark:GA_theta=1}
Theorem \ref{thm:GA_element-wise-trunc-mean} allows us to consider the high-dimensional regime where $d$ diverges with $n$ exponentially. For example, when $\theta = 1$, let $\alpha = 2/19$ by optimizing the above error rate. Then, choose $M \asymp n^{2/19}$ and $\tau \asymp \big[n^{16/19}(\log d)^{-1}\big]^{1/3}$, we allow $\beta < 1/19$, i.e. $\log d = o\big( n^{1/19} \big)$, and we have
\begin{equation*}
    \sup_{t \in \mathbb{R}}\big|\mathbb{P}(\sqrt{n}|\widehat{\bm{\mu}} - \bm{\mu}|_{\infty} \leq t) - \mathbb{P}(|\bm{Z}|_{\infty} \leq t) \big| \lesssim n^{-2(1/19-\beta)/3}.
\end{equation*}
\end{remark}
\cite{lou2017simultaneous} consider the Gaussian approximation of the element-wise truncated mean estimator under independence case. In that setting, the condition on $\log d$ is less restrictive, i.e.~$\log d$ can be as larger as $o(n^{\theta/(4+3\theta)})$. By inspecting the proof of Theorem \ref{thm:GA_element-wise-trunc-mean}, the restriction is due to the use of the block technique. We divide sample into blocks of consecutive data with size $M \asymp n^\alpha$ and work on the re-scaled block means. The value $\alpha$ need to be large enough relative to $d$ (see the conditions (a) and (b) of Theorem \ref{thm:GA_element-wise-trunc-mean}) to preserve the underlying dependence structure. However, when the data are independence, we could set $M = 1$, i.e.~$\alpha = 0$, the conditions (a) and (b) are removed, and the same rate as in the iid setting is obtained.

Although Theorem \ref{thm:GA_element-wise-trunc-mean} considers simultaneous inference for high-dimensional mean vectors, it can also be applied to perform simultaneous inference (such as the test of serial correlations) for high-dimensional autocovariance matrices based on the considered element-wise truncated autocovariance matrix estimator. It is important to mention that we may need to choose different robustification parameters for different purposes as suggested by Theorems 3.3 and \ref{thm:GA_element-wise-trunc-mean}. However, if suggested $\tau$ of \Cref{thm:GA_element-wise-trunc-mean} is chosen, then under the stronger restriction on the dimension $d$ of \Cref{thm:GA_element-wise-trunc-mean}, i.e.~$\log d \asymp n^{\beta}$ with $\beta < 1/19$, the deviation error obtained in Theorem~3.3 still holds. See Remark~5 for a discussion. 

\section{Proofs for Section 2}\label{Appendix:proof_for_sec2}

\begin{proof}[Proof of Theorem~2.2]
For any $i \in \mathbb{Z}$, we have that
    \begin{align*}
        \left|\mathbb{E}[\psi_{\tau}(X_{i} - u)] - \mathbb{E}[X_i - u]\right| =& \left| \mathbb{E}[ (X_{i} - u -\tau) \mathbbm{1}\{X_{i} - u > \tau\}] + \mathbb{E}[ (X_{i} - u+\tau) \mathbbm{1}\{X_{i} - u < -\tau\}] \right|\\
        \leq& \mathbb{E}[(X_{i} - u)\mathbbm{1}\{X_{i} - u > \tau\}] + \mathbb{E}[-(X_{i}-u) \mathbbm{1}\{X_{i}-u < -\tau\}]\\
        =& \mathbb{E}[|X_{i} - u| \mathbbm{1}\{|X_{i} - u| > \tau\}]\\
        \leq& \mathbb{E}[(X_{i}-u)^2]/\tau,
    \end{align*}
    which completes the proof.
\end{proof}

\begin{proof}[Proof of Lemma 2.3]
    The statement is true since the truncation operator given in (5) is a Lipschitz function with a Lipschitz constant being $1$.
\end{proof}

\begin{proof}[Proof of Theorem~2.4]
    Define the projection operator $\mathcal{P}_j\cdot = \mathbb{E}[\cdot|\mathcal{F}_j] - \mathbb{E}[\cdot| \mathcal{F}_{j-1}]$ with $j \in \mathbb{Z}$.  A random variable $X_i$ is decomposed as
    \begin{equation*}
        X_i - \mathbb{E}[X_i] = \sum_{k = 0}^{\infty}\big(\mathbb{E}[X_i|\mathcal{F}_{i-k}] - \mathbb{E}[X_i| \mathcal{F}_{i-k-1}]\big) = \sum_{k = 0}^{\infty}\mathcal{P}_{i-k}X_i.
    \end{equation*}
    It holds for any $l \geq 0$ that 
    \begin{equation*}
    \begin{aligned}
        &\sup_{t \in \mathbb{Z}}\big|\cov(X_{t}, X_{t+l})\big| = \sup_{t \in \mathbb{Z}}\bigg|\mathbb{E}\Big[\Big(\sum_{k = 0}^{\infty}\mathcal{P}_{-k}X_t\Big)\Big(\sum_{k = 0}^{\infty}\mathcal{P}_{l-k}X_{t+l}\Big)\Big]\bigg|\\
        =& \sup_{t \in \mathbb{Z}}\bigg|\sum_{k =     0}^{\infty}\mathbb{E}\big[(\mathcal{P}_{-k}X_t)(\mathcal{P}_{-k}X_{t+l})\big]\bigg| \leq \sum_{k =     0}^{\infty}\sup_{t \in \mathbb{Z}}\Big|\mathbb{E}\big[(\mathcal{P}_{-k}X_t)(\mathcal{P}_{-k}X_{t+l})\big]\Big|\\
        \leq& \sum_{k = 0}^{\infty}\sup_{t \in \mathbb{Z}}\Vert \mathcal{P}_{-k}X_{t} \Vert_2 \Vert \mathcal{P}_{-k}X_{t+l}\Vert_2 \leq \sum_{k = 0}^{\infty} \delta_{k,2}\delta_{k+l,2} \leq \sqrt{\sum_{k = 0}^{\infty} \delta_{k,2}^2}\sqrt{\sum_{k = 0}^{\infty} \delta_{k+l,2}^2}\\
        \leq& \Big(\sum_{k = 0}^{\infty} \delta_{k,2}\Big) \Big(\sum_{k = l}^{\infty} \delta_{k,2}\Big) = \Delta_{0,2}\Delta_{l,2},
    \end{aligned}
    \end{equation*}
    where the first inequality follows the triangle inequality and the second and fourth inequalities follow H\"older's inequality. The second equality also follows the orthogonality of $\mathcal{P}_j\cdot$, i.e.~for $i < j$
    \begin{equation*}
        \mathbb{E}[(\mathcal{P}_iX_r)(\mathcal{P}_jX_s)] = \mathbb{E}[\mathbb{E}[(\mathcal{P}_iX_r)(\mathcal{P}_jX_s)|\mathcal{F}_{i}]] = \mathbb{E}[(\mathcal{P}_iX_r)\mathbb{E}[X_s-X_s|\mathcal{F}_i]] = 0,
    \end{equation*}
    and the orthogonality also holds for $i > j$ by symmetry. The third inequality is due to the fact that
    \begin{equation*}
    \begin{aligned}
       \Vert \mathcal{P}_{j}X_i \Vert_2 =& \Vert \mathbb{E}[X_i| \mathcal{F}_{j}] - \mathbb{E}[X_i| \mathcal{F}_{j-1}]\Vert_2 = \Vert\mathbb{E}[X_i| \mathcal{F}_{j}] - \mathbb{E}[X_{i,\{j\}}| \mathcal{F}_{j-1}]\Vert_2\\
       =& \Vert\mathbb{E}[X_i-X_{i,\{j\}}| \mathcal{F}_{j}]\Vert_2 \leq \Vert X_{i}-X_{i,\{j\}}\Vert_2 \leq \delta_{i-j,2},
    \end{aligned}
    \end{equation*}
    where the second and the third equality follows the definition of the coupled random variables $X_{i,\{j\}}$, the first inequality follows Jensen's inequality, and the second inequality follows the definition of the functional dependence measure. By the same arguments, we have for any $l < 0$ that
    \begin{equation*}
    \begin{aligned}
        &\sup_{t \in \mathbb{Z}}\big|\cov(X_t, X_{t+l})\big| \leq \Delta_{0,2}\Delta_{-l,2}.
    \end{aligned}
    \end{equation*}
Therefore, we have that
\begin{equation}\label{eq:LRV_univariate_ub}
\begin{aligned}
    &C_{\mathrm{LRV}} = \sum_{l = -\infty}^{\infty}\sup_{t \in \mathbb{Z}}\big|\cov(X_t, X_{t+l})\big| \leq 2\Delta_{0,2}\sum_{l = 0}^{\infty}\Delta_{l,2} \leq 2C_{\mathrm{FDM}}^2\sum_{l = 0}^{\infty}\exp(-cl^{\gamma_1}).
\end{aligned}
\end{equation}
To bound \eqref{eq:LRV_univariate_ub}, we compare the series $\{a_m = \exp(-cm^{\gamma_1})\}_{m = 1}^{\infty}$ and $\{b_m = (m+1)^{-\nu}\}_{m=1}^{\infty}$ with $\nu > 1$. Since by the properties of the Riemann zeta function, it holds that $\sum_{m = 1}^{\infty}b_m < \infty$.

We embed the series $\{a_m\}_{m = 1}^{\infty}$ and $\{b_m\}_{m=1}^{\infty}$ into continuous time processes by defining $a_x = a_{\lceil x \rceil}$ and $b_x = b_{\lceil x \rceil}$ for $x \in [1,\infty)$. Since for any $x \in [1,\infty)$, $a_x, b_x > 0$ and we define a function $g(x)$ on $x \in [1, \infty)$ as
\begin{equation*}
    g(x) = \log\bigg(\frac{a_x}{b_x}\bigg) = -cx^{\gamma_1} + \nu\log(x+1).
\end{equation*}
We have the derivative as
\begin{equation*}
    g^{\prime}(x) = -c\gamma_1x^{\gamma_1-1} + \frac{\nu}{x+1}.
\end{equation*}
Since for any absolute constants $c, \gamma_1 > 0$ and $\nu > 1$, there exists a finite $x^* \geq 1$ such that for any $x \geq x^*$, it holds that
\begin{equation*}
    \exp(-cx^{\gamma_1}) \leq (x+1)^{-\nu} \;\; \text{and} \;\; g^{\prime}(x) < 0.
\end{equation*}
Letting $m^* = \lceil x^* \rceil$, we have that
\begin{equation*}
\begin{aligned}
    C_{\mathrm{LRV}} \leq& 2C_{\mathrm{FDM}}^2\sum_{m = 0}^{\infty}\exp(-cm^{\gamma_1})\\
    =& 2C_{\mathrm{FDM}}^2\sum_{m = 0}^{m^*}\exp(-cm^{\gamma_1}) + 2C_{\mathrm{FDM}}^2\sum_{m = m^*+1}^{\infty}\exp(-cm^{\gamma_1})\\
    \leq& 2C_{\mathrm{FDM}}^2\sum_{m = 0}^{m^*}\exp(-cm^{\gamma_1}) + 2C_{\mathrm{FDM}}^2\sum_{m = m^*+1}^{\infty}(m + 1)^{-\nu} < \infty.
\end{aligned}
\end{equation*}
\end{proof}

\section{Proofs for Section 3}
\begin{proof}[Proof of Theorem~3.1]
    It suffices to consider the decay rate of $|\gamma_{\ell,(jk)}|$ for any $j,k \in [d]$. Recall the projection operator $\mathcal{P}_j\cdot = \mathbb{E}[\cdot|\mathcal{F}_j] - \mathbb{E}[\cdot|\mathcal{F}_{j-1}]$. We represent $X_{i,j} - \mu_j = \sum_{h=0}^{\infty}\mathcal{P}_{-h}X_{i,j}$, so that
\begin{align*}
    |\gamma_{\ell,(jk)}| =& \bigg|\sum_{h=0}^{\infty}\mathbb{E}[(\mathcal{P}_{-h}X_{0,j})(\mathcal{P}_{-h}X_{\ell,k})]\bigg| \leq \sum_{h=0}^{\infty}\big|\mathbb{E}[(\mathcal{P}_{-h}X_{0,j})(\mathcal{P}_{-h}X_{\ell,k})]\big| \leq \sum_{h=0}^{\infty}\delta_{h,2,j}\delta_{h+\ell,2,k}\\
    \leq& \|X_{\cdot}\|_2^2\exp(-c\ell),
\end{align*}
where the first equality is due to the orthogonality of the projection operator, the first inequality follows from the triangle inequality, and the second inequality follows from H\"older's inequality and Jensen's inequality, and the third inequality follows from H\"older's inequality.
\end{proof}

\begin{proof}[Proof of Theorem 3.2]
The proof is similar to that of Theorem \ref{Thm:m-est.mean}, thus we only mention the differences. Recall the $(j,k)$-th entry of $\widetilde{\bm{\Sigma}}_{\ell}$,
\begin{align*}
    \widetilde{\gamma}_{\ell,(jk)} = \widetilde{H}_{\ell,(jk)} - \widetilde{\mu}_j\widetilde{\mu}_k,
\end{align*}
with
\begin{align*}
    \widetilde{H}_{\ell,(jk)} = \argmin_{u \in \mathbb{R}} \frac{1}{n -\ell}\sum_{i=\ell+1}^{n}L_{\tau_{\ell}}(X_{i-\ell,j}X_{i,k} - u) \;\; \text{and} \;\; \widetilde{\mu}_j = \argmin_{u \in \mathbb{R}} \frac{1}{n}\sum_{i=1}^nL_{\tau}(X_{i,j} - u).
\end{align*}
The mean estimator $\widetilde{\mu}_j$ has already been studied in Theorem \ref{Thm:m-est.mean}, thus we only consider $\widetilde{H}_{\ell,(jk)}$, which can be treated as the $M$-estimator for the mean of the process $\{X_{i-\ell,j}X_{i,k}\}_{i \in \mathbb{Z}}$.
\\
Let $Y_i = X_{i-\ell,j}$. By \Cref{lemma:fdm_XY}, we have that for any $s \geq 0$ the functional dependence measure of $\{X_{i-\ell,j}X_{i,k}\}_{i \in \mathbb{Z}}$
\begin{align*}
    \delta_{s,2}^{XY} \leq 2\omega_4\delta_{s,4}^X.
\end{align*}
Thus, by Assumptions 1 and 2, we have that
    \begin{align*}
        \sup_{m \geq 0}\exp(-cm)\sum_{s = m}^{\infty}\delta_{s,2}^{XY} \leq 2\omega_4\Vert X_{\cdot}\Vert_{4} < \infty.
    \end{align*}
Moreover, under Assumption~1, we have that
\begin{align*}
    \max_{j,k \in [d]}\|X_{i-\ell,j}X_{i,k}\|_2 \leq \omega_4^2 < \infty. 
\end{align*}
For any $t > 0$, assuming $n$ is large enough such that
\begin{equation*}
    n \geq C(\log n)^2(t + 2\log d),
\end{equation*}
with $C > 0$ being a sufficiently large absolute constant.
Then, by Theorem \ref{Thm:m-est.mean} and choosing
    \begin{equation*}
        \tau_{\ell} \asymp \omega_4^2 (\log (n-\ell))^{-1}\sqrt{\frac{n-\ell}{t + 2\log d}} \asymp \omega_4^2 (\log n)^{-1}\sqrt{\frac{n}{t + 2\log d}},
    \end{equation*}
we have with probability at least $1 - 3e^{-t}$
    \begin{equation}\label{eq:proof.thm.m-est.exp.cross}
        \max_{j,k\in [d]}\left|\widetilde{H}_{\ell,(jk)} - \mathbb{E}[H_{i,\ell,(jk)}]\right|
        \lesssim \omega_4\|X_{\cdot}\|_4\log (n-\ell)\sqrt{\frac{t + 2\log d}{n-\ell}}.
    \end{equation}
    According to H\"older's inequality and Lemma \ref{lemma:DAN_prop}, the conditions in Theorem 3.2 imply the conditions in Theorem \ref{Thm:m-est.mean}. Thus, for any $t > 0$, choose the robustification parameter 
    \begin{equation*}
        \tau \asymp (\log n)^{-1} \sqrt{\frac{n}{t+2\log d}}.
    \end{equation*}
    Then, under the same assumptions as in Theorem 3.2, we have with probability at least $1-2e^{-t}$
    \begin{equation}\label{eq:proof.thm.m-est.exp.mean}
        \max_{j\in [d]}\left|\widetilde{\mu}_j - \mu_j\right| \lesssim \Vert X_{\cdot} \Vert_{2}(\log n )\sqrt{\frac{t+\log d}{n}}.
    \end{equation} 
Finally, combining \eqref{eq:proof.thm.m-est.exp.cross} and \eqref{eq:proof.thm.m-est.exp.mean} concludes the proof.
\end{proof}

\begin{proof}[Proof of Theorem 3.3]
    We consider the deviation error of $\widehat{\gamma}_{\ell,(jk)}$, which is the $(j,k)$-th entry of $\widehat{\bm{\Sigma}}_\ell$. By the triangle inequality, we have
    \begin{align}\label{eq:element_autocov_decompose}
        \left|\widehat{\gamma}_{\ell,(jk)} - \gamma_{\ell,(jk)}\right| \leq& \frac{1}{n-\ell}\bigg|\sum_{i=\ell+1}^{n}\Big\{\psi_{\tau}(X_{i-\ell,j}X_{i,k}) - \mathbb{E}\big[\psi_{\tau}(X_{i-\ell,j}X_{i,k})\big]\Big\}\bigg|\nonumber\\
        &+ \Big|\mathbb{E}\big[\psi_{\tau}(X_{i-\ell,j}X_{i,k})\big] - \mathbb{E}(X_{i-\ell,j}X_{i,k}) \Big| + \big|\mu_j(\widehat{\mu}_k - \mu_k)\big|\nonumber\\
        &+\big|(\mu_j - \widehat{\mu}_j)\mu_k\big| + \big|(\mu_j - \widehat{\mu}_j)(\mu_k - \widehat{\mu}_k)\big|\nonumber\\
        =& I_1 + I_2 + I_3 + I_4 + I_5.
    \end{align}
The terms $I_3$, $I_4$ and $I_5$ are due to the estimation error of $\widehat{\mu}_j$ and $\widehat{\mu}_k$, thus we have
    \begin{equation}
    \label{eq:element_autocov_meanpart}
        I_3+I_4+I_5 \leq 2|\bm{\mu}|_{\infty}\max_{j \in [d]}|\widehat{\mu}_j - \mu_j| + \big(\max_{j \in [d]}|\widehat{\mu}_j - \mu_j|\big)^2.
    \end{equation}
The bias term $I_2$ can be bounded by Theorem~2.2 as
    \begin{equation}
    \label{eq:element_autocov_bias}
        I_2 = \mathbb{E}\big[|X_{i-\ell,j}X_{i,k}| \mathbbm{1}_{|X_{i-\ell,j}X_{i,k}| > \tau} \big] \leq \omega_{4}^4/\tau.
    \end{equation}
For the term $I_1$, we apply \Cref{thm:supp_element.mean2} on $\big\{\psi_{\tau}(X_{i-\ell,j}X_{i,k})\big\}_{i \in \mathbb{Z}}$, under Assumptions 1 and 2. The proof is the same as that of Theorem \ref{thm:supp_element.mean2}, except that the process of interesting is $\big\{\psi_{\tau}(X_{i-\ell,j}X_{i,k})\big\}_{i \in \mathbb{Z}}$. \\
Let $Y_i = X_{i-\ell,j}$. By Theorem~2.3 and \Cref{lemma:fdm_XY}, we have that for any $s \geq 0$ the functional dependence measure of $\{\psi_{\tau}(X_{i-\ell,j}X_{i,k})\}_{i \in \mathbb{Z}}$
\begin{align*}
    \delta_{s,2}^{\psi_{\tau}(XY)} \leq \delta_{s,2}^{XY} \leq 2\omega_4\delta_{s,4}^X.
\end{align*}
Thus, by Assumptions 1 and 2, we have that
    \begin{align*}
        \sup_{m \geq 0}\exp(-cm)\sum_{s = m}^{\infty}\delta_{s,2}^{\psi_{\tau}(XY)} \leq 2\omega_4\Vert X_{\cdot}\Vert_{4} < \infty.
    \end{align*}
Moreover, under Assumption~1, we have that
\begin{align*}
    \max_{j,k \in [d]}\|\psi_{\tau}(X_{i-\ell,j}X_{i,k})\|_2 \leq \omega_4^2 < \infty. 
\end{align*}
For any $t > 0$, assuming $n$ is large enough such that
\begin{equation*}
    n \geq C(\log n)^2(t + 2\log d),
\end{equation*}
with $C > 0$ being a sufficiently large absolute constant.
Then, by Theorem \ref{thm:supp_element.mean2} and choosing
    \begin{equation*}
        \tau \asymp \omega_4^2 (\log (n-\ell))^{-1}\sqrt{\frac{n-\ell}{t + 2\log d}} \asymp \omega_4^2 (\log n)^{-1}\sqrt{\frac{n}{t + 2\log d}},
    \end{equation*}
we have with probability at least $1 - 2e^{-t}$
    \begin{equation}\label{eq:proof.thm.trunc.exp.cross}
        \max_{j,k\in [d]}\left|\widehat{H}_{\ell,(jk)} - \mathbb{E}[H_{i,\ell,(jk)}]\right|
        \lesssim \omega_4\|X_{\cdot}\|_4\log (n-\ell)\sqrt{\frac{t + 2\log d}{n-\ell}}.
    \end{equation}
    According to H\"older's inequality and Lemma \ref{lemma:DAN_prop}, the conditions in Theorem 3.3 imply the conditions in Theorem \ref{thm:supp_element.mean2}. Thus, for any $t > 0$, choose the robustification parameter 
    \begin{equation*}
        \tau \asymp (\log n)^{-1} \sqrt{\frac{n}{t+2\log d}}.
    \end{equation*}
    Then, under the same assumptions as in Theorem 3.3, we have with probability at least $1-2e^{-t}$
    \begin{equation}\label{eq:proof.thm.trunc.exp.mean}
        \max_{j\in [d]}\left|\widehat{\mu}_j - \mu_j\right| \lesssim \Vert X_{\cdot} \Vert_{2}(\log n )\sqrt{\frac{t+\log d}{n}}.
    \end{equation} 
Finally, combining \eqref{eq:proof.thm.trunc.exp.cross} and \eqref{eq:proof.thm.trunc.exp.mean} concludes the proof.
\end{proof}

\section{Proof for Section~4}
\begin{proof}[Proof of Theorem 4.2]
We follow the essential idea of the proof follows that of Theorem 4.1 in \cite{zhang2014bootstrapping}, which is as follows. We first construct the block sums in order to preserve the underline (but unknown) temporal dependence. Then, we approximate these dependent block sums by their corresponding $m$-dependence approximations, and thus we are allowed to use the results for independent data from \cite{chernozhukov2013gaussian}.
\\
\\
As is described in Section~4, we divide interval $[\ell+1,n]$ into $2R$ number of blocks with block size $S$. There are $R$ number of pairs of consecutive odd and even blocks. For $r \in \{1, \dots, R\}$, denote 
\[
\bm{O}_r = \sum_{t \in \mathcal{S}_{2r-1}}\psi_{\tau}(\bm{H}_{t,\ell}) \;\; \text{and} \;\;
    \bm{E}_r^{(S)} = \sum_{t \in \mathcal{S}_{2r}} \psi_{\tau_{\ell}}(\bm{H}_{t,\ell}),
\]
respectively the blocks sums of odd and even blocks associated with the $r$-th pair. Denote $\bm{D}_r = (D_{r,j})_{j = 1}^{d^2}$ as the vectorization of $(\bm{O}_r - \bm{E}_r)$.
Recall the lag-$\ell$ outer product $\bm{H}_{t,\ell}$ defined in \textup{(15)}. Define the $S$-dependent version of $\bm{O}_r$ and $\bm{E}_r$ respectively as
\begin{align*}
    \bm{O}_r^{(S)} = \sum_{t \in \mathcal{S}_{2r-1}} \mathbb{E}\big[\psi_{\tau}(\bm{H}_{t,\ell})| \mathcal{F}_{[t-S, t]}\big] \;\; \text{and} \;\;
    \bm{E}_r^{(S)} = \sum_{t \in \mathcal{S}_{2r}} \mathbb{E}\big[\psi_{\tau_{\ell}}(\bm{H}_{t,\ell})| \mathcal{F}_{[t-S, t]}\big],
\end{align*}
where the filtration $\mathcal{F}_{[t-S, t]} = \sigma(\epsilon_{t-S}, \dots, \epsilon_t)$.
Denote also $\bm{D}_r^{(S)}(c) = (D_{r,j}^{(S)})_{j = 1}^{d^2}$ as the vectorization of $(\bm{O}_r^{(s)} - \bm{E}_r^{(s)})$.
\\
Let $a_t \in \{1,-1\}$, for $t \in \mathbb{Z}$, be a deterministic sequence. Let $Y_t = X_{t-\ell,j}$. By Theorem~2.3 and \Cref{lemma:fdm_XY}, we have that for any $s \geq 0$ the functional dependence measure of $\{a_t\psi_{\tau}(H_{t,\ell,(jk)}) = a_t\psi_{\tau}(X_{t-\ell,j}X_{t,k})\}_{t \in \mathbb{Z}}$ satisfies that for any $0 < q \leq 3$
\begin{align}\label{eq:fdm_psi_H}
    \delta_{s,q}^{a\psi_{\tau}(H)} \leq \delta_{s,q}^{H} \leq 2\omega_{2q}\delta_{s,2q}^X,
\end{align}
under Assumption 4 and for any $m \geq 0$, $1 \leq j,k \leq d$ and $\tau > 0$.
So, for any $1 \leq i \leq d^2$,
\begin{align}\label{eq:fdm_D_block}
    \max_{1 \leq r \leq R}\big\Vert D_{r,i} - D^{(S)}_{r,i} \big\Vert_q \leq& 2S\sup_{t \in \mathbb{Z}}\Big\Vert a_t\{ \psi_{\tau}(H_{t,\ell,(jk)}) - \mathbb{E}\big[\psi_{\tau}(H_{t,\ell,(jk)})| \mathcal{F}_{[t-S, t]}\big]\big\} \Big\Vert_q \nonumber\\
    \leq& 2S\omega_{2q}\Vert X_{\cdot} \Vert_{2q}\exp(-cS),
\end{align}
where the first inequality follows from the triangle inequality, and the second inequality follows from Assumption~5.
Moreover, the second moment of the partial sum satisfies that
\begin{align}\label{eq:D_block_sum_q=2}
    \max_{1 \leq r \leq R}\big\Vert D_{r,i} \big\Vert_2 \leq& \bigg\{2S\sum_{\ell = -\infty}^{\infty}\big|\cov\big(a_t\psi_{\tau}(H_{t,\ell,(jk)}), a_{t+\ell}\psi_{\tau}(H_{t+l,\ell,(jk)}) \big) \big|\bigg\}^{1/2} \leq C_2S^{1/2},
\end{align}
where the second inequality follows from Assumption 5 and the same arguments stated in the proof of Lemma 2.4. Define the projection operator as $\mathcal{P}_{t}\cdot = \mathbb{E}[\cdot|\mathcal{F}_t] - \mathbb{E}[\cdot|\mathcal{F}_{t-1}]$. For any $t$, we decompose the random variable $a_t\psi_{\tau}(H_{t,\ell,(jk)})$ as
    \begin{align*}
        a_t\big\{\psi_{\tau}(H_{t,\ell,(jk)}) - \mathbb{E}[\psi_{\tau}(H_{t,\ell,(jk)})]\big\} = \sum_{m = 0}^{\infty}\mathcal{P}_{t-m}\big(a_t\psi_{\tau}(H_{t,\ell,(jk)})\big),
    \end{align*}
    where $\big\{\mathcal{P}_{t-m}\big(a_t\psi_{\tau}(H_{t,\ell,(jk)})\big)\big\}_{t \in \mathbb{Z}}$ is a martingale difference sequence relative to the filtration $\{\mathcal{F}_t\}_{t \in \mathbb{Z}}$.
For the third moment of the partial sum, it holds that
\begin{align}\label{eq:D_block_sum_q=3}
    \max_{1 \leq r \leq R}\big\Vert D_{r,i} \big\Vert_3 \leq& \sum_{m = 0}^{\infty}\bigg\Vert\sum_{t \in \mathcal{S}_{2r-1} \cup \mathcal{S}_{2r}} \mathcal{P}_{t-m} \big(a_t\psi_{\tau}(H_{t,\ell,(jk)})\big) \bigg\Vert_3 \nonumber\\
    \leq& \sqrt{2}\sum_{m = 0}^{\infty}\bigg( \sum_{t \in \mathcal{S}_{2r-1} \cup \mathcal{S}_{2r}} \big\Vert\mathcal{P}_{t-m}\big( a_t\psi_{\tau}(H_{t,\ell,(jk)})\big)\big\Vert_3^2 \bigg)^{1/2} \nonumber\\
    \leq& 2\sqrt{S} \sum_{m = 0}^{\infty}\sup_{t \in \mathbb{Z}}\big\Vert a_t\psi_{\tau}(H_{t,\ell,(jk)}) - a_t\psi_{\tau}(H_{t,\{t-m\},\ell,(jk)}) \big\Vert_{3} \nonumber\\
    \leq& 4\sqrt{S} \omega_{6}\Vert Y_{\cdot} \Vert_{6},
\end{align}
where the first inequality follows from the triangle inequality, the second inequality follows from Burkholder's inequality, the third inequality follows from Jensen's inequality and the fourth inequality follows from \eqref{eq:fdm_psi_H}.
\\
\\
We have $\{H_{t,\ell}\}_{t = \ell+1}^n$ is stationary and $\mathbb{E}[D_{r,i}] = \mathbb{E}[D^{(S)}_{r,i}] = 0$ for any $1 \leq r \leq R$ and $1 \leq i \leq d^2$. Denote $\Gamma_{(jk)}$ the $(j,k)$-th entry of $\bm{\Gamma}$, the long-run covariance matrix defined in \textup{(21)}.
Define the maximal estimation error of the long-run covariances based on the block differences among all entries as
\begin{align}\label{eq:error_all}
    e_{D} = \max_{1 \leq i,j \leq d^2}\bigg|\frac{1}{2RS}\sum_{r = 1}^{R} D_{r,j}D_{r,k} - \Gamma_{(jk)}\bigg|.
\end{align}
By Assumption 6 and Theorem 3.1 in \cite{chernozhukov2013gaussian}, we have for any $\vartheta > 0$
\begin{align}\label{eq:size_UB}
    \sup_{\alpha \in (0,1)}\Big|\mathbb{P}\Big(T_{\ell} \leq q^{(boot)}(\alpha)\Big) - \alpha\Big| \leq 3\rho_{T_{\ell}} + 2\pi(\vartheta) + 2\mathbb{P}(e_{D} > \vartheta),
\end{align}
where $\rho_{T_{\ell}} = \sup_{t \in \mathbb{R}}\Big|\mathbb{P}\big(T_{\ell} \leq t\big) - \mathbb{P}\big(|\bm{Z}|_{\infty} \leq t\big)\Big|$ and $\pi(\vartheta) = C\vartheta^{1/3}(1 \vee \log(d^2/\vartheta))^{2/3}$ with $C > 0$ being an absolute constant. By Theorem~4.1, we have that
\begin{align}\label{eq:size_gaussian_approx}
    \rho_{T_{\ell}} = n^{-2(1/19 - \beta)/3}.
\end{align}
\\
Next, we focus on the upper bound of the term $\mathbb{P}(e_{D} >\vartheta)$.
\\
We have that
\begin{align}\label{eq:ub_Eo_H0}
        e_{D} \leq& \max_{1 \leq j,k \leq d^2}\bigg|\frac{1}{2RS}\sum_{r = 1}^R \big\{D_{r,j}D_{r,k} - \mathbb{E}[D_{r,j}D_{r,k}]\big\} \bigg| +\max_{1 \leq j,k \leq d^2}\bigg|\frac{1}{2S} \mathbb{E}[D_{r,j}D_{r,k}] - \Gamma_{(jk)}\bigg| \nonumber\\
        \leq& \max_{1 \leq j,k \leq d^2}\bigg|\frac{1}{2RS}\sum_{r = 1}^R (D_{r,j}D_{r,k} - D_{r,j}^{(S)}D_{r,k}^{(S)}) \bigg| + \max_{1 \leq j,k \leq d^2}\bigg|\frac{1}{2RS}\sum_{r = 1}^n \big\{D_{r,j}^{(S)}D_{r,k}^{(S)} - \mathbb{E}[D_{r,j}^{(S)}D_{r,k}^{(S)}]\big\} \bigg| \nonumber\\
        & + \frac{1}{2S}\max_{1 \leq j,k \leq d^2}\Big| \mathbb{E}[D_{r,j}^{(S)}D_{r,k}^{(S)} - D_{r,j}D_{r,k}] \Big| +\max_{1 \leq j,k \leq d^2}\bigg|\frac{1}{2S} \mathbb{E}[D_{r,j}D_{r,k}] - \Gamma_{(jk)}\bigg| \nonumber\\
        =& I_1 + I_2 + I_3 + I_4.
    \end{align}
For the term $I_1$ in \eqref{eq:ub_Eo_H0}, we have for any $r \in [R]$ and $1 \leq j,k \leq d^2$ that
\begin{align*}
        &\mathbb{E}\big[|D_{r,j}D_{r,k} - D_{r,j}^{(S)}D_{r,k}^{(S)}|\big] \leq \mathbb{E}\big[|D_{r,j} - D_{r,j}^{(S)}||D_{r,k}|\big] + \mathbb{E}\big[|D_{r,j}^{(S)}| |D_{r,k} - D_{r,k}^{(S)}|\big]\\
        \leq& \Vert D_{r,j} - D_{r,j}^{(S)}\Vert_2 \Vert D_{r,k}\Vert_2 + \Vert D_{r,j}^{(S)}\Vert_2 \Vert D_{r,k} - D_{r,k}^{(S)}\Vert_2\\
        \leq& 4S\omega_{2q}\Vert X_{\cdot} \Vert_{2q}\exp(-cS)C_2S^{1/2} \leq C_3S^{3/2}\exp(-cS),
    \end{align*}
and, we have
\begin{align}\label{eq:ub_Eo_H0_I1}
    \mathbb{E}[I_1] \leq C_3S^{1/2}\exp(-cS).
\end{align}
For the term $I_2$, note that $\{D_{r,j}^{(S)}D_{r,k}^{(S)}\big\}_{r = 1}^R$ are mutually independent by construction. By the maximal inequality for partial sum of bounded random variables (see Lemma A.1 in \cite{chernozhukov2013gaussian}), we have that
\begin{align*}
    \mathbb{E}[I_2] \leq C_4\bigg(\sigma_{D}\sqrt{\frac{\log d}{R}} + \frac{\tau\log d}{R}\bigg),
\end{align*}
where $\sigma_{D}^2 = S^{-2}\max_{1 \leq j,k \leq d^2}\mathbb{E}\Big[\big(D_{r,j}^{(S)}D_{r,k}^{(S)}\big)^2\Big]$. It follows that
\begin{align*}
    \sigma_{D}^2 \leq \frac{1}{S^2}\max_{j \in [d^2]} \big\Vert D_{r,j}^{(s)}\big\Vert_4^4 \leq \frac{\tau}{S}\max_{j \in [d^2]} \big\Vert D_{r,j}^{(s)}\big\Vert_3^3 \leq C_5\tau S^{1/2},
\end{align*}
where the first inequality follows from H\"older's inequality, the second inequality follows from the fact that $|D_{r,j}^{(s)}| \leq S\tau$, and the third inequality follows from \eqref{eq:D_block_sum_q=3}. Thus,
\begin{align}\label{eq:ub_Eo_H0_I2}
    \mathbb{E}[I_2] \leq C_6\bigg(\sqrt{\frac{\tau S^{1/2}\log d}{R}} + \frac{\tau\log d}{R}\bigg).
\end{align}
\\
For the term $I_3$, by the same argument as for the term $I_1$, We have that
\begin{align}\label{eq:ub_Eo_H0_I3}
    \mathbb{E}[I_3] \leq  C_4S^{1/2} \exp(-cS).
\end{align}
\\
For the term $I_4$, recall the definition \textup{(21)} and note that there is an one to one correspondence between an entry index $i_1$ of $\bm{U}$ and an entry index $(j_1k_1)$ of $|n-\ell|^{1/2}\widehat{\Sigma}_{\ell}$. We rewrite the $(i_1,i_2)$-th entry of $\bm{\Gamma}$, for $1 \leq i_1,i_2 \leq d^2$ as
\begin{align*}
        \Gamma_{(i_1i_2)} = \lim_{n \to \infty}\cov\big(U_{\mathcal{\ell},i_1}, U_{\mathcal{I}_{\ell},i_2}\big)= \sum_{m = -\infty}^{\infty}\cov\big(\psi_{\tau}(H_{t,\ell,(j_1k_1)}), 
        \psi_{\tau}(H_{t+m,\ell,(j_2k_2)})\big),
    \end{align*}
where $(j_1,k_1)$ and $(j_2,k_2)$ correspond respectively to $i_1$ and $i_2$. Direct calculation also leads to
\begin{align*}
    \frac{1}{2S}\mathbb{E}[D_{r,i_1}D_{r,i_2}] = \sum_{m = 1-2S}^{2S-1}\Big( 1 - \frac{3|m|}{2S}\Big)\cov\big(\psi_{\tau}(H_{t,\ell,(j_1k_1)}), 
        \psi_{\tau}(H_{t+m,\ell,(j_2k_2)})\big).
\end{align*}
Then, it follows that
\begin{align*}
        I_4 
        \leq& \max_{1 \leq j_1,k_1,j_2,k_2 \leq d}\sum_{m = 1-2S}^{2S-1}\frac{3|m|}{2S}\Big|\cov\big(\psi_{\tau}(H_{t,\ell,(j_1k_1)}), 
        \psi_{\tau}(H_{t+m,\ell,(j_2k_2)})\big) \Big|\\
        &+ \max_{1 \leq j_1,k_1,j_2,k_2 \leq d}\sum_{|m| \geq 2S}\Big|\cov\big(\psi_{\tau}(H_{t,\ell,(j_1k_1)}), 
        \psi_{\tau}(H_{t+m,\ell,(j_2k_2)})\big)\Big|.
    \end{align*}
The lag-$m$ autocovariance can be upper bounded as
\begin{align*}
    &\Big|\cov\big(\psi_{\tau}(H_{t,\ell,(j_1k_1)}), 
        \psi_{\tau}(H_{t+m,\ell,(j_2k_2)})\big)\Big|\\
    \leq& \bigg|\sum_{h = 0}^{\infty}\mathbb{E}\Big[\mathcal{P}_{-h}\big(\psi_{\tau}(H_{0,\ell,(j_1k_1)})\big)\mathcal{P}_{-h}\big(\psi_{\tau}(H_{m,\ell,(j_2k_2)})\big)\bigg|\\
    \leq& \sum_{h = 0}^{\infty}\Big|\mathbb{E}\Big[\mathcal{P}_{-h}\big(\psi_{\tau}(H_{0,\ell,(j_1k_1)})\big)\mathcal{P}_{-h}\big(\psi_{\tau}(H_{m,\ell,(j_2k_2)})\big)\Big|\\
    \leq& 4\omega_4^2\sum_{h = 0}^{\infty}\delta_{h, 4}\delta_{h+m, 4}.
\end{align*}
It then follows that
\begin{align*}
        &\sum_{|m| \geq 2S}\Big|\cov\big(\psi_{\tau}(H_{t,\ell,(j_1k_1)}), 
        \psi_{\tau}(H_{t+m,\ell,(j_2k_2)})\big)\Big|\\
        \leq& 8\omega_4^2\sum_{m = 2S}^{\infty}\sum_{h = 0}^{\infty}\delta_{h, 4}\delta_{h+m, 4}\\
        \leq& 8\omega_4^2\sum_{h = 0}^{\infty}\delta_{h, 4}\sum_{f = 2S}^{\infty}\delta_{f, 4} \leq 8\omega_4^2 \Vert X_{\cdot} \Vert_4^2\exp(-2cS),
    \end{align*}
and
\begin{align*}
        &\sum_{m = -\infty}^{\infty}|m|\Big|\cov\big(\psi_{\tau}(H_{t,\ell,(j_1k_1)}), 
        \psi_{\tau}(H_{t+m,\ell,(j_2k_2)})\big) \Big|\\
        \leq& 8\omega_4^2\sum_{m = 1}^{\infty}\sum_{n = m}^{\infty}\sum_{h = 0}^{\infty}\delta_{h, 4}\delta_{h+n, 4} \leq 8\omega_4^2\sum_{h = 0}^{\infty}\delta_{h, 4}\sum_{m = 1}^{\infty}\sum_{n = m}^{\infty}\delta_{n, 4}\\
        \leq& 8\omega_4^2 \Vert X_{\cdot} \Vert_4^2 \sum_{m = 1}^{\infty}\exp(-cm) < \infty.
    \end{align*}
Thus, we have that
\begin{align}\label{eq:ub_Eo_H0_I4}
        \mathbb{E}[I_4] \leq C_7\{\exp(-2cS) + S^{-1}\}.
    \end{align}
Combining \eqref{eq:ub_Eo_H0}, \eqref{eq:ub_Eo_H0_I1}, \eqref{eq:ub_Eo_H0_I2}, \eqref{eq:ub_Eo_H0_I3} and \eqref{eq:ub_Eo_H0_I4} together, we have that
\begin{align*}
    \mathbb{E}[e_{D}] \leq& 2C_3S^{1/2} \exp(-cS) + C_6\bigg(\sqrt{\frac{\tau S^{1/2}\log d}{R}} + \frac{\tau\log d}{R}\bigg)+C_7\{\exp(-2cS) + S^{-1}\}\\
    \leq& C_8n^{- 82/399 + 4\beta/21} < C_8n^{(\beta - 1)/5},
    \end{align*}
where the second inequality follows by setting $S = n^{82/399 - 4\beta/21} > n^{(1-\beta)/5}$ which balances the terms $\sqrt{\tau S^{1/2}\log d/R}$ and $S^{-1}$. Hence, by Markov's inequality, we have
\begin{align}\label{eq:size_tail}
    \mathbb{P}(e_{D} > \vartheta) < n^{(\beta - 1)/5}\vartheta^{-1}.
\end{align}
Following \eqref{eq:size_UB}, \eqref{eq:size_gaussian_approx}, \eqref{eq:size_tail} and the definition of $\pi(\vartheta)$, we have that
    \begin{align*}
    \sup_{\alpha \in (0,1)}\Big|\mathbb{P}\Big(T_{\ell} \leq q^*(\alpha)\Big) - \alpha\Big| \leq& C_{9}\big(n^{-2(1/19 - \beta)/3} + \vartheta^{1/3}n^{2\beta/3} + n^{(\beta-1)/5}\vartheta^{-1}\big)\\
    \leq& C_{10}n^{-2(1/19 - \beta)/3},
    \end{align*}
where the second inequality is obtained by letting $\theta = n^{-(7\beta+3)/20}$.
\end{proof}

\section{Proof of \Cref{thm:GA_element-wise-trunc-mean}}
\begin{proof}[Proof of Theorem \ref{thm:GA_element-wise-trunc-mean}]
We have for all $\eta > 0$, 
\begin{equation*}
\begin{aligned}
    \rho_{n,\tau}^{\diamond} =& \sup_{t \in \mathbb{R}}\big|\mathbb{P}(\sqrt{n}|\widehat{\bm{\mu}}^{\tau} - \bm{\mu}|_{\infty} \leq t) - \mathbb{P}(\sqrt{n}|\widehat{\bm{\mu}}^{\tau} - \mathbb{E}\widehat{\bm{\mu}}^{\tau}|_{\infty} \leq t) \big|\\
    \leq& \mathbb{P}(\sqrt{n}|\bm{\mu} - \mathbb{E}\widehat{\bm{\mu}}^{\tau}|_{\infty} > \eta) + \sup_{t \in \mathbb{R}}\mathbb{P}\big(\big|\sqrt{n}|\widehat{\bm{\mu}}^{\tau} - \mathbb{E}\widehat{\bm{\mu}}^{\tau}|_{\infty} - t\big| \leq \eta\big)\\
    \leq& \mathbb{P}(\sqrt{n}|\bm{\mu} - \mathbb{E}\widehat{\bm{\mu}}^{\tau}|_{\infty} > \eta) + \rho_{n,\tau}^* + \sup_{t \in \mathbb{R}}\mathbb{P}\big(\big||\bm{Z}^{\tau}|_{\infty} - t\big| \leq \eta\big),
\end{aligned}
\end{equation*}
where the second line is due to the triangle inequality, i.e. $|\widehat{\bm{\mu}}^{\tau} - \bm{\mu}|_{\infty} - |\widehat{\bm{\mu}}^{\tau} - \mathbb{E}\widehat{\bm{\mu}}^{\tau}|_{\infty} \leq |\bm{\mu} - \mathbb{E}\widehat{\bm{\mu}}^{\tau}|_{\infty}$.
Under Assumption \ref{assumption:L2+theta}, we have the bound of the element-wise bias of the truncated estimator for all $j \in [d]$ as
\begin{equation*}
    \begin{aligned}
    &|\mu_j - \mathbb{E}\widehat{\mu}^{\tau}_j| = |\mathbb{E}X_{i,j} - \mathbb{E}\psi_{\tau}(X_{i,j})| = |\mathbb{E}(X_{i,j} - \tau)\mathbbm{1}_{\{X_{i,j} > \tau\}} + \mathbb{E}(X_{i,j} + \tau)\mathbbm{1}_{\{X_{i,j} < -\tau\}}|\\
    \leq& \mathbb{E}X_{i,j}\mathbbm{1}_{\{X_{i,j} > \tau\}} + \mathbb{E}(-X_{i,j})\mathbbm{1}_{\{X_{i,j} < -\tau\}} = \mathbb{E}|X_{i,j}|\mathbbm{1}_{\{|X_{i,j}| > \tau\}} \leq \tau^{-(1+\theta)}M_{\theta}.
    \end{aligned}
\end{equation*}
Let $\eta = \sqrt{n}\tau^{-(1+\theta)}M_{\theta}$, we have
\begin{equation*}
    \begin{aligned}
    \rho_{n,\tau}^{\diamond} \leq& 0 + \rho_{n,\tau}^* + \sup_{t \in \mathbb{R}}\mathbb{P}\big(\big||\bm{Z}^{\tau}|_{\infty} - t\big| \leq \sqrt{n}\tau^{-(1+\theta)}M_{\theta}\big)\\
    \lesssim& \rho_{n,\tau}^* + \sqrt{n}\tau^{-(1+\theta)}M_{\theta}\sqrt{\log d},
    \end{aligned}
\end{equation*}
where the last line is due to Lemma 2.1 in \cite{chernozhukov2013gaussian}, noting that $\bm{Z}^{\tau} \sim N(\bm{0}, \bm{\Sigma}^{\tau})$.

The term $\rho_{n,\tau}^{\circ}$ is the comparison between two distributions of Gaussian maxima, i.e. $|\bm{Z}^{\tau}|_{\infty}$ and $|\bm{Z}|_{\infty}$, with different covariance matrices. We have
\begin{equation*}
    \rho_{n,\tau}^{\circ} = \sup_{t \in \mathbb{R}}\big|\mathbb{P}(|\bm{Z}^{\tau}|_{\infty} \leq t) - \mathbb{P}(|\bm{Z}|_{\infty} \leq t) \big|.
\end{equation*}
Denote $\gamma_{\ell,(jk)}^{\tau} = \mathbb{E}\big[(\psi_{\tau}(X_{0,j})-\mathbb{E}\psi_{\tau}(X_{0,j}))(\psi_{\tau}(X_{\ell,k})-\mathbb{E}\psi_{\tau}(X_{\ell,k}))\big]$ and $\gamma_{\ell,(jk)} = \mathbb{E}\big[(X_{0,j}-\mu_j)(X_{\ell,k} -\mu_k)\big]$, then 
\begin{equation*}
    |\bm{\Sigma}^{\tau} - \bm{\Sigma}|_{\max} \leq \max_{j,k \in [d]}\sum_{\ell=-\infty}^{\infty}\big| \gamma_{\ell,(jk)}^{\tau} - \gamma_{\ell,(jk)}\big|.
\end{equation*}
To bound $\big| \gamma_{\ell,(jk)}^{\tau} - \gamma_{\ell,(jk)}\big|$, we use
the coupling method. 
Denote the random variable $\widetilde{X}_{i,k} = X_{i, \{0,-\infty\}, k}$ defined in \textup{(2)}. Note that $\widetilde{X}_{\ell,k}$ is independent of $\bm{X}_{0}$, but has the same marginal distribution as $X_{0,k}$. Under Assumptions \ref{assumption:L2+theta} and \ref{assumption:APP_decay.2order}, we have
\begin{equation*}
\begin{aligned}
    \big| \gamma_{j,k}^{\tau}(\ell) - \gamma_{j,k}(\ell)\big|
    \leq& \big| \mathbb{E}\big\{[\psi_{\tau}(X_{0,j}) - \mathbb{E}\psi_{\tau}(X_{0,j})][\psi_{\tau}(X_{\ell,k}) - \mathbb{E}\psi_{\tau}(X_{\ell,k}) - X_{\ell,k} + \mu_k]\big\}\big|\\
    & + \big|\mathbb{E}\big\{[\psi_{\tau}(X_{0,j}) - \mathbb{E}\psi_{\tau}(X_{0,j}) - X_{0,j} + \mu_j](X_{\ell,k} - \mu_k)\big\}\big|\\
    =& \big| \mathbb{E}\big\{[\psi_{\tau}(X_{0,j}) - \mathbb{E}\psi_{\tau}(X_{0,j})][\psi_{\tau}(X_{\ell,k}) - \psi_{\tau}(\widetilde{X}_{\ell,k}) - X_{\ell,k} + \widetilde{X}_{\ell,k}]\mathbbm{1}_{\{|X_{\ell,k}|> \tau\}}\big\}\big|\\
    & + \big|\mathbb{E}\big\{[\psi_{\tau}(X_{0,j}) - \mathbb{E}\psi_{\tau}(X_{0,j}) - X_{0,j} + \mu_j][X_{\ell,k} - \widetilde{X}_{\ell,k}]\mathbbm{1}_{\{|X_{0,k}|> \tau\}}\big\}\big|\\
    \leq& 2\Vert \psi_{\tau}(X_{0,j}) - \mathbb{E}\psi_{\tau}(X_{0,j}) \Vert_{2+\theta} \Vert X_{\ell,k} - \widetilde{X}_{\ell,k} \Vert_{2+\theta} \big[\mathbb{P}(|X_{\ell,k}|> \tau)\big]^{\theta/(2+\theta)}\\
    & + 2\Vert X_{0,j} - \mu_{j} \Vert_{2+\theta} \Vert X_{\ell,k} - \widetilde{X}_{\ell,k} \Vert_{2+\theta} \big[\mathbb{P}(|X_{0,k}|> \tau)\big]^{\theta/(2+\theta)}\\
    \leq& 4\Vert X_{0,j} \Vert_{2+\theta}\sum_{m = \ell}^{\infty}\delta_{m,2+\theta,j}\Vert X_{\ell,k} \Vert_{2+\theta}^{\theta}\tau^{-\theta}\\
    \leq& 4M_{\theta}^{(1+\theta)/(2+\theta)}\tau^{-\theta}\sum_{m = l}^{\infty}\delta_{m,2+\theta}.
\end{aligned}
\end{equation*}
Therefore, we have
\begin{equation*}
\begin{aligned}
    |\bm{\Sigma}^{\tau} - \bm{\Sigma}|_{\max} \leq 2\max_{j,k \in [d]}\sum_{\ell=0}^{\infty}\big| \gamma_{\ell,(jk)}^{\tau} - \gamma_{\ell,(jk)}\big| \leq 8M_{\theta}^{(1+\theta)/(2+\theta)}\Vert X_. \Vert_{2+\theta}\tau^{-\theta}.
\end{aligned}
\end{equation*}
By Lemma 3.1 in \cite{chernozhukov2013gaussian} and let $\pi(x) = x^{1/3}\big(1 \vee \log(d/x) \big)^{2/3}$ for $x > 0$, we have
\begin{equation*}
\begin{aligned}
    \rho_{n,\tau}^{\circ} = \sup_{t \in \mathbb{R}}\big|\mathbb{P}(|\bm{Z}^{\tau}|_{\infty} \leq t) - \mathbb{P}(|\bm{Z}|_{\infty} \leq t) \big|
    \lesssim \pi(M_{\theta}^{(1+\theta)/(2+\theta)}\Vert X_. \Vert_{2+\theta}\tau^{-\theta}) \lesssim \tau^{-\theta/3}(\log d)^{2/3}.
\end{aligned}
\end{equation*}

In the following, we analyze the term $\rho_{n,\tau}^*$.
We use the $m$-dependence approximation combined with the `big-and-small' blocking technique to study this term.
Consider the $m$-dependence sequence $\{\psi_{\tau,m}(\bm{X}_i)\}_{i \in \mathbb{Z}}$ where $\psi_{\tau,m}(\bm{X}_i) \coloneqq \mathbb{E}[\psi_{\tau}(\bm{X}_i)|\epsilon_i, \dots, \epsilon_{i-m}]$. Define 
\begin{equation*}
    \bm{T}_X^{\tau} = \sum_{i=1}^n[\psi_{\tau}(\bm{X}_i) - \mathbb{E}\psi_{\tau}(\bm{X}_i)] \quad \text{and} \quad \bm{T}_{X,m}^{\tau} = \sum_{i=1}^n[\psi_{\tau,m}(\bm{X}_i) - \mathbb{E}\psi_{\tau,m}(\bm{X}_i)].
\end{equation*}
Note that $\mathbb{E}\psi_{\tau}(\bm{X}_i) = \mathbb{E}\psi_{\tau,m}(\bm{X}_i)$. Let $M, m, \omega \in \mathbb{N}^+$, representing the size of big block, the size of the small block and the number of big blocks (or small blocks), respectively. Let $M \asymp n^{\alpha}$ with $\alpha \in (0,1/2)$ and $m \asymp \log n$. For simplicity, suppose $n = (M + m)\omega$. We divide the interval $[1,n]$ into alternating big blocks $L_b = [(b-1)(M+m)+1, bM+(b-1)m]$ and small blocks $S_b = [bM + (b-1)m + 1, b(M+m)]$, for $b \in [1,\omega]$.  For big blocks, define
\begin{equation*}
    \bm{Y}_b^{\tau} = \sum_{i \in L_b}[\psi_{\tau}(\bm{X}_i) - \mathbb{E}\psi_{\tau}(\bm{X}_i)], \;\; \bm{Y}_{b,m}^{\tau} = \sum_{i \in L_b}[\psi_{\tau,m}(\bm{X}_i) - \mathbb{E}\psi_{\tau,m}(\bm{X}_i)], \;\; \bm{T}_{Y}^{\tau} = \sum_{b=1}^{\omega}\bm{Y}_b^{\tau}, \;\; \bm{T}_{Y,m}^{\tau} = \sum_{b=1}^{\omega}\bm{Y}_{b,m}^{\tau}.
\end{equation*}
By construction, we have $\{\bm{Y}_{b,m}^{\tau}\}_{b \in [1,\omega]}$ are iid. Denote $\{\bm{Z}_b^{\tau}\}_{b \in [1,\omega]}$ be a sequence of iid random vectors following $N(\bm{0}, M\bm{B}^{\tau})$. Also, denote $\{\bm{Z}_{b,m}^{\tau}\}_{b \in [1,\omega]}$ be a sequence of iid random vectors following $N(\bm{0},M\widetilde{\bm{B}}^{\tau})$. The covariance matrices $\bm{B}^{\tau}$ and $\widetilde{\bm{B}}^{\tau}$ are respectively defined by
\begin{equation*}
    \bm{B}^{\tau} = (b_{ij}^{\tau})_{i,j \in [d]} = \cov(\bm{Y}_b^{\tau}/\sqrt{M}) \quad \text{and} \quad \widetilde{\bm{B}}^{\tau} = (\widetilde{b}_{ij}^{\tau})_{i,j \in [d]} = \cov(\bm{Y}_{b,m}^{\tau}/\sqrt{M}).
\end{equation*}
Define also $\bm{T}_{Z,m}^{\tau} = \sum_{b=1}^{\omega}\bm{Z}_{b,m}^{\tau}$. 

\noindent By the triangle inequality and an elementary inequality, We have
\begin{equation*}
    \mathbb{P}(|\bm{T}_{X}^{\tau} - \bm{T}_{Y,m}^{\tau}|_{\infty} \geq y) \leq \mathbb{P}(|\bm{T}_{X}^{\tau} - \bm{T}_{X,m}^{\tau}|_{\infty} \geq y/2) + \mathbb{P}(|\bm{T}_{X,m}^{\tau} - \bm{T}_{Y,m}^{\tau}|_{\infty} \geq y/2) ,
\end{equation*}
In the following analysis, we start with the first term on the right hand side. In order to apply Theorem 2.5, we verify the following conditions regarding 
\begin{equation*}
    \bm{T}_{X}^{\tau} - \bm{T}_{X,m}^{\tau} = \sum_{i=1}^n[\psi_{\tau}(\bm{X}_i)-\psi_{\tau,m}(\bm{X}_i)].
\end{equation*}
First, we have for any $j \in [d]$
\begin{equation*}
    |\psi_{\tau}(X_{i,j})-\psi_{\tau,m}(X_{i,j})| \leq 2\tau.
\end{equation*}
Second, under Assumption \ref{assumption:APP_decay.2order} we have for all $m \geq 0$ and for any $j \in [d]$
\begin{equation*}
    \sup_{k \geq 0}\rho^{-k}\sum_{i=k}^{\infty}\big\Vert [\psi_{\tau}(X_{i,j})-\psi_{\tau,m}(X_{i,j})] - [\psi_{\tau}^{\prime}(X_{i,j})-\psi_{\tau,m}^{\prime}(X_{i,j})] \big\Vert_2 < \infty,
\end{equation*}
then, we write $\psi_{\tau}(X_{i,j})-\psi_{\tau,m}(X_{i,j}) = \sum_{k=m}^{\infty}[\psi_{\tau,k+1}(X_{i,j})-\psi_{\tau,k}(X_{i,j})]$, by the Burkholder's inequality
\begin{equation}\label{eq:distance_m-dep-approx}
    \Big\Vert \sum_{i=1}^n[\psi_{\tau}(X_{i,j})-\psi_{\tau,m}(X_{i,j})] \Big\Vert_{2} \leq \sum_{k=m
    }^{\infty}\Big\Vert \sum_{i=1}^n[\psi_{\tau,k+1}(X_{i,j})-\psi_{\tau,k}(X_{i,j})] \Big\Vert_2 \lesssim \sqrt{n}\Delta_{m,2} \leq \sqrt{n}\rho^m\Vert X_.\Vert_2.
\end{equation}
Then, we have
\begin{equation*}
    \nu^2 = \max_{j \in [d]}\lim_{n \to \infty}n^{-1}\Big\Vert \sum_{i=1}^n[\psi_{\tau}(X_{i,j})-\psi_{\tau,m}(X_{i,j})] \Big\Vert_{2}^2 \lesssim \rho^{2m}\Vert X_.\Vert_2^2.
\end{equation*}
By Theorem 2.5, we have
\begin{equation*}
    \mathbb{P}(|\bm{T}_{X}^{\tau} - \bm{T}_{X,m}^{\tau}|_{\infty} \geq y/2) \leq d\exp\Big(- \frac{y^2}{16C_1(nC_3\rho^{2m}\Vert X_. \Vert_2^2 + 4\tau^2) + 8C_2\tau(\log n)^2y} \Big).
\end{equation*}
Next, we consider $\bm{T}_{X,m}^{\tau} - \bm{T}_{Y,m}^{\tau} = \sum_{b=1}^{\omega}\sum_{i \in S_b}[\psi_{\tau,m}(\bm{X}_i) - \mathbb{E}\psi_{\tau,m}(\bm{X}_i)]$, where $\sum_{i \in S_b}[\psi_{\tau,m}(\bm{X}_i) - \mathbb{E}\psi_{\tau,m}(\bm{X}_i)]$ for $b \in [1,\omega]$ are iid. By the Bernstein's inequality for sum of iid random variables, we have
\begin{equation*}
    \mathbb{P}(|\bm{T}_{X,m}^{\tau} - \bm{T}_{Y,m}^{\tau}|_{\infty} \geq y/2) \leq d\exp\Big(- \frac{y^2}{8\omega m\Vert X_.\Vert_2^2 + 8/3m\tau y} \Big).
\end{equation*}
Therefore, we have
\begin{equation*}
    \begin{aligned}
    \mathbb{P}(|\bm{T}_{X}^{\tau} - \bm{T}_{Y,m}^{\tau}|_{\infty} \geq y) \leq& d\exp\Big(- \frac{y^2}{16C_1(nC_3\rho^{2m}\Vert X_.\Vert_2^2 + 4\tau^2) + 8C_2\tau(\log n)^2y} \Big)\\
    &+ d\exp\Big(- \frac{y^2}{8\omega m\Vert X_.\Vert_2^2 + 8/3m\tau y} \Big).
    \end{aligned}
\end{equation*}

Next, we consider the Kolmogorov-Smirnov distance between $\bm{T}_{Y,m}^{\tau}$ and $\bm{T}_{Z,m}^{\tau}$, i.e.
\begin{equation}
\label{eq:GA.trunc.m-dep}
    \sup_{t \in \mathbb{R}}\big| \mathbb{P}(|\bm{T}_{Y,m}^{\tau}/\sqrt{n}|_{\infty} \leq t) - \mathbb{P}(|\bm{T}_{Z,m}^{\tau}/\sqrt{n}|_{\infty} \leq t) \big|.
\end{equation}
We apply Theorem 2.1 in \cite{chernozhukov2017central} to obtain the decay rate of (\ref{eq:GA.trunc.m-dep}). Denote $Y_{bj,m}^{\tau} = \sum_{i \in L_b}[\psi_{\tau,m}(X_{i,j}) - \mathbb{E}\psi_{\tau,m}(X_{i,j})]$ for $j \in [d]$, the $j$-th coordinate of $\bm{Y}_{b,m}^{\tau}$. First, we need to verify that $\min_{j \in [d]}\mathbb{E}(Y_{bj,m}^{\tau}/\sqrt{M+m})^2 > 0$.
By triangle inequality, we have
\begin{equation}\label{eq:verify_var_nondegen}
    \begin{aligned}
    &\big|\mathbb{E}(Y_{bj,m}^{\tau}/\sqrt{M+m})^2 - \mathbb{E}\big[\sum_{i\in L_b}(X_{i,j} - \mu_j)/\sqrt{M}\big]^2\big|\\
    \leq&\big|\mathbb{E}(Y_{bj,m}^{\tau}/\sqrt{M+m})^2 - \mathbb{E}(Y_{bj}^{\tau}/\sqrt{M+m})^2\big| + \Big|\mathbb{E}(Y_{bj}^{\tau}/\sqrt{M+m})^2 - \mathbb{E}\Big[\sum_{i\in L_b}(X_{i,j} - \mu_j)/\sqrt{M+m}\Big]^2\Big|\\
    &+\Big|\mathbb{E}\big[\sum_{i\in L_b}(X_{i,j} - \mu_j)/\sqrt{M+m}\big]^2 - \mathbb{E}\big[\sum_{i\in L_b}(X_{i,j} - \mu_j)/\sqrt{M}\big]^2\Big|\\
    =& I + II + III.
    \end{aligned}
\end{equation}
For the term $I$, by H\"older's inequality and the similar argument as in (\ref{eq:distance_m-dep-approx}), we have, as $M,m \to \infty$,
\begin{equation}\label{eq:verify_var_nondegen_1}
    \begin{aligned}
    &I \leq \frac{1}{M+m}\big|\mathbb{E}[(Y_{bj,m}^{\tau} + Y_{bj}^{\tau})(Y_{bj,m}^{\tau} - Y_{bj}^{\tau})]\big|
    \leq \frac{2\sqrt{M}\Vert X_. \Vert_2}{M+m}\Vert Y_{bj,m}^{\tau} - Y_{bj}^{\tau}\Vert_2 \leq 2\frac{M}{M+m}\Vert X_. \Vert_2^2\rho^m \to 0\\
    \end{aligned}
\end{equation}
For the term $II$, by triangle, the H\"older's and the Burkholder's inequalities, we have as $M, m \to \infty$
\begin{equation*}
    \begin{aligned}
    &II\\
    \leq& \frac{1}{M+m}\Big\Vert \sum_{i\in L_b}\big[(X_{i,j} - \mu_j) - (\psi_{\tau}(X_{i,j}) - \mathbb{E}\psi_{\tau}(X_{i,j}))\big]\Big\Vert_2\Big\Vert\sum_{i\in L_b}\big[(X_{i,j} - \mu_j) + (\psi_{\tau}(X_{i,j}) - \mathbb{E}\psi_{\tau}(X_{i,j}))\big] \Big\Vert_2\\
    \leq& \frac{4M\Vert X_.\Vert_2^2}{M+m} \to 0.
    \end{aligned}
\end{equation*}
Moreover, as $M,m \to \infty$, we have
\begin{equation}\label{eq:verify_var_nondegen_3}
    \begin{aligned}
    III = \frac{m}{M(M+m)}\mathbb{E}\big[\sum_{i\in L_b}(X_{i,j} - \mu_j)\big]^2 \leq \frac{mM\Vert X_.\Vert_2^2}{M(M+m)} \to 0.
    \end{aligned}
\end{equation}
Therefore, under Assumption \ref{assumption:blockvar_nondegen}, and combining (\ref{eq:verify_var_nondegen})-(\ref{eq:verify_var_nondegen_3}), we obtain
\begin{equation*}
    \min_{j\in [d]}\mathbb{E}(Y_{bj,m}^{\tau}/\sqrt{M+m})^2 \geq b + o(1) > 0.
\end{equation*}

Next, we adopt some quantities defined in Theorem 2.1 in \cite{chernozhukov2017central}. Let $\tau$ be 
\begin{equation}
\label{eq:2Mtau}
    2\tau = \frac{\sqrt{n}}{4\sqrt{M}\phi\log d},
\end{equation}
where $\phi = C\Big( \frac{L_n^2(\log d)^4}{n} \Big)^{-1/6}$ with $L_n = \max_{j \in [d]}\mathbb{E}\big|Y_{bj,m}^{\tau}/\sqrt{M+m}\big|^3$ defined also in Theorem 2.1 in \cite{chernozhukov2017central}. Thus, we have $(M+m)^{-1/2}\max_{j \in [d]}Y_{bj,m}^{\tau} \leq 2M(M+m)^{-1/2}\tau < \sqrt{n}/(4\phi\log d)$.

The term $L_n$ can be bounded as follows. Write
\begin{equation*}
    \psi_{\tau,m}(X_{i,j}) - \mathbb{E}\psi_{\tau,m}(X_{i,j}) = \sum_{k=0}^m\mathcal{P}_{i-k}\psi_{\tau}(X_{i,j}),
\end{equation*}
then, by the Burkholder's inequality
\begin{align*}
    \Big\Vert \sum_{i=1}^M\mathcal{P}_{i-k}\psi_{\tau}(X_{i,j}) \Big\Vert_{3}^2 &\leq C\sum_{i=1}^M\Vert \mathcal{P}_{i-k}\psi_{\tau}(X_{i,j}) \Vert_3^2 \leq C\sum_{i=1}^M\Vert \psi_{\tau}(X_{k,j}) - \psi_{\tau}(X_{k,j}^*) \Vert_3^2\\
    &\leq CM(2\tau)^{2(1-\theta)/3}\delta_{k,2+\theta,j}^{2(2+\theta)/3},
\end{align*}
and
\begin{align*}
    \Vert Y_{bj,m}^{\tau} \Vert_{3} &\leq \sum_{k=0}^m\Big\Vert \sum_{i=1}^M\mathcal{P}_{i-k}\psi_{\tau}(X_{i,j}) \Big\Vert_{3} \leq C^{1/2}M^{1/2}(2\tau)^{(1-\theta)/3} \sum_{k=0}^m\delta_{k,2+\theta,j}^{(2+\theta)/3}\\
    &\leq C^{1/2}M^{1/2}(2\tau)^{(1-\theta)/3} \sum_{k=0}^{\infty}\delta_{k,2+\theta,j}^{(2+\theta)/3}.
\end{align*}
Under GMC($2+\theta$), we have $\sum_{k=0}^{\infty}\delta_{k,2+\theta,j}^{(2+\theta)/3} < \infty$, thus
\begin{equation}
\label{eq:Ln}
    L_n = \max_{j \in [d]}\mathbb{E}\big|Y_{bj,m}^{\tau}/\sqrt{M+m}\big|^3 = (M+m)^{-3/2}\max_{j \in [d]}\mathbb{E}|Y_{bj,m}^{\tau}|^3 \lesssim (M+m)^{-3/2} M^{3/2} \tau^{1-\theta} \leq \tau^{1-\theta}.
\end{equation}
Therefore, we can set
\begin{equation}
\label{eq:tau}
    \tau \asymp \Big(\frac{n}{M^{3/2}\log d}\Big)^{\frac{1}{2+\theta}},
\end{equation}
which implies that (\ref{eq:2Mtau}) is satisfied. Moreover, we have
\begin{equation*}
    M_{n,Y}(\phi) = \mathbb{E}\Big[\max_{j \in [d]} \Big|\frac{1}{\sqrt{M+m}}Y_{bj,m}^{\tau}\Big|^3\mathbbm{1}\Big\{(M+m)^{-1/2}\max_{j \in [d]} |Y_{bj,m}^{\tau}| > \sqrt{n}/(4\phi\log d)\Big\} \Big] = 0,
\end{equation*}
and
\begin{equation*}
    \begin{aligned}
    &M_{n,Z}(\phi) = \mathbb{E}\Big[\max_{j \in [d]} \Big|\frac{1}{\sqrt{M+m}}Z_{bj,m}^{\tau}\Big|^3\mathbbm{1}\Big\{\max_{j \in [d]} |Z_{bj,m}^{\tau}| > \sqrt{n}/(4\phi\log d)\Big\} \Big]\\
    =& (M+m)^{-3/2}\mathbb{E}\Big[|\bm{Z}_{b,m}^{\tau}|_{\infty}^3\mathbbm{1}\big\{(M+m)^{-1/2}|\bm{Z}_{b,m}^{\tau}|_{\infty} > \sqrt{n}/(4\phi\log d)\big\} \Big]\\
    \leq& (M+m)^{-3/2}\mathbb{E}\Big[|\bm{Z}_{b,m}^{\tau}|_{\infty}^3\mathbbm{1}\big\{|\bm{Z}_{b,m}^{\tau}|_{\infty} > 2M\tau\big\} \Big]\\
    \lesssim& M^3\tau^3(M+m)^{-3/2}\mathbb{P}\big(|\bm{Z}_{b,m}^{\tau}|_{\infty} > 2M\tau\big) + (M+m)^{-3/2}\int_{2M\tau}^{\infty}\mathbb{P}\big(|\bm{Z}_{b,m}^{\tau}|_{\infty} > x)x^2dx.
    \end{aligned}
\end{equation*}
Since $\bm{Z}_{b,m} \sim N(\bm{0}, M\widetilde{\bm{B}}^{\tau})$, we have $\mathbb{P}\big(|\bm{Z}_{b,m}^{\tau}|_{\infty} > x\big) \leq d\exp\Big(-\frac{Cx^2}{MM_{\theta}}\Big)$. Then
\begin{equation*}
    \mathbb{P}\big(|\bm{Z}_{b,m}^{\tau}|_{\infty} > 2M\tau\big) \leq d\exp\Big(-\frac{CM\tau^2}{M_{\theta}}\Big),
\end{equation*}
and
\begin{equation*}
    \int_{2M\tau}^{\infty}\mathbb{P}\big(|\bm{Z}_{b,m}^{\tau}|_{\infty} > x)x^2dx \lesssim dM^2M_{\theta}\tau\exp(-CM\tau^2/M_{\theta}) + dM^{3/2}M_{\theta}^{3/2}\exp(-CM\tau^2/M_{\theta}).
\end{equation*}
Then, we have
\begin{align*}
    M_{n,Z}(\phi) \lesssim& M^3\tau^3(M+m)^{-3/2}\mathbb{P}\big(|\bm{Z}_{b,m}^{\tau}|_{\infty} > 2M\tau\big) + (M+m)^{-3/2}\int_{2M\tau}^{\infty}\mathbb{P}\big(|\bm{Z}_{b,m}^{\tau}|_{\infty} > x)x^2dx\\
    \lesssim& \big[M^3\tau^3 + M^2M_{\theta}\tau + M^{3/2}M_{\theta}^{3/2}\big]d(M+m)^{-3/2}\exp(-CM\tau^2/M_{\theta})\\
    \lesssim& M^3\tau^3 d(M+m)^{-3/2}\exp(-CM\tau^2/M_{\theta}) \leq dM^{3/2}\tau^3\exp(-CM\tau^2/M_{\theta}),
\end{align*}
and,
\begin{align*}
    M_{n}(\phi) = M_{n,Y}(\phi) + M_{n,Z}(\phi) \lesssim dM^{3/2}\tau^3\exp(-CM\tau^2/M_{\theta}).
\end{align*}
Under the assumption $M^{(4\theta-1)/3}(\log d)^{4+3\theta} = o(n^{\theta})$, we have
\begin{equation}
\label{eq:GA.trunc.m-dep.rate}
\begin{aligned}
    &\sup_{t \in \mathbb{R}}\big| \mathbb{P}(|\bm{T}_{Y,m}^{\tau}/\sqrt{n}|_{\infty} \leq t) - \mathbb{P}(|\bm{T}_{Z,m}^{\tau}/\sqrt{n}|_{\infty} \leq t) \big|\\
    &\lesssim \Big(\frac{\tau^{2(1-\theta)}(\log d)^7}{\omega} \Big)^{1/6} + dM^{3/2}\tau^{3}\exp(-CM\tau^2/M_{\theta})\\
    &\lesssim \bigg[\frac{M^{(4\theta-1)/3}(\log d)^{(4+3\theta)}}{n^{\theta}}\bigg]^{1/(4+2\theta)} + d\bigg[\frac{M^{(3+3\theta)/2}n^{3}}{(\log d)^{3}}\bigg]^{1/(2+\theta)}\exp\Bigg\{-C/M_{\theta}\bigg[\frac{M^{(1+2\theta)/2}n^{2}}{(\log d)^{2}}\bigg]^{1/(2+\theta)}\Bigg\}\\
    &\lesssim \bigg[\frac{M^{(4\theta-1)/3}(\log d)^{(4+3\theta)}}{n^{\theta}}\bigg]^{1/(4+2\theta)} + \frac{1}{n}.
\end{aligned}
\end{equation}

Then, we consider the Kolmogorov-Smirnov distance between $\bm{T}_{Z,m}^{\tau}$ and $\bm{Z}^{\tau}$, 
\begin{align*}
    \sup_{t \in \mathbb{R}}\big| \mathbb{P}(|\bm{T}_{Z,m}^{\tau}/\sqrt{n}|_{\infty} \leq t) - \mathbb{P}(|\bm{Z}^{\tau}|_{\infty} \leq t) \big|.
\end{align*}
Recall that $\bm{T}_{Z,m}^{\tau} = \sum_{b=1}^{\omega}\bm{Z}_{b,m}^{\tau}$ with $\bm{Z}_{b,m}^{\tau}$ be iid $N(\bm{0},M\widetilde{\bm{B}}^{\tau})$, and $\bm{Z}^{\tau} \sim N(\bm{0}, \bm{\Sigma}^{\tau})$.
We have
\begin{align*}
    \cov(\bm{T}_{Z,m}^{\tau}/\sqrt{n}) = \frac{M\omega}{n}\widetilde{\bm{B}}^{\tau}.
\end{align*}
And the difference of the covariance matrix in max norm can be bounded by
\begin{align*}
    \Big|\frac{M\omega}{n}\widetilde{\bm{B}}^{\tau} - \bm{\Sigma}^{\tau}\Big|_{\max} \leq \frac{M\omega}{n}|\widetilde{\bm{B}}^{\tau} - \bm{B}^{\tau}|_{\max} + \frac{M\omega}{n}|\bm{B}^{\tau} - \bm{\Sigma}^{\tau}|_{\max} + \Big(1 - \frac{M\omega}{n}\Big)|\bm{\Sigma}^{\tau}|_{\max},
\end{align*}
where
\begin{align*}
    |b_{jk}^{\tau} - \widetilde{b}_{jk}^{\tau}| =& \frac{1}{M}\big|\mathbb{E}(Y_{bj}^{\tau}Y_{bk}^{\tau} - Y_{bj,m}^{\tau}Y_{bk,m}^{\tau})\big| = \frac{1}{M}\big|\mathbb{E}[Y_{bj}^{\tau}(Y_{bk}^{\tau} - Y_{bk,m}^{\tau}) + (Y_{bj}^{\tau}- Y_{bj,m}^{\tau})Y_{bk,m}^{\tau})]\big|\\
    \leq& \frac{1}{M}\Vert Y_{bj}^{\tau} \Vert_2 \Vert Y_{bk}^{\tau} - Y_{bk,m}^{\tau}\Vert_2 + \frac{1}{M}\Vert Y_{bj}^{\tau}- Y_{bj,m}^{\tau}\Vert_2 \Vert Y_{bk,m}^{\tau}\Vert_2\\
    \leq& 2\rho^{m+1}\Vert X_.\Vert_2^2,
\end{align*}
and
\begin{align*}
    |\sigma_{jk}^{\tau} - b_{jk}^{\tau}| =& \Big|\sum_{|\ell| > M}\gamma_{\ell,(jk)}^{\tau} + \sum_{\ell = -M}^{M}\frac{|\ell|}{M}\gamma_{\ell,(jk)}^{\tau}\Big| \leq \sum_{|\ell| > M}|\gamma_{\ell,(jk)}^{\tau}| + \sum_{\ell = -M}^{M}\frac{|\ell|}{M}|\gamma_{\ell,(jk)}^{\tau}|\\
    \leq& 2\sum_{\ell = M+1}^{\infty}\sum_{h = 0}^{\infty}\delta_{h,2,j}\delta_{h+\ell,2,k} + \frac{2}{M}\sum_{\ell = 1}^{M}\sum_{k=\ell}^{M}\sum_{h = 0}^{\infty}\delta_{h,2,j}\delta_{h+k,2,k}\\
    \leq& 2\Delta_{0,2,j}\Delta_{M+1,2,k} + \frac{2}{M}\Delta_{0,2,j}\sum_{\ell = 1}^{M}\Delta_{\ell,2,k}\\
    \leq& 2\Vert X_.\Vert_2^2\rho^{M+1} + \frac{2\rho(1-\rho^M)}{M(1-\rho)}\Vert X_.\Vert_2^2,
\end{align*}
and
\begin{align*}
    |\sigma_{jk}^{\tau}| \leq \max_{j\in[d]}\sigma_{jj} \leq \Vert X_.\Vert_2^2.
\end{align*}
Therefore, we have
\begin{align*}
    &\Big|\frac{M\omega}{n}\widetilde{\bm{B}}^{\tau} - \bm{\Sigma}^{\tau}\Big|_{\max} \leq \frac{M\omega}{n}|\widetilde{\bm{B}}^{\tau} - \bm{B}^{\tau}|_{\max} + \frac{M\omega}{n}|\bm{B}^{\tau} - \bm{\Sigma}^{\tau}|_{\max} + \Big(1 - \frac{M\omega}{n}\Big)|\bm{\Sigma}^{\tau}|_{\max}\\
    \leq& \frac{2M\omega\Vert X_. \Vert_2^2(\rho^{m+1}+\rho^{M+1})}{n} + \frac{2\omega\Vert X_. \Vert_2^2\rho(1-\rho^M)}{(1-\rho)n} + \frac{m\omega\Vert X_.\Vert_2^2}{n}.
\end{align*}
By Lemma 3.1 in \cite{chernozhukov2013gaussian} and recall the big block size $M \asymp n^{\alpha}$ with $\alpha \in (0,1)$ and the small block size $m \asymp \log n$, we have
\begin{equation}
\label{eq:gaussian.maxima.rate}
\begin{aligned}
    &\sup_{t \in \mathbb{R}}\big| \mathbb{P}(|\bm{T}_{Z,m}^{\tau}/\sqrt{n}|_{\infty} \leq t) - \mathbb{P}(|\bm{Z}^{\tau}|_{\infty} \leq t) \big|\\
    \lesssim& \pi\Big(\frac{2M\omega\Vert X_. \Vert_2^2(\rho^{m+1}+\rho^{M+1})}{n} + \frac{2\omega\Vert X_. \Vert_2^2\rho(1-\rho^M)}{(1-\rho)n} + \frac{m\omega\Vert X_.\Vert_2^2}{n} \Big)\\
    \lesssim& \frac{(\log d)^{2/3}}{M^{1/3}}.
\end{aligned}
\end{equation}
Combining (\ref{eq:GA.trunc.m-dep.rate}) and (\ref{eq:gaussian.maxima.rate}), we have
\begin{align*}
    &\sup_{t \in \mathbb{R}}\big| \mathbb{P}(|\bm{T}_{Y,m}^{\tau}/\sqrt{n}|_{\infty} \leq t) - \mathbb{P}(|\bm{Z}^{\tau}|_{\infty} \leq t) \big|\\
    \leq& \sup_{t \in \mathbb{R}}\big| \mathbb{P}(|\bm{T}_{Y,m}^{\tau}/\sqrt{n}|_{\infty} \leq t) - \mathbb{P}(|\bm{T}_{Z,m}^{\tau}/\sqrt{n}|_{\infty} \leq t) \big| + \sup_{t \in \mathbb{R}}\big| \mathbb{P}(|\bm{T}_{Z,m}^{\tau}/\sqrt{n}|_{\infty} \leq t) - \mathbb{P}(|\bm{Z}^{\tau}|_{\infty} \leq t) \big|\\
    \lesssim& \bigg[\frac{M^{(4\theta-1)/3}(\log d)^{(4+3\theta)}}{n^{\theta}}\bigg]^{1/(4+2\theta)} + \frac{1}{n} + \frac{(\log d)^{2/3}}{M^{1/3}}.
\end{align*}
For any $\kappa > 0$, we have
\begin{align*}
    &\sup_{t \in \mathbb{R}}\big| \mathbb{P}(|\bm{T}_{X}^{\tau}/\sqrt{n}|_{\infty} \leq t) - \mathbb{P}(|\bm{T}_{Y,m}^{\tau}/\sqrt{n}|_{\infty} \leq t) \big|\\
    \leq& \mathbb{P}(|\bm{T}_{X}^{\tau}/\sqrt{n} - \bm{T}_{Y,m}^{\tau}/\sqrt{n}|_{\infty} > \kappa) + \sup_{t \in \mathbb{R}}\mathbb{P}(||\bm{T}_{Y,m}^{\tau}/\sqrt{n}|_{\infty} - t| \leq \kappa)\\
    \leq& \mathbb{P}(|\bm{T}_{X}^{\tau}/\sqrt{n} - \bm{T}_{Y,m}^{\tau}/\sqrt{n}|_{\infty} > \kappa) + \sup_{t \in \mathbb{R}}\big| \mathbb{P}(|\bm{T}_{Y,m}^{\tau}/\sqrt{n}|_{\infty} \leq t) - \mathbb{P}(|\bm{T}_{Z,m}^{\tau}/\sqrt{n}|_{\infty} \leq t) \big|\\
    &+ \sup_{t \in \mathbb{R}}\mathbb{P}(||\bm{T}_{Z,m}^{\tau}/\sqrt{n}|_{\infty} - t| \leq \kappa)\\
    \leq& \mathbb{P}(|\bm{T}_{X}^{\tau} - \bm{T}_{Y,m}^{\tau}|_{\infty} > \sqrt{n}\kappa) + \sup_{t \in \mathbb{R}}\big| \mathbb{P}(|\bm{T}_{Y,m}^{\tau}/\sqrt{n}|_{\infty} \leq t) - \mathbb{P}(|\bm{T}_{Z,m}^{\tau}/\sqrt{n}|_{\infty} \leq t) \big| + \kappa\sqrt{\log d}.
\end{align*}
Let $\kappa \asymp n^{-\gamma \beta}$ with $\gamma > 1/2$. Set $\alpha$, $\beta$ and $\gamma$ such that (a) $\alpha > (1+2\gamma)\beta$ and (b) $3\alpha + \theta > (10 + 6\theta + \gamma)\beta$. In summary, we have
\begin{align*}
    &\rho^*_{n,\tau} = \sup_{t \in \mathbb{R}}\big| \mathbb{P}(|\bm{T}_{X}^{\tau}/\sqrt{n}|_{\infty} \leq t) - \mathbb{P}(|\bm{Z}^{\tau}|_{\infty} \leq t) \big|\\
    \leq& \sup_{t \in \mathbb{R}}\big| \mathbb{P}(|\bm{T}_{X}^{\tau}/\sqrt{n}|_{\infty} \leq t) - \mathbb{P}(|\bm{T}_{Y,m}^{\tau}/\sqrt{n}|_{\infty} \leq t) \big| + \sup_{t \in \mathbb{R}}\big| \mathbb{P}(|\bm{T}_{Y,m}^{\tau}/\sqrt{n}|_{\infty} \leq t) - \mathbb{P}(|\bm{Z}^{\tau}|_{\infty} \leq t) \big|\\
    \lesssim& d\exp\Big(- \frac{n\kappa^2}{16C_1(nC_3\Vert X_.\Vert_2^2\rho^{2m} + 4\tau^2) + 8C_2\tau(\log n)^2\sqrt{n}\kappa} \Big) + d\exp\Big(- \frac{n\kappa^2}{8\omega m\Vert X_.\Vert_2^2 + 8/3m\tau \sqrt{n}\kappa} \Big)\\
    &+ \kappa\sqrt{\log d} + \bigg[\frac{M^{(4\theta-1)/3}(\log d)^{(4+3\theta)}}{n^{\theta}}\bigg]^{1/(4+2\theta)} + \frac{1}{n} + \frac{(\log d)^{2/3}}{M^{1/3}}\\
    \lesssim& \frac{1}{n} + \kappa\sqrt{\log d} + \bigg[\frac{M^{(4\theta-1)/3}(\log d)^{(4+3\theta)}}{n^{\theta}}\bigg]^{1/(4+2\theta)} + \frac{(\log d)^{2/3}}{M^{1/3}}.
\end{align*}

Combining everything together, and optimize the error rate by letting $\kappa = (\log d)^{1/6}M^{-1/3} \asymp n^{\alpha/(3\beta) - 1/6}$ and requiring (a) $\alpha > 2\beta$ and (b) $16\alpha + 6\theta > (59+36\theta)\beta$, we have
\begin{equation*}
\begin{aligned}
    &\rho_{n,\tau} \leq \rho_{n,\tau}^* + \rho_{n,\tau}^{\diamond} + \rho_{n,\tau}^{\circ}\\
    \lesssim& \frac{1}{n} + \kappa\sqrt{\log d} + \bigg[\frac{M^{(4\theta-1)/3}(\log d)^{(4+3\theta)}}{n^{\theta}}\bigg]^{1/(4+2\theta)} + \frac{(\log d)^{2/3}}{M^{1/3}} + \bigg[\frac{M^{(3\theta+3)}(\log d)^{(4+3\theta)}}{n^{\theta}}\bigg]^{1/(4+2\theta)}\\
    &+ \bigg[\frac{M^{3\theta/2}(\log d)^{(4+3\theta)}}{n^{\theta}}\bigg]^{1/(6+3\theta)}\\
    \lesssim& \frac{(\log d)^{2/3}}{M^{1/3}} + \bigg[\frac{M^{(3\theta+3)}(\log d)^{(4+3\theta)}}{n^{\theta}}\bigg]^{1/(4+2\theta)} \lesssim n^{(2\beta - \alpha)/3} \vee n^{ [(3+3\theta)\alpha + (4+3\theta)\beta - \theta]/(4+2\theta)}.
\end{aligned}
\end{equation*}
Note that the condition (a) $\alpha > 2\beta$ is equivalent to $(\log d)^{2} = o(M)$, and the condition (b) $16\alpha + 6\theta > (59+36\theta)\beta$ implies $M^{3\theta+3}(\log d)^{4+3\theta} = o(n^{\theta})$. Therefore, under conditions (a) $\alpha > 2\beta$ and (b) $16\alpha + 6\theta > (59+36\theta)\beta$, and as $n, d \to \infty$, we have
\begin{equation}\label{eq:gaussian_approx}
\begin{aligned}
    \rho_{n,\tau} \leq \rho_{n,\tau}^* + \rho_{n,\tau}^{\diamond} + \rho_{n,\tau}^{\circ} \lesssim n^{(2\beta - \alpha)/3} \vee n^{ [(3+3\theta)\alpha + (4+3\theta)\beta - \theta]/(4+2\theta)} \to 0.
\end{aligned}
\end{equation}
\end{proof}

\section{Proofs of Bernstein's inequalities}\label{sec:proof_Bernstein}
\subsection{Proof of Theorem~2.5}\label{sec:proof_bern_1}
Before providing the proof of Theorem 2.5, we introduce some necessary notations and useful tools.
Let $\{Z_i\}_{i \in \mathbb{N}}$ be a sequence of random variables of the form \textup{(1)}. Assume for any $i \in \mathbb{N}$ that $\mathbb{E}[Z_t] = 0$ and there exists a positive $M$ such that $|Z_i| \leq M$. In the following proofs, we will frequently need to divide an interval into several subintervals. To avoid nondivisibility and  notational complexity, we can embed $\{Z_i\}_{i \in \mathbb{N}}$ into a continuous time process $\{Z_t\}_{t > 0}$ by defining $Z_t = Z_{\lceil t \rceil}$. We can also embed the index of the functional dependence measure \textup{(3)} into continuous time by defining $\delta_{t,q} = \delta_{\lceil t \rceil, q}$ for $t \geq 0$. For a Borel set $\mathcal{A}$, define
\begin{equation*}
    S_{\mathcal{A}} = \int_{\mathcal{A}}Z_t dt.
\end{equation*}
Further, we denote the Lebesgue measure of $\mathcal{A}$ as $\lambda(\mathcal{A})$.

Let $A \geq 2$ be a real number. Our first goal is to upper bound the log-Laplace transform of partial sums, i.e.~$\log \mathbb{E}[\exp(tS_{(0,A]})]$, for any small $t > 0$. To this end, we introduce the construction of the Cantor-like set $K_A = \bigcup_{i=1}^{2^l} I_{l,i} \subset (0, A]$, where $\{I_{l,i}\}_{i = 1}^{2^l}$ are left half-open intervals with the same Lebesgue measure $n_l$, and all neighboring intervals are separated by some left half-open intervals. Heuristically, this construction reduces the dependence by creating gaps.

\textbf{Construction of $K_A$}.
The construction of $K_A$ of $(0, A]$ involves $l$ recursive steps. Let $\delta$ be some constant in $(0, 1)$, whose choices will be given later.
We define 
\begin{equation}\label{eq:choose_l}
    l = \max\Big\{k \in \mathbb{N}, A\Big(\frac{1-\delta}{2}\Big)^k \geq 1\Big\}.
\end{equation}
Note that by \eqref{eq:choose_l} we have that 
\begin{equation}\label{eq:l_UB}
    l \leq \frac{\log A}{\log \frac{2}{1-\delta}} < \frac{\log A}{\log 2}.
\end{equation}

\begin{enumerate}
    \item[Step C$1$.] Divide the interval $(0,A]$ into three left half-open intervals and delete the middle one $D_{1,0}$ with Lebesgue measure $A\delta$. The remaining ordered left half-open intervals are denoted as $I_{1,1}$ and $I_{1,2}$ with the same Lebesgue measure $A(1-\delta)/2$.
    \item[Step C$2$.] For $I_{1,1}$ (resp. $I_{1,2}$), divide it into three left half-open intervals and delete the middle one $D_{2,1}$ (resp. $D_{2,2}$) with Lebesgue measure $A\delta(1-\delta)/2$. The remaining ordered four left half-open intervals are denoted as $I_{2,1}$, $I_{2,2}$, $I_{2,3}$ and $I_{2,4}$ with the same Lebesgue measure $A((1-\delta)/2)^2$.
    \item[Step C$k$.] We repeat the procedure. At Step C$k$ with $1 \leq k \leq l$, we obtain left half-open intervals $I_{k, j}$ for $j = 1, \dots, 2^{k}$, each with Lebesgue measure $A((1-\delta)/2)^{k}$, and delete left halp-open intervals $D_{k,i}$ for $i = 1, \dots, 2^{k-1}$, each with Lebesgue measure $A\delta((1-\delta)/2)^{k-1}$.
\end{enumerate}
Finally, after Step C$l$, we obtain $\{I_{l,i}\}_{i = 1}^{2^l}$ and $K_A^{(l)} = \bigcup_{i = 1}^{2^l}I_{l,i}$. Moreover, for any $k = 0, 1, \dots, l$ and $j = 1, \dots, 2^k$, we also define
\begin{equation}\label{eq:cantor_decomp_1}
    K_{A,k,j} = \bigcup_{i = (j-1)2^{l-k}+1}^{j2^{l-k}}I_{l,i},
\end{equation}
and we have
\begin{equation}\label{eq:cantor_decomp_2}
    K_A = \bigcup_{j=1}^{2^k}K_{A,k,j}^{(l)}.
\end{equation}
Since after Step C$l$, the total length of all deleted intervals satisfies that
\begin{equation*}
    \sum_{i = 0}^{l-1}A\delta(1-\delta)^i \leq lA\delta.
\end{equation*}
Therefore, we have that the Lebesgue measure of $K_A^{(l)}$
\begin{equation}\label{eq:lebesgue_KA}
    A > \lambda(K_A) = A-\sum_{i = 0}^{l-1}A\delta(1-\delta)^i > A - \delta\frac{A\log A }{\log 2}.
\end{equation}

We summarize our auxiliary results as the following proposition.
\begin{prop}\label{prop:log-laplace}
    Let $\{Z_i\}_{i \in \mathbb{Z}}$ be a sequence of random variables of the form \textup{(1)} with mean zero. Assume there exists a positive $M$ such that $\sup_{i \in \mathbb{Z}}|Z_i| \leq M$, and there exist absolute constants $c, C_{\mathrm{FDM}} > 0$ such that
    \begin{equation*}
        \sup_{m \geq 0}\exp(cm)\Delta_{m,2} \leq C_{\mathrm{FDM}}.
    \end{equation*}
    Define the following absolute constant depending only on $c$ and $C_{\mathrm{FDM}}$ as 
    \begin{equation*}
        c_0 = \frac{c}{8} \wedge \sqrt{\frac{c\log 2}{8}}, \;\; c_1 = \frac{e^cC_{\mathrm{FDM}}}{(1 - e^{-c})c} \;\; \text{and} \;\; c_2 = \frac{2}{c\log 2}.
    \end{equation*}
    We have the following results.
    \begin{enumerate}
        \item[(i)] Let $A \geq 2(c \vee 8)$, the Cantor-like set $K_A$ has Lebesgue measure strictly larger than $A/2$ and satisfies for any $t > 0$ such that $tM \leq c_0/(log A) \wedge 1/2$ that
    \begin{equation*}
    \begin{aligned}
        \log\mathbb{E}[\exp(tS_{K_A})]
        \leq 6.2At^2C_{\mathrm{LRV}} + 6c_1A^{-1}t^2M.
    \end{aligned}
    \end{equation*}
        \item[(ii)] Let $A \geq 2(c \vee 2)$, it satisfies for any $t > 0$ such that $tM \leq (c \wedge 1)/2$ that
    \begin{equation*}
    \begin{aligned}
        \log\mathbb{E}[\exp(tS_{(0,A]})]
        \leq 6.2At^2C_{\mathrm{LRV}} + 3c_1At^2M + c_2t^2M^2A\log A.
    \end{aligned}
    \end{equation*}
    \end{enumerate}
\end{prop}

Proposition \ref{prop:log-laplace} directly leads to the following corollary.
\begin{corollary}\label{coro:log-laplace}
    Let $\{Z_i\}_{i \in \mathbb{Z}}$ be a sequence of random variables of the form \textup{(1)} with mean zero. Assume there exists a positive $M$ such that $\sup_{i \in \mathbb{Z}}|Z_i| \leq M$, and there exist absolute constants $c, C_{\mathrm{FDM}} > 0$ such that
    \begin{equation*}
        \sup_{m \geq 0}\exp(cm)\Delta_{m,2} \leq C_{\mathrm{FDM}}.
    \end{equation*}
    Let $C = \max\{c_0, 6c_1, c_2\}$, where $c_0$, $c_1$ and $c_2$ are given in Proposition \ref{prop:log-laplace}.
    We have the following results.
    \begin{enumerate}
        \item[(i)] Let $A \geq 2(c \vee 8)$, the Cantor-like set $K_A$ has Lebesgue measure strictly larger than $A/2$ and satisfies for any $t > 0$ such that $tM \leq c_0/(log A) \wedge 1/2$ that
    \begin{equation*}
    \begin{aligned}
        \log\mathbb{E}[\exp(tS_{K_A})]
        \leq \frac{CAt^2\big(\sqrt{C_{\mathrm{LRV}}} + \sqrt{M}/A\big)^2}{1 - 2tM/(2c_0/(\log A) \wedge 1)}.
    \end{aligned}
    \end{equation*}
        \item[(ii)] Let $A \geq 2(c \vee 2)$, it satisfies for any $t > 0$ such that $tM \leq (c \wedge 1)/2$ that
    \begin{equation*}
    \begin{aligned}
        \log\mathbb{E}[\exp(tS_{(0,A]})]
        \leq \frac{CAt^2\log A \big(\sqrt{C_{\mathrm{LRV}}} + M\big)^2}{1 - 2tM/(c \wedge 1)}.
    \end{aligned}
    \end{equation*}
    \end{enumerate}
\end{corollary}

Now, we are ready to prove Theorem~2.5

\begin{proof}[Proof of Theorem 2.5]
    \textbf{Case 1.} Suppose $n \leq 16(c \vee 8)^2$. For any $t > 0$ such that $tM \leq 4^{-1}(c \vee 8)^{-2}$, we have that $|tS_n| \leq tMn \leq 4$. Using the same argument as in the proof of Lemma \ref{lemma:log-Lapalce_low_level}, we have that
    \begin{equation*}
        \log\mathbb{E}[\exp(tS_{K_B})] \leq 3.1nt^2C_{\mathrm{LRV}} \leq \frac{3.1nt^2C_{\mathrm{LRV}}}{1 - 4tM(c \vee 8)^2}.
    \end{equation*}
    
    \textbf{Case 2.} Suppose $n > 16(c \vee 8)^2$. The proof follows from constructing recursively the Cantor-like set until the Lebesgue measure of the remaining interval is small enough. More specifically, choose $\delta$ as \eqref{eq:choose_delta}. Define a nondecreasing and continuous function from $(0,A]$ onto $(0, A- \lambda(K_A)]$ as
    \begin{equation*}
        F_A(t) = \lambda((0,t] \cap K_A^c) \;\; \text{for any} \;\; t \in (0,A],
    \end{equation*}
    where $K_A^c = (0,A] \setminus K_A$. Let $F_A^{-1}$ be the inverse function of $F_A$.
    We start with $\{X_t^{(0)}\}_t = \{X_t\}_t$ for $t$ in the interval $(0,A_0]$ where $A_0 = n$. After constructing the Cantor-like set $K_{A_0}$ of $(0, A_0]$, we connect all the gap intervals following their original order and define the connected interval as $K_{A_0}^c = (0,A_0] \setminus K_{A_0}$. Define also $A_1 = \lambda(K_{A_0}^c) = A_0 - \lambda(K_{A_0})$, and replace the original time index of interval $K_{A_0}^c$ by $(0,A_1]$. Define
    \begin{equation*}
        X_t^{(1)} = X_{F_{A_0}^{-1}(t)} \;\; \text{for any} \;\; t \in (0, A_1].
    \end{equation*}
    Then, we construct the Cantor-like set $K_{A_1}$ of $K_{A_0}^c$, connect all the gap intervals as $K_{A_1}^c = K_{A_0}^c \setminus K_{A_1}$, define $A_2 = \lambda(K_{A_1}^c) = A_1 - \lambda(K_{A_1})$, replace the original time index of interval $K_{A_1}^c$ by $(0,A_2]$, and define
    \begin{equation*}
        X_t^{(2)} = X_{F_{A_1}^{-1}(t)} \;\; \text{for any} \;\; t \in (0, A_2].
    \end{equation*}
    In general, for $j \geq 2$, we construct the Cantor-like set $K_{A_{j-1}}$ of $K_{A_{j-2}}^c$, connect all the gap intervals as $K_{A_{j-1}}^c = K_{A_{j-2}}^c \setminus K_{A_{j-1}}$, define $A_j = \lambda(K_{A_{j-1}}^c) = A_{j-1} - \lambda(K_{A_{j-1}})$, replace the original time index of interval $K_{A_{j-1}}^c$ by $(0,A_j]$, and define
    \begin{equation*}
        X_t^{(j)} = X_{F_{A_{j-1}}^{-1}(t)} \;\; \text{for any} \;\; t \in (0, A_j].
    \end{equation*}
    This procedure continues until $A_j$ is small enough. Let
    \begin{equation*}
        L = L_n = \inf\big\{j \in \mathbb{N}, A_j \leq 2(c \vee 8)\big\}.
    \end{equation*}
    Due to the choice of $\delta$ and \eqref{eq:lebesgue_KA}, we have for any $j \geq 1$ that $A_j < n/2^j$ and
    \begin{equation*}
        L \leq \Big\lfloor \frac{\log n - \log(2(c \vee 8))}{\log 2} \Big\rfloor + 1.
    \end{equation*}
    Moreover, we have for any $j = 0, \dots, L-1$ that
    \begin{equation*}
        A_j \geq A_{L-1} \geq 2(c \vee 8).
    \end{equation*}
    Therefore, we have the following decomposition
    \begin{equation*}
        \int_{0}^n X_u du = \sum_{j = 0}^{L-1}\int_{K_{A_j}}X_u^{(j)} du + \int_0^{A_L}X_u^{(L)} du.
    \end{equation*}
    Denote
    \begin{equation*}
        Y_j = \int_{K_{A_j}}X_u^{(j)} du \;\; \text{for} \;\; j = 0, \dots, L-1 \;\; \text{and} \;\; Y_L = \int_0^{A_L}X_u^{(L)} du.
    \end{equation*}
    For any $j = 0, \dots, L-1$, applying Corollary \ref{coro:log-laplace} (i), we have for any $t > 0$ such that $tM \leq c_0/(\log (n/2^j)) \wedge 1/2$ that
    \begin{equation*}
    \begin{aligned}
        \log\mathbb{E}[\exp(tY_j)]
        \leq \frac{C(n/2^j)t^2\big(\sqrt{C_{\mathrm{LRV}}} + \sqrt{M}(n/2^j)^{-1}\big)^2}{1 - 2tM/(2c_0/(\log (n/2^j)) \wedge 1)}.
    \end{aligned}
    \end{equation*}
    For $Y_L$, we first assume $A_L \geq 2(c \vee 2)$, we can apply Corollary \ref{coro:log-laplace} (ii). Then, we have for any $t > 0$ such that $tM \leq (c \wedge 1)/2$ that
    \begin{equation*}
    \begin{aligned}
        \log\mathbb{E}[\exp(tY_L)]
        \leq \frac{C^{\prime}t^2 \big(\sqrt{C_{\mathrm{LRV}}} + M\big)^2}{1 - 2tM/(c \wedge 1)},
    \end{aligned}
    \end{equation*}
    where $C^{\prime} > 0$ is some absolute constant depending only on $c$ and $C_{\mathrm{FDM}}$.
    
    We apply Lemma \ref{lemma:aggregate_log-laplace} by letting for $j = 0, \dots, L-1$
    \begin{equation*}
        c_j = 2M/(2c_0/(\log (n/2^j)) \wedge 1) \;\; \text{and} \;\; \sigma_j = \sqrt{C(n/2^j)}\big(\sqrt{C_{\mathrm{LRV}}} + \sqrt{M}(n/2^j)^{-1}\big)
    \end{equation*}
    and
    \begin{equation*}
        c_L = 2M/(c \wedge 1) \;\; \text{and} \;\; \sigma_L = \sqrt{C^{\prime}}\big(\sqrt{C_{\mathrm{LRV}}} + M\big).
    \end{equation*}
    There exist absolute constants $C_1, C^{\prime\prime} > 0$ depending only on $c, C_{\mathrm{FDM}}$ such that
    \begin{equation*}
        \sum_{j = 1}^L c_j \leq ML\log n/c_0 \vee 2M(L+1/c) \leq C_1M(\log n)^2,
    \end{equation*}
    and
    \begin{equation*}
        \sum_{j = 1}^L \sigma_j \leq C^{\prime\prime}\big(\sqrt{nC_{\mathrm{LRV}}} + M\big).
    \end{equation*}
    Then, for any $t \in [0,1/(C_1M(\log n)^2))$, we have that
    \begin{equation*}
        \log\mathbb{E}[\exp(tS_n)] \leq \frac{C_2(C_{\mathrm{LRV}}n + M^2) t^2}{1-tMC_1(\log n)^2},
    \end{equation*}
    where $C_2 = 2(C^{\prime\prime})^2$.
    Finally, by Chebyshev's inequality, we have for any $x > 0$ that
    \begin{equation*}
    \begin{aligned}
        \mathbb{P}(S_n \geq x) \leq& \exp\Big(\frac{C_2(C_{\mathrm{LRV}}n + M^2)t^2}{1-tMC_1(\log n)^2} - tx\Big)\\
        \leq& \exp\Big(-\frac{x^2}{4C_2(C_{\mathrm{LRV}}n + M^2) + 2C_1M(\log n)^2x}\Big),
    \end{aligned}
    \end{equation*}
    where the second inquality follows by letting 
    \begin{equation*}
        t = \frac{x}{2C_2(C_{\mathrm{LRV}}n + M^2) + MC_1(\log n)^2x}.
    \end{equation*}
    Since $Z_i$ are centered, applying the same argument for $-S_n$ give the same result. Therefore, we have for any $x > 0$ that
    \begin{equation*}
    \begin{aligned}
        \mathbb{P}(|S_n| \geq x) \leq 2\exp\Big(-\frac{x^2}{4C_2(C_{\mathrm{LRV}}n + M^2) + 2C_1M(\log n)^2x}\Big).
    \end{aligned}
    \end{equation*}
    
    In addition, suppose $A_L \leq 2(c \vee 2)$. Let $tM \leq 2/(c \vee 2)$, we have $|tY_L| \leq 4$. Using the similar argument as in \textbf{Case 1.}, we have that
    \begin{equation*}
        \log\mathbb{E}[\exp(tY_L)] \leq \frac{6.2(c \vee 2)t^2C_{\mathrm{LRV}}}{1 - 2tM(c \vee 2)}.
    \end{equation*}
    Let
    \begin{equation*}
        c_L = 2M(c \wedge 2) \;\; \text{and} \;\; \sigma_L = \sqrt{6.2(c \vee 2)C_{\mathrm{LRV}}}.
    \end{equation*}
    Note that using the updated $c_L$ and $\sigma_L$ would only affect the absolute constant $C_1$ and $C_2$. The same result follows.
\end{proof}

\subsection{Proofs for \Cref{sec:proof_bern_1}}
The first lemma relates the Laplace transform of the partial sum to the product of the Laplace transforms of each individual random variable.
\begin{lemma}\label{lemma:diff_laplace_joint_marginal}
Let $\{Z_i\}_{i \in \mathbb{Z}}$ be an $\mathbb{R}$-valued potentially nonstationary process of form \textup{(1)}. Assume there exists a positive $M$ such that $|Z_i| \leq M$ for any $i \in \mathbb{Z}$. Then for any $a > 0$, we have
\begin{equation*}
\begin{aligned}
    \bigg| \mathbb{E}\Big[\exp\Big(a\sum_{i = 1}^n Z_i \Big) \Big] - \prod_{i = 1}^n \mathbb{E}[\exp(aZ_i)] \bigg| \leq& a\exp(anM)\sum_{i = 2}^n\Vert Z_i - Z_{i, \{i-1, -\infty\}} \Vert_2\\
    \leq& a\exp(anM)(n-1) \Delta_{1,2}.
\end{aligned}
\end{equation*}
\end{lemma}
\begin{proof}[Proof of \Cref{lemma:diff_laplace_joint_marginal}]
For the product of the Laplace transforms of each individual random variable, we have the telescoping decomposition as
\begin{equation*}
\begin{aligned}
    &\prod_{i = 1}^n \mathbb{E}[\exp(aZ_i)]\\
    =& \mathbb{E}\Big[\exp\Big(a\sum_{i = 1}^{n-1} Z_i \Big) \Big]\mathbb{E}[\exp(aZ_n)]\\
    &+\mathbb{E}\Big[\exp\Big(a\sum_{i = 1}^{n-2} Z_i \Big) \Big]\prod_{j = n-1}^n\mathbb{E}[\exp(aZ_j)] - \mathbb{E}\Big[\exp\Big(a\sum_{i = 1}^{n-1} Z_i \Big) \Big]\mathbb{E}[\exp(aZ_n)]\\
    &+\mathbb{E}\Big[\exp\Big(a\sum_{i = 1}^{n-3} Z_i \Big) \Big]\prod_{j = n-2}^n\mathbb{E}[\exp(aZ_j)] - \mathbb{E}\Big[\exp\Big(a\sum_{i = 1}^{n-2} Z_i \Big) \Big]\prod_{j = n-1}^n\mathbb{E}[\exp(aZ_n)]\\
    &+ \dots\\
    &+\prod_{j = 1}^n\mathbb{E}[\exp(aZ_j)] - \mathbb{E}\Big[\exp\Big(a\sum_{i = 1}^{2} Z_i \Big) \Big]\prod_{j = 3}^n\mathbb{E}[\exp(aZ_n)].
\end{aligned}
\end{equation*}
For notational simplicity, given a real value sequence $\{b_i\}_{i = 1}^n$, we write $\prod_{i = n+1}^n b_i = 1$. 
Then, it satisfies that
\begin{equation}\label{eq:laplace_extract_all}
\begin{aligned}
    &\mathbb{E}\Big[\exp\Big(a\sum_{i = 1}^n Z_i \Big) \Big] - \prod_{i = 1}^n \mathbb{E}[\exp(aZ_i)]\\
    =& \sum_{s = 2}^{n}\bigg\{\bigg(\mathbb{E}\Big[\exp\Big(a\sum_{i = 1}^{s} Z_i \Big) \Big] - \mathbb{E}\Big[\exp\Big(a\sum_{i = 1}^{s-1} Z_i \Big) \Big]\mathbb{E}[\exp(aZ_s)] \bigg)\\
    &\times \prod_{j = s+1}^n\mathbb{E}[\exp(aZ_j)] \bigg\}.
\end{aligned}
\end{equation}
Using coupling, we have that $Z_{s, \{s-1, -\infty\}}$ and $Z_{s^{\prime}}$ are independent for any $s^{\prime} \leq s-1$, and $Z_{s, \{s-1, -\infty\}}$ and $Z_s$ have the same distribution. We have that 
\begin{equation}\label{eq:laplace_extract_one}
\begin{aligned}
    &\bigg|\mathbb{E}\Big[\exp\Big(a\sum_{i = 1}^{s} Z_i \Big) \Big] - \mathbb{E}\Big[\exp\Big(a\sum_{i = 1}^{s-1} Z_i \Big) \Big]\mathbb{E}[\exp(aZ_s)]\bigg|\prod_{j = s+1}^n\mathbb{E}[\exp(aZ_j)]\\
    =& \bigg|\mathbb{E}\Big[\exp\Big(a\sum_{i = 1}^{s-1} Z_i \Big) \Big(\exp(aZ_s) - \exp(aZ_{s,\{s-1, -\infty\}})\Big)\Big]\bigg| \prod_{j = s+1}^n\mathbb{E}[\exp(aZ_j)]\\
    \leq& a\exp(anM)\mathbb{E}\big|Z_s - Z_{s,\{s-1, -\infty\}}\big|\\
    \leq& a\exp(anM)\big\Vert Z_s - Z_{s,\{s-1, -\infty\}}\big\Vert_2,
\end{aligned}
\end{equation}
where the first inequality is due to the mean value theorem and the fact that $Z_i$ are bounded, and the second inequality follows from H\"older's inequality.
Combining \eqref{eq:laplace_extract_all} and \eqref{eq:laplace_extract_one}, we have
\begin{equation*}
\begin{aligned}
    \bigg|\mathbb{E}\Big[\exp\Big(a\sum_{i = 1}^n Z_i \Big) \Big] - \prod_{i = 1}^n \mathbb{E}[\exp(aZ_i)]\bigg| \leq& a\exp(anM)\sum_{i = 2}^n\big\Vert Z_i - Z_{i,\{i-1, -\infty\}}\big\Vert_2\\
    \leq& a(n-1)\exp(anM)\sum_{j = 1}^{\infty}\delta_{j,2},
\end{aligned}
\end{equation*}
where the second inequality follows from the definition of $\delta_{s,q}$.
\end{proof}

The following lemma relates the bound of the log-Laplace transform of sum of random variables to that of each individual random variable. This lemma is the Lemma 13 in \cite{merlevede2009bernstein}, we reproduce it for completeness.
\begin{lemma}\label{lemma:aggregate_log-laplace}
    Let $Z_1, Z_2, \dots$ be a sequence of $\mathbb{R}$-valued random variables. Assume that there exist positive constants $\sigma_1, \sigma_2, \dots$ and $c_1, c_2, \dots$ such that, for any positive integer $i$ and any $t \in [0, 1/c_i)$,
    \begin{equation*}
        \log\mathbb{E}[\exp(tZ_i)] \leq (\sigma_it)^2/(1-c_it).
    \end{equation*}
    Then, for any positive $n$ and any $t$ in $[0,1/(c_1+c_2+ \cdots+c_n))$,
    \begin{equation}\label{eq:tensorization}
        \log\mathbb{E}[\exp(t(Z_1+Z_2+\cdots+Z_n))] \leq (\sigma t)^2/(1-Ct),
    \end{equation}
    where $\sigma = \sigma_1+\sigma_2+\cdots+\sigma_n$ and $C = c_1+c_2+\cdots+c_n$.
\end{lemma}
\begin{proof}[Proof of \Cref{lemma:aggregate_log-laplace}]
    For $i \geq 1$, denote the partial sums $S_i = \sum_{j = 1}^iZ_j$. The proof is by induction.
    For $n = 1$, we have $\sigma = \sigma_1$ and $C = c_1$, and \eqref{eq:tensorization} holds obviously.
    
    Assuming \eqref{eq:tensorization} holds for $n = k$, i.e.~for any $t \in [0,1/(c_1+c_2+\cdots+c_k))$, it satisfies that
    \begin{equation*}
        \log\mathbb{E}[\exp(tS_k)] \leq \Big(t\sum_{j=1}^k\sigma_j\Big)^2\Big/\Big(1-t\sum_{j=1}^kc_j\Big).
    \end{equation*}
    For $n = k+1$, by H\"older's inequality we have for any $u \in (0,1)$ that
    \begin{equation}\label{eq:tensorization_stepk}
    \begin{aligned}
        &\log\mathbb{E}[\exp(t(S_k+Z_{k+1}))]\\
        \leq& u\log\mathbb{E}[\exp(u^{-1}tS_k)] + (1-u)\log\mathbb{E}[\exp((1-u)^{-1}tZ_{k+1})].
    \end{aligned}
    \end{equation}
    Choose
    \begin{equation*}
        u = \Big(\sum_{j = 1}^{k}\sigma_j\Big/\sum_{j = 1}^{k+1}\sigma_j\Big)\Big(1 - t\sum_{j = 1}^{k+1}c_j\Big) + t\sum_{j = 1}^{k}c_j,
    \end{equation*}
    and thus 
    \begin{equation*}
        1-u = \Big(\sigma_{k+1}\Big/\sum_{j = 1}^{k+1}\sigma_j\Big)\Big(1 - t\sum_{j = 1}^{k+1}c_j\Big) + tc_{k+1}.        
    \end{equation*}
    Since $1/(c_1+c_2+\cdots+c_{k+1})$ is less than $1/(c_1+c_2+\cdots+c_{k})$ and $1/c_{k+1}$,
    we have for any $t \in [0, 1/(c_1+c_2+\cdots+c_{k+1}))$ that
    \begin{equation*}
        \eqref{eq:tensorization_stepk} \leq \frac{t^2\big(\sum_{j=1}^{k}\sigma_j\big)^2}{u - t\sum_{j=1}^{k}c_j} + \frac{t^2\sigma_{k+1}^2}{(1-u)-tc_{k+1}} = \frac{\big(t\sum_{j=1}^{k+1}\sigma_j\big)^2}{1 - t\sum_{j = 1}^{k+1}c_j},
    \end{equation*}
    which completes the proof.
\end{proof}

The following lemma provides an upper bound on log-Laplace transform of partial sums within a finite union of intervals with small Lebesgue measure. We use this lemma as building blocks in the proofs of our main results.
\begin{lemma}\label{lemma:log-Lapalce_low_level}
    Let $\{Z_i\}_{i \in \mathbb{Z}}$ be a sequence of random variables of the form \textup{(1)} with mean zero. Assume there exists a positive $M$ such that $\sup_{i \in \mathbb{Z}}|Z_i| \leq M$, and there exist absolute constants $c, C_{\mathrm{FDM}} > 0$ such that
    \begin{equation*}
        \sup_{m \geq 0}\exp(cm)\Delta_{m,2} \leq C_{\mathrm{FDM}}.
    \end{equation*}
    Let $B \geq 2$ and $a \geq 0$. Then for any $K_{B} \subset (a, a+B]$ which is a finite union of intervals, and for any $t > 0$ such that $tM \leq 1/2 \wedge \sqrt{c/(2B)}$, we have that
    \begin{equation*}
        \log\mathbb{E}[\exp(tS_{K_B})] \leq Bt^2\Big[6.2C_{\mathrm{LRV}} + \frac{e^cC_{\mathrm{FDM}}M}{1 - e^{-c}}\exp\Big(-\frac{c}{2tM}\Big)\Big],
    \end{equation*}
    where $C_{\mathrm{LRV}}$ is defined in \textup{(10)}.
\end{lemma}
\begin{proof}[Proof of Lemma \ref{lemma:log-Lapalce_low_level}]
    Note that the function $x \mapsto g(x) = x^{-2}(e^x-x-1)$ is increasing on $\mathbb{R}$. For any centered random variable $U \in \mathbb{R}$ such that $|U| \leq M$, and any $t > 0$, we have that
    \begin{equation}\label{eq:laplace_small_quantity}
        \mathbb{E}[\exp(tU)] \leq 1+t^2g(tM)\mathbb{E}[U^2].
    \end{equation}
    
    \textbf{Case 1.} Suppose $tM \leq 4/B$, we have that $|tS_{K_B}| \leq tMB \leq 4$. By \eqref{eq:laplace_small_quantity} and $g(4) \leq 3.1$, it holds that
    \begin{equation*}
        \mathbb{E}[\exp(tS_{K_B})] \leq 1+3.1t^2\mathbb{E}[S_{K_B}^2] \leq 1+3.1Bt^2C_{\mathrm{LRV}} \leq \exp\big( 3.1Bt^2C_{\mathrm{LRV}} \big),
    \end{equation*}
    where the second inequality follows from $\mathbb{E}[S_{K_B}^2] \leq B\sum_{l = -\infty}^{\infty}|\cov(X_0, X_l)|$ and the definition of $C_{\mathrm{LRV}}$, and the third inequality follows from the inequality $1 + x \leq \exp(x)$.
    
    \textbf{Case 2.} Suppose $tM > 4/B$. Let $r = (tM)^{-1}$ and $k = \lfloor B/(2r) \rfloor$. We list the following facts.
    \begin{itemize}
        \item By the condition $tM > 4/B$ and the definition of $k$, we have 
        \begin{equation}\label{eq:k>b/(4r)}
            \frac{B}{2r} > 2, \;\; \text{and} \;\; k \geq 2 \vee \frac{tMB}{4}.
        \end{equation}
        \item By the assumption $tM \leq 2^{-1}$ and the definition of $k$, we have 
        \begin{equation}\label{eq:k<tmb/2}
            \frac{B}{2k} \geq r \geq 2 \;\; \text{and} \;\; k \leq \frac{B}{2r} = \frac{tMB}{2}.
        \end{equation}
        \item By the assumption $tM \leq \sqrt{c/(2B)}$, we have
        \begin{equation}\label{eq:tmb-cb/(2k)}
            tMB - \frac{cB}{2k} \leq tMB - cr = tMB - \frac{c}{tM} \leq -\frac{c}{2tM}.
        \end{equation}
    \end{itemize}

    We divide interval $(a, a+B]$ into $2k$ consecutive left half-open intervals $\{I_j\}_{j = 1}^{2k}$ of equal size $B/(2k)$. Note that the odd intervals, i.e.~$\{I_{2j-1}\}_{j=1}^k$ are separated by the even intervals, i.e.~$\{I_{2j}\}_{j=1}^k$, and vice versa.
    Denote the partial sum among all odd intervals and that of even intervals respectively as
    \begin{equation*}
        \overline{S}_{\mathrm{odd}} = \sum_{j = 1}^kS_{I_{2j-1}} \;\; \text{and} \;\; \overline{S}_{\mathrm{even}} = \sum_{j = 1}^kS_{I_{2j}}.
    \end{equation*}
    By H\"older's inequality, we have that
    \begin{equation*}
        \log \mathbb{E}[\exp(tS_{K_B})] \leq \frac{1}{2}\log \mathbb{E}[\exp(2t\overline{S}_{\mathrm{odd}})] + \frac{1}{2}\log \mathbb{E}[\exp(2t\overline{S}_{\mathrm{even}})].
    \end{equation*}
    
    Note that $|S_{I_{2j-1}}| \leq MB/(2k)$ for any $j = 1, \dots, k$. Denote $I_{2j-1} = (l_{2j-1}, u_{2j-1}]$. By definition, we have $\{S_{I_{2j-1}}\}_{j = 1}^k$ is adapted to filtration $\{\mathcal{F}_{\lceil u_{2j-1}\rceil}\}_{j = 1}^k$. For any $j = 2, \dots, k$, let $S_{I_{2j-1}, \{\lceil u_{2j-3} \rceil, -\infty\}} = \int_{I_{2j-1}}Z_{t, \{\lceil u_{2j-3} \rceil, -\infty\}}dt$ and it satisfies that
    \begin{equation}\label{eq:embed_fdm_bound}
    \begin{aligned}
        &\Vert S_{I_{2j-1}} - S_{I_{2j-1}, \{\lceil u_{2j-3} \rceil, -\infty\}} \Vert_2 \leq \int_{I_{2j-1}}\Vert Z_t - Z_{t, \{\lceil u_{2j-3} \rceil, -\infty\}} \Vert_2 dt\\
        \leq& \int_{I_{2j-1}}\Vert Z_t - Z_{t, \{\lceil u_{2j-3} \rceil\}} \Vert_2 dt\\
        &+ \int_{I_{2j-1}} \sum_{m = 0}^{\infty}\Vert Z_{t, \{\lceil u_{2j-3} \rceil, \lceil u_{2j-3} \rceil - m\}} - Z_{t, \{\lceil u_{2j-3} \rceil, \lceil u_{2j-3} \rceil - m-1\}} \Vert_2 dt\\
        \leq& \int_{I_{2j-1}}\sum_{m = 0}^{\infty}\delta_{t - \lceil u_{2j-3}\rceil + m, 2} dt \leq \sum_{i = \lceil l_{2j-1}\rceil}^{\lceil u_{2j-1}\rceil}\sum_{m = 0}^{\infty}\delta_{i - \lceil u_{2j-3}\rceil + m, 2}\\
        \leq& C_{\mathrm{FDM}}\sum_{i = \lceil l_{2j-1}\rceil}^{\lceil u_{2j-1}\rceil}\exp(-c(i - \lceil u_{2j-3}\rceil)) \leq \frac{C_{\mathrm{FDM}}\exp(-c(\lceil l_{2j-1}\rceil - \lceil u_{2j-3}\rceil))}{1 - e^{-c}}\\
        \leq& \frac{C_{\mathrm{FDM}}e^c}{1 - e^{-c}}\exp\Big(-\frac{cB}{2k}\Big),
    \end{aligned}
    \end{equation}
    where the first inequality follows from the triangle inequality, the second inequality follows from the telescoping sums and the triangle inequality, the third inequality follows from the definition of the functional dependence measure, the fifth inequality follows from the assumption on dependence measure, and the last inequality follows from the fact that $\lceil l_{2j-1}\rceil - \lceil u_{2j-3}\rceil \geq B/(2k) - 1$.
    Applying Lemma \ref{lemma:diff_laplace_joint_marginal} and using \eqref{eq:embed_fdm_bound}, we have for any $t > 0$ that
    \begin{equation*}
    \begin{aligned}
        &\Big| \mathbb{E}[\exp(2t\overline{S}_{\mathrm{odd}})] - \prod_{j = 1}^k \mathbb{E}[\exp(2tS_{I_{2j-1}})] \Big|\\
        \leq& 2t\exp(tMB)\sum_{j = 2}^k\Vert S_{I_{2j-1}} - S_{I_{2j-1}, \{\lceil u_{2j-3} \rceil, -\infty\}} \Vert_2\\
        \leq& \frac{2t(k-1)C_{\mathrm{FDM}}e^c}{1 - e^{-c}}\exp\Big(-\frac{c}{2tM}\Big) \leq \frac{C_{\mathrm{FDM}}e^cMBt^2}{1 - e^{-c}}\exp\Big(-\frac{c}{2tM}\Big),
    \end{aligned}
    \end{equation*}
    where the second and third inequalities follow respectively from \eqref{eq:tmb-cb/(2k)} and \eqref{eq:k<tmb/2}.
    Since $Z_i$ are centered, by Jensen's inequality the Laplace transforms of $\overline{S}_{\mathrm{odd}}$ and $S_{I_{2j-1}}$ are greater than $1$. Applying the inequality
    \begin{equation}\label{eq:logx-logy}
        |\log x - \log y| \leq |x - y| \;\; \text{for} \; x\geq 1 \; \text{and} \; y \geq 1,
    \end{equation}
    we have that
    \begin{equation*}
    \begin{aligned}
        \Big| \log\mathbb{E}[\exp(2t\overline{S}_{\mathrm{odd}})] - \log\prod_{j = 1}^k \mathbb{E}[\exp(2tS_{I_{2j-1}})] \Big|
        \leq \frac{C_{\mathrm{FDM}}e^cMBt^2}{1 - e^{-c}}\exp\Big(-\frac{c}{2tM}\Big).
    \end{aligned}
    \end{equation*}
    Since $|S_{I_{2j-1}}| \leq MB/(2k)$ for any $j = 1, \dots, k$. By \eqref{eq:k>b/(4r)}, we have that
    \begin{equation*}
        |2tS_{I_{2j-1}}| \leq \frac{tMB}{k} \leq 4.
    \end{equation*}
    By the same arguments as in \textbf{Case 1.}, we have that
    \begin{equation*}
        \log\prod_{j = 1}^k \mathbb{E}[\exp(2tS_{I_{2j-1}})] \leq 6.2Bt^2C_{\mathrm{LRV}}.
    \end{equation*}
    The triangle inequality leads to that
    \begin{equation*}
    \begin{aligned}
        \log\mathbb{E}[\exp(2t\overline{S}_{\mathrm{odd}})]
        \leq Bt^2\Big[6.2C_{\mathrm{LRV}} + \frac{e^cC_{\mathrm{FDM}}M}{1 - e^{-c}}\exp\Big(-\frac{c}{2tM}\Big)\Big].
    \end{aligned}
    \end{equation*}
    The same upper bound on the log-Laplace transform of $\overline{S}_{\mathrm{even}}$ can be obtained similarly. Therefore, we have that
    \begin{equation*}
        \log \mathbb{E}[\exp(tS_{K_B})] \leq Bt^2\Big[6.2C_{\mathrm{LRV}} + \frac{e^cC_{\mathrm{FDM}}M}{1 - e^{-c}}\exp\Big(-\frac{c}{2tM}\Big)\Big].
    \end{equation*}
    Combining the above two cases completes the proof.
\end{proof}

\begin{proof}[Proof of Proposition \ref{prop:log-laplace}]
    (i).
    The proof use the construction of $K_A$.
    
    \textbf{Case 1.} Suppose $tM \leq c_0/(\log A) \wedge \sqrt{c/(2A)}$. In this case, $tM$ is small enough and we can apply Lemma \ref{lemma:log-Lapalce_low_level} directly. Since $A \geq 2c$, we have $tM \leq \sqrt{c/(2A)} \leq 1/2$. Applying Lemma \ref{lemma:log-Lapalce_low_level}, we have that
    \begin{equation}\label{eq:log-Laplace_SKA_case1}
    \begin{aligned}
        \log\mathbb{E}[\exp(tS_{K_A})] \leq& \lambda(K_A)t^2\Big[6.2C_{\mathrm{LRV}} + \frac{e^cC_{\mathrm{FDM}}M}{1 - e^{-c}}\exp\Big(-\frac{c}{2tM}\Big)\Big]\\
        \leq& At^2\Big[6.2C_{\mathrm{LRV}} + \frac{e^cC_{\mathrm{FDM}}M}{1 - e^{-c}}\exp\Big(-\frac{c}{2tM}\Big)\Big].
    \end{aligned}
    \end{equation}
    Since $tM \leq c_0/(\log A) \leq c/(8\log A)$, then $-2\log A \geq -c/(4tM)$ and we have that
    \begin{equation}\label{eq:bound_exp(c/(2tm))}
        \exp\Big(-\frac{c}{2tM}\Big) \leq \exp\Big(-\frac{c}{4tM} - 2\log A\Big) = A^{-2}\exp\Big(-\frac{c}{4tM}\Big).
    \end{equation}
    Combining \eqref{eq:log-Laplace_SKA_case1} and \eqref{eq:bound_exp(c/(2tm))}, we have that
    \begin{equation}\label{eq:log-Laplace_SKA_case1_bound}
    \begin{aligned}
        \log\mathbb{E}[\exp(tS_{K_A})] 
        \leq& t^2\Big[6.2AC_{\mathrm{LRV}} + \frac{e^cC_{\mathrm{FDM}}M}{A(1 - e^{-c})}\exp\Big(-\frac{c}{4tM}\Big)\Big]\\
        \leq& 6.2At^2C_{\mathrm{LRV}} + \frac{2e^cC_{\mathrm{FDM}}}{(1 - e^{-c})c}A^{-1}t^2M.
    \end{aligned}
    \end{equation}
    
    \textbf{Case 2.} Suppose $\sqrt{c/(2A)} < tM \leq c_0/(\log A) \wedge 1/2$.  We choose
    \begin{equation}\label{eq:choose_delta}
        \delta = \frac{\log 2}{2\log A}.
    \end{equation}
    With the choice of $\delta$, by \eqref{eq:lebesgue_KA}, we have that $\lambda(K_A) > A/2$.
    
    Let
    \begin{equation}\label{choose_k*}
        k^* = \inf\Big\{k \in \mathbb{N}, A\Big( \frac{1-\delta}{2} \Big)^k \leq \frac{c}{2(tM)^2} \Big\}.
    \end{equation}
    Note that $k^* \geq 1$ since $2(tM)^2 > c/A$. Since $tM \leq \sqrt{c/8}$, by the definitions of $l$ and $k^*$, we have $k^* \leq l$. Recall \eqref{eq:choose_delta} and \eqref{eq:choose_l}. Moreover, the definition of $k^*$ implies that
    \begin{equation}\label{eq:k^*_implication}
        A\delta\Big( \frac{1-\delta}{2} \Big)^{k^*-1} > \frac{c\delta}{2(tM)^2} \geq \frac{4\delta(\log A)^2}{\log 2} = 2\log A > 5,
    \end{equation}
    where the second inequality follows from $tM \leq c_0/(\log A) \leq \sqrt{c\log 2/(8(\log A)^2)}$, and the third inequality follows from that $A \geq 16$.
    In addition, by the definitions of $k^*$ and $l$ and the fact $k^* \leq l$, we have for any $j = 1, \dots, 2^{k^*}$ that
    \begin{equation*}
        \lambda(K_{A,k^*,j}) = A\Big(\frac{1-\delta}{2}\Big)^l2^{l-k^*} \leq \frac{c(1-\delta)^{l-k^*}}{2(tM)^2} \leq \frac{c}{2(tM)^2}.
    \end{equation*}
    Therefore, we have that
    \begin{equation}\label{eq:verify_condition_tm}
        tM \leq \sqrt{\frac{c}{2\lambda(K_{A,k^*,j})}} \wedge 1/2,
    \end{equation}
    which satisfies the condition for Lemma \ref{lemma:log-Lapalce_low_level}.
    
    We first use Lemma \ref{lemma:diff_laplace_joint_marginal} recursively at each level $k = 1, \dots k^*$, and then apply Lemma \ref{lemma:log-Lapalce_low_level} on $S_{K_{A,k^*,j}}$, whose size is small enough. The validity of the recursive procedure is justified by \eqref{eq:k^*_implication}.
    Denote by $l_{A,k,j}$ $u_{A,k,j}$ respectively the left and right boundaries of $K_{A,k,j}$. At the level $k = 1$, by \eqref{eq:cantor_decomp_2}, we have that $K_A = \bigcup_{j = 1}^{2}K_{A, 1, j}$, and $K_{A,1,1}$ and $K_{A,1,2}$ are separated by an interval of size $A\delta \geq 2$ by \eqref{eq:choose_delta} and $A \geq 16$. Thus, we have that $\lceil l_{A,1,2} \rceil - \lceil u_{A,1,1} \rceil \geq A\delta - 1$. Applying Lemma \ref{lemma:diff_laplace_joint_marginal}, we have that
    \begin{equation*}
    \begin{aligned}
        &\Big| \mathbb{E}[\exp(tS_{K_A})] - \prod_{j = 1}^2 \mathbb{E}[\exp(tS_{K_{A,1,j}})] \Big|\\
        \leq& t\exp(A(1-\delta)tM)\Vert S_{K_{A,1,2}} - S_{K_{A,1,2}, \{\lceil u_{A,1,1} \rceil, -\infty\}} \Vert_2\\
        \leq& \frac{tC_{\mathrm{FDM}}e^c}{1 - e^{-c}}\exp\big(-cA\delta+ A(1-\delta)tM\big),
    \end{aligned}
    \end{equation*}
    where the second inequality follows from the similar arguments used in \eqref{eq:embed_fdm_bound}.
    Since $Z_i$ are centered, by the inequality \eqref{eq:logx-logy} and the triangle inequality, we have for any $t > 0$ that
    \begin{equation*}
    \begin{aligned}
        &\log\mathbb{E}[\exp(tS_{K_A})] \\
        \leq& \sum_{j = 1}^2 \log\mathbb{E}[\exp(tS_{K_{A,1,j}})] + \frac{tC_{\mathrm{FDM}}e^c}{1 - e^{-c}}\exp\big(-cA\delta+ A(1-\delta)tM\big).
    \end{aligned}
    \end{equation*}
    Using the above arguments recursively at level $k = 2, \dots, k^*$, we obtain for any $t > 0$ that
    \begin{equation*}
    \begin{aligned}
        &\log\mathbb{E}[\exp(tS_{K_A})]\\
        \leq& \sum_{j = 1}^{2^{k^*}} \log\mathbb{E}[\exp(tS_{K_{A,k^*,j}})]\\
        &+ \frac{tC_{\mathrm{FDM}}e^c}{1 - e^{-c}}\sum_{k = 0}^{k^*-1}2^k\exp\Big(-cA\delta\Big(\frac{1-\delta}{2}\Big)^k+ 2AtM\Big(\frac{1-\delta}{2}\Big)^{k+1}\Big).
    \end{aligned}
    \end{equation*}
    For any $tM \leq c\delta/2$, the above inequality can be simplied as
    \begin{equation}\label{eq:(i)case2_loglap_KA}
    \begin{aligned}
        &\log\mathbb{E}[\exp(tS_{K_A})]\\
        \leq& \sum_{j = 1}^{2^{k^*}} \log\mathbb{E}[\exp(tS_{K_{A,k^*,j}})] + \frac{tC_{\mathrm{FDM}}e^c}{1 - e^{-c}}\sum_{k = 0}^{k^*-1}2^k\exp\Big(-\frac{cA\delta}{2}\Big(\frac{1-\delta}{2}\Big)^k\Big)\\
        \leq& \sum_{j = 1}^{2^{k^*}} \log\mathbb{E}[\exp(tS_{K_{A,k^*,j}})] + \frac{tC_{\mathrm{FDM}}e^c}{1 - e^{-c}}\exp\Big(-\frac{cA\delta}{2}\Big(\frac{1-\delta}{2}\Big)^{k^*-1}\Big)2^{k^*}\\
        \leq& \sum_{j = 1}^{2^{k^*}} \log\mathbb{E}[\exp(tS_{K_{A,k^*,j}})] + \frac{tC_{\mathrm{FDM}}e^c}{1 - e^{-c}}A\exp\Big(-\frac{c^2\delta}{(2tM)^2}\Big)\\
        \leq& \sum_{j = 1}^{2^{k^*}} \log\mathbb{E}[\exp(tS_{K_{A,k^*,j}})] + \frac{tC_{\mathrm{FDM}}e^c}{1 - e^{-c}}A\exp\Big(-\frac{c}{2tM}\Big),
    \end{aligned}
    \end{equation}
    where the first inequality follows from $tM \leq c\delta/2$, the third inequality follows from \eqref{eq:k^*_implication} and $2^{k^*} < A$, and the fourth inequality follows from $tM \leq c\delta/2$ again.
    
    Recall \eqref{eq:verify_condition_tm}. Applying Lemma \ref{lemma:log-Lapalce_low_level}, we have for any $j = 1, \dots, k^*$ and for any $tM \leq c_0/(\log A) \wedge 1/2$ that
    \begin{equation*}
        \log\mathbb{E}[\exp(tS_{K_{A,k^*,j}})] \leq A\frac{(1-\delta)^l}{2^{k^*}}t^2\Big[6.2C_{\mathrm{LRV}} + \frac{e^cC_{\mathrm{FDM}}M}{1 - e^{-c}}\exp\Big(-\frac{c}{2tM}\Big)\Big].
    \end{equation*}
    Note that $c_0/(\log A) \leq c/(8\log A) < c\log 2/(4 \log A) = c\delta/2$. Therefore, we have for any $\sqrt{c/(2A)} < tM \leq c_0/(\log A) \wedge 1/2$ that
    \begin{equation}\label{eq:log-Laplace_SKA_case2}
    \begin{aligned}
        &\log\mathbb{E}[\exp(tS_{K_A})]\\
        \leq& At^2\Big[6.2C_{\mathrm{LRV}} + \frac{e^cC_{\mathrm{FDM}}M}{1 - e^{-c}}\exp\Big(-\frac{c}{2tM}\Big)\Big] + \frac{tC_{\mathrm{FDM}}e^c}{1 - e^{-c}}A\exp\Big(-\frac{c}{2tM}\Big)\\
        \leq& 6.2At^2C_{\mathrm{LRV}} + (tM+1)\frac{e^cC_{\mathrm{FDM}}}{1 - e^{-c}}A^{-1}t\exp\Big(-\frac{c}{4tM}\Big)\\
        \leq& 6.2At^2C_{\mathrm{LRV}} + \frac{6e^cC_{\mathrm{FDM}}}{(1 - e^{-c})c}A^{-1}t^2M,
    \end{aligned}
    \end{equation}
    where the second inequality follows from \eqref{eq:bound_exp(c/(2tm))}. Combining \eqref{eq:log-Laplace_SKA_case1_bound} and \eqref{eq:log-Laplace_SKA_case2} finishes the proof.
    
    (ii).
    The proof use the same construction of $K_A$. However, instead of removing the random variables in the gap intervals as in (i), we treat them as bounded random variables.
    
    \textbf{Case 1.} Suppose $tM \leq \sqrt{c/(2A)}$. Since $tM \leq 1/2$, we apply Lemma \ref{lemma:log-Lapalce_low_level} and have that
    \begin{equation}\label{eq:log-Laplace_S(0,A]_case1}
    \begin{aligned}
        &\log\mathbb{E}[\exp(tS_{(0,A]})]
        \leq At^2\Big[6.2C_{\mathrm{LRV}} + \frac{e^cC_{\mathrm{FDM}}M}{1 - e^{-c}}\exp\Big(-\frac{c}{2tM}\Big)\Big]\\
        \leq& At^2\Big[6.2C_{\mathrm{LRV}} + \frac{2e^cC_{\mathrm{FDM}}}{(1 - e^{-c})c}tM^2\Big] \leq 6.2At^2C_{\mathrm{LRV}} + \frac{e^cC_{\mathrm{FDM}}}{(1 - e^{-c})c}At^2M,
    \end{aligned}
    \end{equation}
    where the second inequality follows from $x \leq e^x$, and the third inequality follows from $tM < 1/2$.
    
    \textbf{Case 2.} Suppose $\sqrt{c/(2A)} < tM \leq (c \wedge 1)/2$. We choose \begin{equation}\label{eq:choose_delta_2}
    \delta = \frac{2tM}{c}.
    \end{equation}
    Note that $\delta < 1$, since $tM < c/2$.
    We choose the same $l$ and $k^*$ as in \eqref{eq:choose_l} and \eqref{choose_k*}. Since $tM \leq (c \wedge 1)/2$, by the definitions of $l$ and $k^*$, we have $k^* \leq l$. The proof is similar as that of $S_{K_A}$, except that we consider the gap intervals. At the level $k = 1$, we have $\lambda(D_{1,0}) = A\delta \geq 2$, since $A \geq 2c$ and $tM > \sqrt{c/(2A)} \geq c/A$. Since $\exp(tS_{D_{1,0}}) \leq \exp(tMA\delta)$, it satisfies for any $t > 0$ that
    \begin{equation*}
    \begin{aligned}
        &\mathbb{E}[\exp(tS_{(0,A]})]
        \leq \mathbb{E}[\exp(tS_{I_{1,1}})\exp(tS_{I_{1,2}})]\exp(tMA\delta)\\
        \leq& \Big\{\prod_{j=1}^2\mathbb{E}[\exp(tS_{I_{1,j}})] + \frac{tC_{\mathrm{FDM}e^c}}{1-e^{-c}}\exp(-cA\delta + A(1-\delta)tM)\Big\}\exp(tMA\delta)\\
        =& \exp(tMA\delta)\prod_{j=1}^2\mathbb{E}[\exp(tS_{I_{1,j}})] + \frac{tC_{\mathrm{FDM}e^c}}{1-e^{-c}}\exp\Big(-\frac{cA\delta}{2}\Big),
    \end{aligned}
    \end{equation*}
    where the second inequality follows by applying Lemma \ref{lemma:diff_laplace_joint_marginal}, and the equality follows from \eqref{eq:choose_delta_2}. Since $Z_i$ are centered, applying the inequality \eqref{eq:logx-logy} leads to that for any $t > 0$
    \begin{equation*}
    \begin{aligned}
        &\log\mathbb{E}[\exp(tS_{(0,A]})]\\
        \leq& \sum_{j=1}^2\log\mathbb{E}[\exp(tS_{I_{1,j}})] + tMA\delta +  \frac{tC_{\mathrm{FDM}e^c}}{1-e^{-c}}\exp\Big(-\frac{cA\delta}{2}\Big).
    \end{aligned}
    \end{equation*}
    Using the above arguments recursively at level $k = 2, \dots, k^*$, we obtain for any $t > 0$ that
    \begin{equation*}
    \begin{aligned}
        &\log\mathbb{E}[\exp(tS_{(0,A]})]\\
        \leq& \sum_{j = 1}^{2^{k^*}} \log\mathbb{E}[\exp(tS_{I_{k^*,j}})]\\
        &+ \sum_{k = 0}^{k^*-1}\Big\{\frac{tC_{\mathrm{FDM}}e^c}{1 - e^{-c}}2^k\exp\Big(-\frac{cA\delta(1-\delta)^k}{2^{k+1}}\Big) + tMA\delta(1-\delta)^{k}\Big)\Big\}\\
        \leq& \sum_{j = 1}^{2^{k^*}} \log\mathbb{E}[\exp(tS_{I_{k^*,j}})] + \frac{tC_{\mathrm{FDM}}e^c}{1 - e^{-c}}A\exp\Big(-\frac{c}{2tM}\Big) + tMA\delta k^*\\
        \leq& \sum_{j = 1}^{2^{k^*}} \log\mathbb{E}[\exp(tS_{I_{k^*,j}})] + \frac{2C_{\mathrm{FDM}}e^c}{(1 - e^{-c})c}At^2M + \frac{2}{c\log 2}t^2M^2A\log A,
    \end{aligned}
    \end{equation*}
    where the second inequality follows the same arguments used in \eqref{eq:(i)case2_loglap_KA} and note that $\delta = 2tM/c$, and the third inequality follows from $x \leq e^x$, $k^* \leq l$ and \eqref{eq:l_UB}.
    The definition of $k^*$ \eqref{choose_k*} and the condition $tM \leq 1/2$ imply that
    \begin{equation*}
        A\delta\Big( \frac{1-\delta}{2} \Big)^{k^*-1} > \frac{c\delta}{2(tM)^2} = \frac{2tM}{2(tM)^2} \geq 2,
    \end{equation*}
    which justify the validity of the recursive procedure. Moreover, by the definitions of $k^*$, we have for any $j = 1, \dots, 2^{k^*}$ that
    \begin{equation*}
        \lambda(I_{k^*,j}) = A\Big( \frac{1-\delta}{2} \Big)^{k^*} \leq \frac{c}{2(tM)^2},
    \end{equation*}
    which shows that the conditions of Lemma \ref{lemma:log-Lapalce_low_level} is satisfied. Applying Lemma \ref{lemma:log-Lapalce_low_level}, we have for any $j = 1, \dots, k^*$ and for any $tM \leq (c \wedge 1)/2$ that
    \begin{equation*}
        \log\mathbb{E}[\exp(tS_{I_{k^*,j}})] \leq A\frac{(1-\delta)^{k^*}}{2^{k^*}}t^2\Big[6.2C_{\mathrm{LRV}} + \frac{e^cC_{\mathrm{FDM}}M}{1 - e^{-c}}\exp\Big(-\frac{c}{2tM}\Big)\Big].
    \end{equation*}
    Therefore, we have for any $\sqrt{c/(2A)} < tM \leq (c \wedge 1)/2$ that
    \begin{equation}\label{eq:log-Laplace_S(0,A]_case2}
    \begin{aligned}
        \log\mathbb{E}[\exp(tS_{(0,A]})]
        \leq 6.2At^2C_{\mathrm{LRV}} + \frac{3e^cC_{\mathrm{FDM}}}{(1 - e^{-c})c}At^2M + \frac{2}{c\log 2}t^2M^2A\log A.
    \end{aligned}
    \end{equation}
    Combining \eqref{eq:log-Laplace_S(0,A]_case1} and \eqref{eq:log-Laplace_S(0,A]_case2} finishes the proof.
\end{proof}

\subsection{Proof of Theorem~2.6}
\begin{proof}[Proof of Theorem~2.6]
The proof follows that of Theorem 2.4 in \cite{liu2021robust}.
Define the projection operator $\mathcal{P}_j\cdot = \mathbb{E}[\cdot|\mathcal{F}_j] - \mathbb{E}[\cdot|\mathcal{F}_{j-1}]$. Then, we have the following decomposition
\begin{equation*}
    \psi_M(X_i) = \sum_{j = -\infty}^i\mathcal{P}_j\psi_M(X_i), \;\; \text{and} \;\; \sum_{i = 1}^n\psi_M(X_i) = \sum_{i = 1}^n\sum_{j = -\infty}^i\mathcal{P}_j\psi_M(X_i) = \sum_{j = -\infty}^n\sum_{i = 1}^n\mathcal{P}_j\psi_M(X_i),
\end{equation*}
where the last equality follows that $\mathcal{P}_j\psi_M(X_i) = 0$ for all $j \geq i+1$. Denote $L_j = \sum_{i = 1}^n\mathcal{P}_j\psi_M(X_i)$.
Our goal is to bound the tail probability of the partial sum for $x > 0$ as
\begin{equation}\label{eq:tail_prob}
\begin{aligned}
    &\mathbb{P}\bigg( \sum_{i = 1}^n\big\{\psi_M(X_i) - \mathbb{E}[\psi_M(X_i)]\big\} \geq x \bigg) = \mathbb{P}\bigg( \sum_{j = -\infty}^nL_j \geq x \bigg) \leq e^{-\lambda x}\mathbb{E}\Big[\exp\Big\{\lambda \sum_{j = -\infty}^nL_j\Big\} \Big]\\
    =& e^{-\lambda x}\mathbb{E}\bigg[\mathbb{E}\Big[\exp\Big\{\lambda \sum_{j = -\infty}^nL_j\Big\}\Big|\mathcal{F}_{n-1}\Big] \bigg]\\
    =& e^{-\lambda x}\mathbb{E}\bigg[\mathbb{E}\big[\exp\{\lambda L_n\}|\mathcal{F}_{n-1}\big] \mathbb{E}\Big[\exp\Big\{\lambda \sum_{j = -\infty}^{n-1}L_j\Big\}\Big|\mathcal{F}_{n-2} \Big]\bigg]\\
    =& \cdots\\
    =& e^{-\lambda x}\mathbb{E}\Big[\prod_{j = -\infty}^{n}\mathbb{E}\big[\exp\{\lambda L_j\}|\mathcal{F}_{j-1}\big] \Big],
\end{aligned}
\end{equation}
where the first inequality follows from Markov's inequality with some $\lambda > 0$, the second equality follows from the tower property, the third equality follows from the conditional independence, and the last line follows from iteratively taking the conditional expectation on $\mathcal{F}_{n-3}, \dots, \mathcal{F}_{-\infty}$.
\\
For $j = -\infty, \dots, n$, we have by Taylor's expansion that
\begin{equation}\label{eq:exp_conditional}
\begin{aligned}
    &\mathbb{E}\big[\exp\{\lambda L_j \}|\mathcal{F}_{j-1}\big] = 1 + \mathbb{E}\big[ \lambda L_j| \mathcal{F}_{j-1} \big] + \sum_{k = 2}^{\infty}\frac{1}{k!}\mathbb{E}\big[ \lambda^k L_j^k| \mathcal{F}_{j-1} \big]\\
    =& 1 + \sum_{k = 2}^{\infty}\frac{1}{k!}\mathbb{E}\big[ \lambda^k L_j^k| \mathcal{F}_{j-1} \big],
\end{aligned}
\end{equation}
where we have $\mathbb{E}[ \lambda L_j| \mathcal{F}_{j-1}] = 0$ by definition. Since $|X_i| < M$ for any $i = 1, \dots, n$, we have
\begin{equation*}
\begin{aligned}
    |L_j| \leq& \sum_{i = 1 \vee j}^n \min\Big\{\big| \mathbb{E}[\psi_{M}(X_i)|\mathcal{F}_j] - \mathbb{E}[\psi_M(X_i)|\mathcal{F}_{j-1}] \big|, 2M\Big\}\\
    \leq& \sum_{i = 1 \vee j}^n \min\Big\{\mathbb{E}\big[|\psi_{M}(X_i) - \psi_M(X_{i,\{j\}})|\big|\mathcal{F}_j\big], 2M\Big\}\\
    \leq& \sum_{i = 1 \vee j}^n \min\Big\{\mathbb{E}\big[|X_i - X_{i,\{j\}}|\big|\mathcal{F}_j\big], 2M\Big\}\\
    =& 2M\sum_{i = 1 \vee j}^n\mathbbm{1}\Big\{\mathbb{E}\big[|a_{i}(i-j)||\epsilon_i - \epsilon_{i}^{\prime}|\big|\mathcal{F}_j\big] \geq 2M\Big\}\\
    &+ \sum_{i = 1 \vee j}^n\mathbb{E}\big[|a_{i}(i-j)||\epsilon_i - \epsilon_{i}^{\prime}|\big|\mathcal{F}_j\big]\mathbbm{1}\Big\{\mathbb{E}\big[|a_{i}(i-j)||\epsilon_i - \epsilon_{i}^{\prime}|\big|\mathcal{F}_j\big] < 2M\Big\}\\
    =& I_j + II_j,
\end{aligned}
\end{equation*}
where the first inequality follows from the triangle inequality, the second inequality follows from Jensen's inequality and the fact that $$\mathbb{E}[\psi_M(X_{i})|\mathcal{F}_{j-1}] = \mathbb{E}[\psi_M(X_{i,\{j\}})|\mathcal{F}_{j-1}] = \mathbb{E}[\psi_M(X_{i,\{j\}})|\mathcal{F}_{j}],$$
and the third inequality follows from the Lipschitz continuity of $\psi_M(\cdot)$.
Further, we have that
\begin{align*}
    \mathbb{E}\big[|L_j|^k\big|\mathcal{F}_{j-1}\big] \leq& \mathbb{E}\big[(I_j + II_j)^k\big|\mathcal{F}_{j-1}\big]\\
    \leq& 2^{k-1}\mathbb{E}\big[I_j^k\big|\mathcal{F}_{j-1}\big] +  2^{k-1}\mathbb{E}\big[II_j^k\big|\mathcal{F}_{j-1}\big],
\end{align*}
where the second inequality follows from that $(a+b)^k \leq 2^{k-1}(a^k + b^k)$ for any $a,b \in \mathbb{R}$ and $k \geq 2$.
\\
For $I_j$, we have that
\begin{align}\label{eq:I_ub}
    \big(\mathbb{E}[I_j^k|\mathcal{F}_{j-1}]\big)^{1/k} 
    \leq& 2M\sum_{i = 1 \vee j}^n\Big\{\mathbb{P}_{|\mathcal{F}_{j-1}}\big(\mathbb{E}\big[|a_{i}(i-j)||\epsilon_i - \epsilon_{i}^{\prime}|\big|\mathcal{F}_j\big] \geq 2M\big)\Big\}^{1/k}\nonumber\\
    \leq& (2M)^{1-2/k}\sum_{i = 1 \vee j}^n |a_{i}(i-j)|^{2/k}\Big\{\mathbb{E}\big[\big(\mathbb{E}\big[|\epsilon_i - \epsilon_{i}^{\prime}|\big|\mathcal{F}_j\big]\big)^2\big|\mathcal{F}_{j-1}\big]\Big\}^{1/k}\nonumber\\
    \leq& (2M)^{1-2/k}\sum_{i = 1 \vee j}^n |a_{i}(i-j)|^{2/k}\Big\{\mathbb{E}\big[|\epsilon_i - \epsilon_{i}^{\prime}|^2\big|\mathcal{F}_{j-1}\big]\Big\}^{1/k}\nonumber\\
    =& (2M)^{1-2/k}\sum_{i = 1 \vee j}^n |a_{i}(i-j)|^{2/k}\Big\{\mathbb{E}[|\epsilon_i - \epsilon_{i}^{\prime}|^2]\Big\}^{1/k}\nonumber\\
    =& (2M)^{1-2/k}(2\sigma_{\epsilon}^2)^{1/k}\sum_{i = 1 \vee j}^n|a_{i}(i-j)|^{2/k}\nonumber\\
    \leq& (2M)^{1-2/k}(2C_{Lin}^2\sigma_{\epsilon}^2)^{1/k} (1-\rho)^{2/k}\sum_{\ell = (1-j) \vee 0}^{\infty}\rho^{2\ell/k}\nonumber\\
    =& (2M)^{1-2/k}(2C_{Lin}^2\sigma_{\epsilon}^2)^{1/k} (1-\rho)^{2/k}\frac{\rho^{(2(1-j) \vee 0)/k}}{1-\rho^{2/k}},
\end{align}
where the first inequality follows from the triangle inequality, the second inequality follows from Markov's inequality, the third inequality follows from Jensen's inequality, and the fourth inequality follows from the assumption that $\sup_{i \in \mathbb{Z}}|a_{i}(\ell)|$ decays exponentially in $\ell$.
\\
For $II_j$, by the similar arguments as for $I_j$, we have that
\begin{align}\label{eq:II_ub}
    &\big(\mathbb{E}[II_j^k|\mathcal{F}_{j-1}]\big)^{1/k}\nonumber\\
    \leq& \sum_{i = 1 \vee j}^n\Big\{\mathbb{E}\Big[\Big(\mathbb{E}\big[|a_{i}(i-j)||\epsilon_i - \epsilon_{i}^{\prime}|\big|\mathcal{F}_j\big]\mathbbm{1}\Big\{\mathbb{E}\big[|a_{i}(i-j)||\epsilon_i - \epsilon_{i}^{\prime}|\big|\mathcal{F}_j\big] < 2M\Big\}\Big)^{k}\Big|\mathcal{F}_{j-1}\Big]\Big\}^{1/k}\nonumber\\
    \leq& (2M)^{1-2/k}\sum_{i = 1 \vee j}^n\Big\{\mathbb{E}\Big[\Big(\mathbb{E}\big[|a_{i}(i-j)||\epsilon_i - \epsilon_{i}^{\prime}|\big|\mathcal{F}_j\big]\Big)^{2}\Big|\mathcal{F}_{j-1}\Big]\Big\}^{1/k}\nonumber\\
    \leq& (2M)^{1-2/k}(2C_{Lin}^2\sigma_{\epsilon}^2)^{1/k} (1-\rho)^{2/k}\frac{\rho^{(2(1-j) \vee 0)/k}}{1-\rho^{2/k}}.
\end{align}
Then \eqref{eq:I_ub} and \eqref{eq:II_ub} lead to
\begin{equation}\label{eq:L_ub}
\begin{aligned}
    \mathbb{E}\big[|L_j|^k\big|\mathcal{F}_{j-1} \big] \leq& 2^k(2M)^{k-2}(2C_{Lin}^2\sigma_{\epsilon}^2)(1-\rho)^{2}\frac{\rho^{(2(1-j) \vee 0)}}{(1-\rho^{2/k})^k}\\
    \leq& k^k\rho^{-2}\big(\log(1/\rho)\big)^{-k}(2M)^{k-2}(1-\rho)^{2}\rho^{(2(1-j) \vee 0)}(2C_{Lin}^2\sigma_{\epsilon}^2)\\
    \leq& k!(2\pi)^{-1/2}(2M\rho)^{-2}\big(\log(1/\rho)\big)^{-k}(2Me)^{k}(1-\rho)^{2}\rho^{(2(1-j) \vee 0)}(2C_{Lin}^2\sigma_{\epsilon}^2),
\end{aligned}
\end{equation}
where the second inequality follows from that
$1-x \geq -x\log x$ for any $x \in (0,1)$, and the last inequality follow from Stirling's formula.
Plugging \eqref{eq:L_ub} into \eqref{eq:exp_conditional}, we have for $0 < \lambda <  (\log(1/\rho))(2Me)^{-1}$,
\begin{align}\label{eq:exp_conditional_ub}
    &\mathbb{E}\big[\exp\{\lambda L_j \}|\mathcal{F}_{j-1}\big]
    = 1 + \sum_{k = 2}^{\infty}\frac{1}{k!}\mathbb{E}\big[ \lambda^k L_j^k| \mathcal{F}_{j-1} \big]\nonumber\\
    \leq& 1 + (2\pi)^{-1/2}(2M\rho)^{-2}(1-\rho)^{2}\rho^{(2(1-j) \vee 0)}(2C_{Lin}^2\sigma_{\epsilon}^2)\sum_{k = 2}^{\infty}\big(2Me\lambda(\log(1/\rho))^{-1}\big)^{k}\nonumber\\
    =& 1 + (2\pi)^{-1/2}\rho^{-2}(1-\rho)^{2}\rho^{(2(1-j) \vee 0)}(2C_{Lin}^2\sigma_{\epsilon}^2) \frac{\big(e\lambda(\log(1/\rho))^{-1}\big)^{2}}{1 - 2Me\lambda(\log(1/\rho))^{-1}}\nonumber\\
    \leq& \exp\bigg\{ (2\pi)^{-1/2}\rho^{-2}(1-\rho)^{2}(2C_{Lin}^2\sigma_{\epsilon}^2)\rho^{(2(1-j) \vee 0)} \frac{\big(e\lambda(\log(1/\rho))^{-1}\big)^{2}}{1 - 2Me\lambda(\log(1/\rho))^{-1}} \bigg\},
\end{align}
where the last inequality follows from that $1 + x \leq \exp(x)$ for any $x \in \mathbb{R}$.
Plugging \eqref{eq:exp_conditional_ub} into \eqref{eq:tail_prob}, we have that
\begin{equation}\label{eq:tail_prob_ub}
\begin{aligned}
    &\mathbb{P}\Big( \sum_{i = 1}^nX_i \geq x \Big)= e^{-\lambda x}\mathbb{E}\Big[\prod_{j = -\infty}^{n}\mathbb{E}\big[\exp\{\lambda L_j\}|\mathcal{F}_{j-1}\big] \Big]\\
    \leq& e^{-\lambda x}\exp\bigg\{\frac{(1-\rho)^{2}(2C_{Lin}^2\sigma_{\epsilon}^2)}{(2\pi)^{1/2}\rho^{2}} \frac{\big(e\lambda(\log(1/\rho))^{-1}\big)^{2}}{1 - 2Me\lambda(\log(1/\rho))^{-1}}\Big(n+ \frac{\rho^2}{1-\rho^2}\Big)  \bigg\}\\
    =& e^{-\lambda x}\exp\bigg\{C_1(2C_{Lin}^2\sigma_{\epsilon}^2) \frac{C_2^2\lambda^{2}}{1 - 2C_2M\lambda }\Big(n+ \frac{\rho^2}{1-\rho^2}\Big)  \bigg\}\\
    =& \exp\bigg\{-\frac{x^2}{8C_1C_2^2C_{Lin}^2\sigma_{\epsilon}^2(n + \rho^2/(1-\rho^2)) + 4C_2Mx}\bigg\},
\end{aligned}
\end{equation}
where we let $C_1 = \frac{(1-\rho)^{2}}{(2\pi)^{1/2}\rho^{2}}$, $C_2 = e(\log(1/\rho))^{-1}$ and the last equality follows by letting
\begin{equation*}
    \lambda = \frac{x}{4C_1C_2^2C_{Lin}^2\sigma_{\epsilon}^2(n + \rho^2/(1-\rho^2)) + 2C_2Mx}.
\end{equation*}
\end{proof}

\clearpage
\bibliographystyle{imsart-number}
\bibliography{Ref}

\end{document}